\documentclass[sigconf,nonacm]{acmart}

\usepackage{comment}
\usepackage{graphicx}

\usepackage{amsmath}
\usepackage{bm}
\usepackage{setspace}
\usepackage{appendix}

\usepackage{xcolor}
\usepackage{enumitem}
\usepackage{stackengine}
\usepackage{setspace}
\usepackage{graphicx}
\usepackage{algorithm,algpseudocode,float}
\usepackage{multirow,makecell}
\usepackage{dblfloatfix}
\usepackage{url}

\newcommand{\sys}{\mathsf{Bolt\textnormal{-}Dumbo\textnormal{\ }Transformer}}
\newcommand{\fACS}{\mathsf{f\textnormal{-}ACS}}

\newcommand{\Bolt}{\mathsf{Bolt}}
\newcommand{\Dumbo}{\mathsf{Dumbo}}
\newcommand{\Trans}{\mathsf{Transformer}}
\newcommand{\BDT}{\mathsf{BDT}}
\newcommand{\nwabc}{$\mathsf{nw\textnormal{-}ABC}$}
\newcommand{\Pessimistic}{\mathsf{Pessimistic}}

\newtheorem{claim}{Claim}

\newcommand{\blocks}{\mathsf{blocks}}
\newcommand{\block}{\mathsf{block}}
\newcommand{\prf}{\mathsf{Proof}}

\newcommand{\tx}{\mathsf{tx}}
\newcommand{\out}{\mathsf{log}}

\newcommand{\payload}{\mathsf{TXs}}

\newcommand{\dumbo}{\mathsf{Dumbo}}
\newcommand{\trans}{\mathsf{Transformer}}
\newcommand{\txs}{\mathsf{txs}}

\newcommand{\abandon}{\mathsf{abandon}}
\newcommand{\maxview}{\mathsf{maxPace}}
\newcommand{\pace}{{p}}
\newcommand{\id}{\mathsf{id}}

\newcommand{\tobc}{\mathsf{Bolt}}
\newcommand{\tcvba}{\mathsf{tcv\textnormal{-}BA}}

\newcommand{\acs}{\mathsf{ACS}}

\newcommand{\gethelp}{\mathsf{CallHelp}}
\newcommand{\help}{\mathsf{Help}}
\newcommand{\buf}{\mathsf{buf}}

\newcommand{\gap}{\mathsf{gap}}
\newcommand{\tip}{\mathsf{tip}}

\newcommand{\TPKE}{\mathsf{TPKE}}

\newcommand{\Enc}{\mathsf{Enc}}
\newcommand{\DecShare}{\mathsf{DecShare}}
\newcommand{\Dec}{\mathsf{Dec}}

\newcommand{\notarized}{\mathsf{pending}}

\newcommand{\bigO}{\mathcal{O}}
\newcommand{\node}{\mathcal{P}}
\newcommand{\hash}{\mathcal{H}}

\newcommand{\TSIG}{\mathsf{TSIG}}

\newcommand{\ABA}{\mathsf{ABBA}}
\newcommand{\RBC}{\mathsf{RBC}}

\newcommand{\ACS}{\mathsf{ACS}}

\newcommand{\Combine}{\mathsf{Combine}}

\newcommand{\SignShare}{\mathsf{SignShare}}
\newcommand{\VrfyShare}{\mathsf{VrfyShare}}
\newcommand{\Vrfy}{\mathsf{Vrfy}}
\newcommand{\N}{\mathbb{N}}

\newcommand{\est}{\mathsf{est}}

\newcommand{\coin}{\mathsf{Coin}}
\newcommand{\getcoin}{\mathsf{GetCoin}}

\newcommand{\view}{\textsc{PaceSync}}

\newcommand{\hlp}{\textsc{Help}}
\newcommand{\ghlp}{\textsc{CallHelp}}
\newcommand{\val}{\textsc{Value}}

\newcommand{\proposal}{\textsc{proposal}}

\newcommand{\vote}{\textsc{Vote}}

\newcommand{\prbc}{\mathsf{PRBC}}
\newcommand{\rbc}{\mathsf{RBC}}
\newcommand{\Verify}{\mathsf{Verify}}
\newcommand{\LogVerify}{\mathsf{Bolt.verify}}
\newcommand{\Extract}{\mathsf{Bolt.extract}}

\newcommand{\ignore}[1]{}

\newcommand{\rev}[1]{\textcolor{black}{#1}}

\algrenewcommand\algorithmicindent{1.0em}
\algrenewcommand\algorithmiccomment[1]{\hfill\hskip3em{\color{orange} $\rhd$ #1}}

\algdef{SE}[PROCEDURE]{Procedure}{EndProcedure}{}{\algorithmicend\ }
\algtext*{Procedure}
\algtext*{EndProcedure}

\setcopyright{none} 
\begin{document}
	\fancyhead{}
\title{Bolt-Dumbo Transformer:   Asynchronous  Consensus   As Fast As the Pipelined BFT} 

\author{Yuan Lu} 
\authornote{Authors are listed alphabetically. Yuan Lu and Zhenliang Lu contributed equally. A preliminary version of this paper will appear at {\sf ACM CCS 2022}.}
\affiliation{\institution{Institute of Software\and Chinese Academy of Sciences}}
\email{luyuan@iscas.ac.cn}

\author{Zhenliang Lu}\authornotemark[1]
\affiliation{\institution{School of Computer Science \and The University of Sydney}}
\email{zhlu9620@uni.sydney.edu.au}

\author{Qiang Tang}\authornotemark[1]
\affiliation{\institution{School of Computer Science \and The University of Sydney}}
\email{qiang.tang@sydney.edu.au}

\begin{abstract}
An urgent demand of deploying BFT consensus (e.g., atomic broadcast) over the Internet is   raised for implementing  (permissioned) blockchain services. The deterministic (partial) synchronous protocols can be simple and fast in good network conditions, but are subject to denial-of-service (or even safety vulnerability) when synchrony assumption fails.  Asynchronous protocols, on the contrary, are robust against the adversarial network, but are substantially more complicated and slower for the inherent use of randomness.

Facing the issues, optimistic asynchronous atomic broadcast (Kursawe-Shoup, 2002; Ramasamy-Cachin, 2005) was proposed to improve the normal-case performance of the slow asynchronous consensus. They run a deterministic fastlane if the network condition remains good, and can   fall back to a fully asynchronous protocol via a pace-synchronization mechanism (analog to view-change with asynchronous securities) if the fastlane fails. Unfortunately,  existing pace-synchronization directly uses a heavy tool of asynchronous multi-valued validated Byzantine agreement (MVBA). When such fallback frequently occurs in the fluctuating wide-area network setting, the benefits of adding fastlane can be eliminated.

We present $\sys$ (BDT), a generic framework for \emph{practical} optimistic asynchronous atomic broadcast. At the core of BDT, we set forth a new fastlane abstraction that is simple and fast, while  preparing honest parties to gracefully face potential fastlane failures caused by malicious leader or bad network. This enables a highly efficient pace-synchronization to handle fallback. The resulting design reduces a cumbersome MVBA  to a variant of the conceptually simplest  \emph{binary} agreement only. Besides  detailed security analyses, we also give concrete instantiations of our framework and implement them. Extensive experiments  demonstrate that BDT can enjoy both the low latency of deterministic protocols (e.g. 2-chain version of HotStuff) and  the   robustness of state-of-the-art asynchronous protocols in practice.

\end{abstract}



\maketitle


\section{Introduction}
\label{sec:intro}



The explosive popularity of decentralization \cite{bitcoin, buterin2014next}   
creates an unprecedented demand of deploying robust Byzantine fault tolerant (BFT) consensus on the global Internet.  
These consensus protocols were conventionally abstracted as BFT atomic broadcast ($\mathsf{ABC}$) to replicate an ever-growing linearized log of transactions
among $n$ parties \cite{cachin2017blockchain}.
Informally, $\mathsf{ABC}$ ensures {\em Safety} and {\em Liveness}   despite  that an adversary controls the communication network (e.g.,  delay messages) and corrupt some participating parties (e.g., $n/3$).  Safety ensures all honest parties to eventually  output the same log of transactions, and liveness guarantees that any transaction inputted by some honest party eventually appears in honest parties' logs.


\ignore{
\begin{itemize}
	\item {\em Safety}.  The honest parties can  output the same linearly-ordered log of transactions;
	\item {\em Liveness}. Any transaction inputted by some honest party can  appear in each honest party's output log eventually; sometimes, one might consider a higher efficient liveness that the transactions shall  output reasonably quickly \cite{cachin01,ittai19,pass2017sleepy}, e.g., in a polynomial amount of rounds.\footnote{It can even be necessary to require output  within a polynomial number of rounds (in the security/system parameters), if using computationally-secure cryptographic primitives such as digital signature and hash function.}
\end{itemize}
}

\medskip
\noindent
{\bf A desideratum for robust BFT in the absence of   synchrony}. The dynamic nature of   Internet poses new fundamental challenges for implementing  secure yet still highly efficient BFT consensus protocols. 
Traditionally, most practical BFT   protocols   were studied for the in-house scenarios where participating parties are geographically close and well connected. 
Unsurprisingly, their securities rely on some form of assumptions about the network conditions. For example,  classic synchrony assumption needs all messages to deliver within a known delay, and its weaker variant called partial synchrony \cite{partialsync} (a.k.a. eventual synchrony) 
assumes that after an unknown global stabilization time (GST), all messages can be delivered synchronously.
%
%
%
Unfortunately, these    synchrony assumptions may not always hold in the wide-area network (WAN), because of fluctuating bandwidth, unreliable links, substantial delays, and even network attacks.
What's worse, in an   asynchronous network \cite{attiya2004distributed},
such (partially)  synchronous protocols \cite{pbft,tendermint,chan2020streamlet,sbft,bft-smart,thunderella,amir2010prime,700,700aublin,veronese2009spin,aublin2013rbft} will grind to a halt (i.e., suffers from the inherent {\em loss of liveness} \cite{FLP85,honeybadger}), and Bitcoin might even have a {\em safety} issue of potential double-spending \cite{saad2021revisiting} when the adversary can arbitrarily schedule message deliveries.
That said, when the network is adversarial,  
relying on synchrony could lead to fatal vulnerabilities.
 
It becomes a {\em sine qua non}  to consider robust BFT consensus that can thrive in   the unstable or even adversarial Internet  for mission-critical applications (e.g.,  financial  services or cyber-physical systems).
Noticeably, the class of fully asynchronous protocols  \cite{cachin01,honeybadger,beat,guo2020dumbo,yang2021dispersedledger}  can ensure safety and liveness simultaneously without any form of network synchrony, and thus become the  arguably most  robust candidates for implementing mission-critical applications.
\medskip
\noindent
{\bf Fully asynchronous BFT? Robustness with a high price!}
Nevertheless, the higher security assurance of asynchronous BFT consensus does not come for free: the seminal FLP ``impossibility'' \cite{FLP85}  states that no {\em deterministic} protocol can ensure both safety and liveness in an asynchronous network.
So    asynchronous $\mathsf{ABC}$   must run {\em randomized} subroutines to circumvent the ``impossibility'', which already hints its complexity. 
%
%
%
%
%
%
%
%
Indeed, few   asynchronous   protocols have been   deployed in practice during the past decades due to large complexities, 
until the recent  $\mathsf{HoneyBadgerBFT}$  \cite{honeybadger} and  $\dumbo$ protocols \cite{guo2020dumbo} (and very recent their improved variants \cite{yang2021dispersedledger,guo2022speeding}) provide novel paths to {\em practical} asynchronous ABC in terms of realizing optimal linear communication cost per output transaction. 

%

Despite those  recent progresses, 
the actual performance of state-of-the-art randomized asynchronous consensus  is still far worse than the  deterministic (partial) synchronous ones (e.g., $\mathsf{HotStuff}$  \cite{yin2018hotstuff-full}\footnote{Remark that \cite{yin2018hotstuff-full} gave  a 3-chain $\mathsf{HotStuff}$ protocol along with a 2-chain variant. Throughout the paper, we let $\mathsf{HotStuff}$ refer to the 2-chain version (with minor difference to fix the view-change issue) for the lower latency of the 2-chain version.}), especially regarding the critical latency metric. 
For example, in the same WAN deployment environments consisting of $n$=16, 64, 100 Amazon  EC2 instances across the globe,
$\mathsf{HotStuff}$ is dozens of times faster than the state of the art $\dumbo$ protocol \cite{guo2020dumbo}. 
Even worse,
the inferior latency performance of  asynchronous protocols 
 stems from the fact that all parties  generate some common randomness (e.g.,  ``common coin'' \cite{canetti1993fast,cachin00}),
and multiple repetitions are necessary to ensure the parties to coincidentally output with an overwhelming probability.
Even if in one of the fastest existing asynchronous protocols such as $\dumbo$ \cite{guo2020dumbo} and its improved version $\mathsf{Speeding}$-$\dumbo$ \cite{guo2022speeding}, they still cost about a dozen of   rounds  on average.
While for their (partial) synchronous counterparts,
only a very small number of rounds are required in the optimistic cases when the underlying  communication network is   luckily synchronous \cite{abraham2021good}, e.g., 5 in the two-chain $\mathsf{HotStuff}$ and 3 in PBFT. 



The above issues correspond to a fundamental ``dilemma'' lying in the  design space of BFT consensus protocols suitable for the open Internet: the cutting-edge (partially) synchronous deterministic protocols can optimistically work very fast,  but lack  liveness guarantee in adversarial networks;
on the contrary, the fully asynchronous randomized protocols  are robust even in malicious networks, but  suffer from poor latency performance in the normal case. 
Facing that, a natural question arises:  

\begin{center}
	{\em Can we    design a BFT consensus  achieving  the best of both synchronous and asynchronous paradigms, such that it (i)  is   ``as fast as'' the state-of-the-art    deterministic  BFT consensus on the normal Internet with fluctuations and (ii)    performs nearly same to the existing performant    asynchronous BFT consensus  even if in a worst-case asynchronous network?
	}
\end{center}
 
%
%
%


\subsection{Our contributions}

We answer  the aforementioned question affirmatively, by presenting the first {\em practical} and {\em generic} framework for   optimistic  asynchronous  atomic broadcast  called $\sys$ (or $\BDT$ for short). Here, optimistic  asynchronous  atomic broadcast \cite{kursawe05,cachin05,gelashvili2021jolteon} refers to an asynchronous consensus that has a deterministic fastlane that might luckily progress in the benign network environment (optimistic case) without randomized execution.

\begin{figure}[h] 
	\vspace{-0.2cm}
	\begin{center}
		\includegraphics[width=8.3cm]{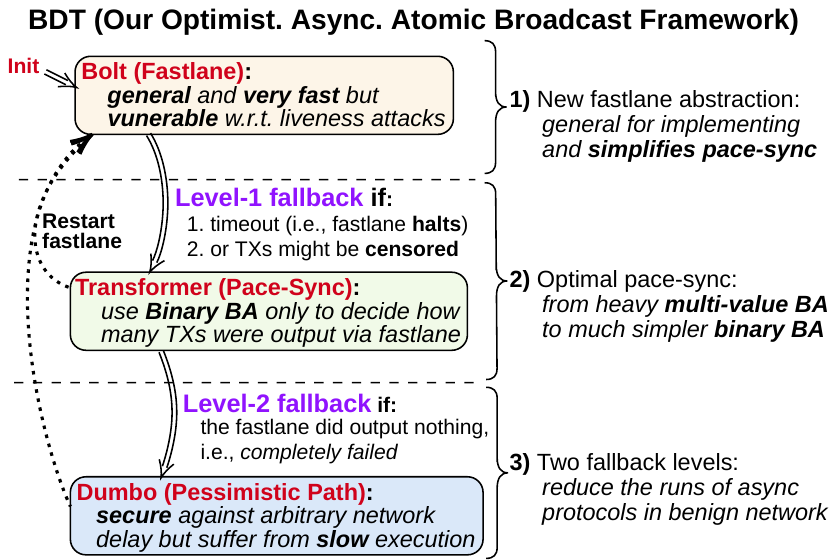}
		\vspace{-0.2cm}
		\caption{The overview of $\sys$.}	\label{fig:BDT}
	\end{center}
	\vspace{-0.3cm}
\end{figure}


At a very high-level, $\BDT$   has three phases as shown in Fig. \ref{fig:BDT}:
\begin{itemize}[leftmargin=0.7cm]
	\item {\em Fastlane} (nickname $\mathsf{Bolt}$): It initially runs a deterministic protocol as  fastlane to quickly progress in the optimistic case that synchrony assumption holds.
	\item {\em Pace-synchronization} (called pace-sync for short or nickname $\trans$): If the fastlane fails to progress in time, a  fast {pace-synchronization} mechanism (analog to view-change with asynchronous securities)  is triggered to make all honest parties agree on from where and how to restart (directly restart $\Bolt$ or enter pessimistic path).
	\item {\em Pessimistic path} (nickname $\Dumbo$ \footnote{Here we call the pessimistic path $\Dumbo$ and call the concrete ABC protocol in \cite{guo2020dumbo} Dumbo-BFT. The latter could be an instantiation for the former (as in our experiments), but other better ABC protocol could also instantiate the former pessimistic path.}): In case the fastlane was completely failed to make progress, the honest parties enter a {pessimistic path} of asynchronous BFT to ensure liveness even in the worst case. 
\end{itemize}

{$\BDT$ is featured with guaranteed {\em liveness} and safety even in a hostile asynchronous network} with optimal tolerance against $n/3$ byzantine parties.
More importantly, as depicted in Fig. \ref{fig:changing-net}, it indeed realizes the best of its both paths, i.e., it is as fast as deterministic protocols (such as the state-of-the-art pipeline BFT protocols, e.g.,   2-chain   $\mathsf{HotStuff}$) in ``good'' synchronous periods; and also as robust as asynchronous protocols in the ``bad'' network. 

\ignore{
\begin{itemize}
	\item In the normal cases where  the network is mostly synchronous  and only encounters short-term network fluctuations,
	$\BDT$ can be as fast as the state-of-the-art pipeline BFT protocols (e.g.,   2-chain   $\mathsf{HotStuff}$) and   performs much faster than existing asynchronous   protocols.
	\item In the worse cases where    long-lasting   asynchronous periods exist, $\BDT$ can   perform as robustly as the underlying asynchronous protocol to closely track the actual network bandwidth and delay, despite the synchronous protocols might suffer from denial-of-service in the environment.
\end{itemize}
}

\begin{figure}[h] 
	\vspace{-0.2cm}
	\begin{center}
		\includegraphics[width=8cm]{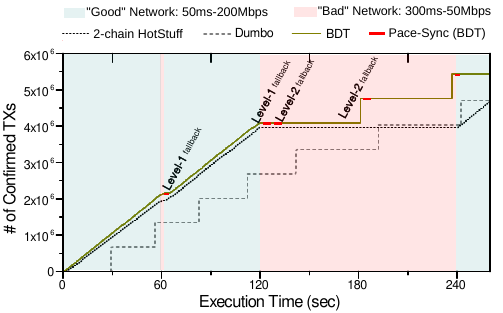}\\
		\vspace{-0.3cm}
		\caption{Simulated executions of $\mathsf{BDT}$, 2-chain $\mathsf{HotStuff}$ and $\dumbo$-{\sf BFT} under fluctuating network ($n$=64). Bad network happens twice: one lasts 2 seconds and one lasts 120 seconds. See Section \ref{evaluation} for the details on simulation setup.}	\label{fig:changing-net}
	\end{center}
	\vspace{-0.3cm}
\end{figure}

\smallskip
\noindent
{\bf Technical overview}.
Different from pioneering studies  \cite{kursawe05,cachin05,700,700aublin} that only demonstrated theoretic feasibility and had questionable practicability because of complex and slow {\em asynchronous} pace-synchronization (cf. Section \ref{sec:prior} for detailed discussions on their efficiency bottleneck),
$\BDT$ makes  several technical contributions   to harvest the best of both paths in practice. In greater detail,



	\smallskip
		\noindent\underline{\smash{\em A new fastlane abstraction \rev{better prepared for failures}}}.
%
\rev{
To simplify the   complicated pace-sync, 
we propose a new fastlane abstraction of notarizable weak atomic broadcast ($\mathsf{nw\textnormal{-}ABC}$ for short)
to prepare    honest parties in a graceful condition  when facing potential fastlane failures.} 
Notably, $\mathsf{nw\textnormal{-}ABC}$ realizes $\mathsf{ABC}$   in the optimistic case, and only ensures ``notarizability'' otherwise:  any output block is with a quorum proof to  attest that sufficient honest parties have  received a previous block (along with valid  proof). Such an $\mathsf{nw\textnormal{-}ABC}$ can be easily constructed to be very fast, e.g.,   from a sequence of simple (provable) multicasts;
and more importantly, the notarizability (as we will carefully analyze) guarantees that any two honest parties will be at neighboring blocks when entering pace-sync, \rev{thus enables us to leverage simpler binary agreement to replace the cumbersome full-fledged asynchronous atomic broadcast or  multi-value agreement  used in prior art \cite{700aublin,700,kursawe05,cachin05}. }

\medskip
	\noindent\underline{\smash{\em  Cheapest possible pace-synchronization}}. More importantly, with the preparation of $\mathsf{nw\textnormal{-}ABC}$, 
	$\Trans$
 reduces pace-sync to a problem that we call {\em two-consecutive-valued Byzantine agreement} ($\tcvba$), which is essentially an  asynchronous {\em binary} Byzantine agreement ($\ABA$). In contrast, prior art  \cite{kursawe05,cachin05} leveraged cumbersome multi-valued agreement ($\mathsf{MVBA}$) for pace-sync (cf. Sec. \ref{sec:prior} for a careful review).  
	%
	%
	  $\Trans$ thus improves the communication complexity of pace-sync   by an $\bigO(n)$ factor, and is essentially optimal for  pace-sync, because the pace-sync problem can be viewed as a version of asynchronous consensus, and $\ABA$ is the arguably simplest asynchronous consensus. 
 In practice, $\Trans$  attains a minimal overhead   similar to the fastlane latency. Further care is needed for invoking $\mathsf{nw\textnormal{-}ABC}$ to ensure the safety (see next section).


%

\medskip
\noindent\underline{\smash{\em Avoiding pessimistic path whenever we can}}. To further exploit  the benefits brought by   fast $\Trans$, we add a simple check after pace-sync to create two-level fallbacks: 
	if pace-sync reveals that  the fastlane still made some output, it immediately restarts another fastlane  without running the actual pessimistic path.
	%
	This is in contrast with previous works \cite{kursawe05,cachin05} where the slow pessimistic path will always   run after each pace-sync, which is often unnecessarily costly if there are only short-term network fluctuations. Remark that the  earlier studies cannot effectively adopt our two-level fallback tactic, because their heavy pace-sync  might   bring extra cost \rev{and it may even nullify the advantages of the fastlane} in case of frequent fallbacks.
	
	
	\medskip 
	\noindent\underline{\smash{\em Generic framework enabling flexible instantiations}}.	$\mathsf{BDT}$  is generic, as it  enables  flexible choices of the underlying building blocks for all three phases. For example, we present two exemplary fastlane instantiations, resulting in two $\mathsf{BDT}$ implementations that   favor  latency and throughput, respectively, so one can instantiate $\mathsf{BDT}$ according to the actual application scenarios. Also,  $\Trans$ can be constructed around any asynchronous binary agreement, thus having the potential of using any more efficient $\ABA$ to further reduce the fallback overhead  \rev{(e.g., by adopting the  recent progress from Crain \cite{crain2020two}, Das et al. \cite{das2021practical} and Zhang et al. \cite{zhang2022pace})}. Similarly,  though currently we use  $\mathsf{Dumbo}$-BFT as the pessimistic path, this can be   replaced by more efficient recent designs \cite{guo2022speeding}.

 \medskip
\noindent
{\bf Extensive  evaluations}.
To demonstrate the practical performance of $\BDT$,
we  implement the  framework using $\Dumbo$-{\sf BFT} \cite{guo2020dumbo} as the exemplary   pessimistic path. 
%
%
We compare   two typical $\BDT$ implementations to   $\dumbo$ and $\mathsf{HotStuff}$, and conduct extensive experiments in   real-world/simulated environments. 

We first deploy all protocols  in the same real-world WAN environment consisting of up to 100 Amazon EC2 c5.large instances     across the globe. 
Some highlighting   experimental results could be found in Table \ref{table:summary}, which are: (i) 
the $\BDT$ implementation based on sequential multicasts can  attain a basic latency about only 0.44 second (i.e., nearly same to 2-chain HotStuff's and less than 3\% of $\dumbo$'s latency),
even if we intentionally raise frequent pace synchronizations after every 50 optimistic blocks; 
(ii) in the worst case that we intentionally make fastlane to always fail,
the throughput of $\BDT$ remains 90\% of $\dumbo$ BFT's and is close to 20,000 transaction per second, 
even if we do let $\BDT$ wait as long as 2.5 seconds to timeout before invoking pace-sync.
Those demonstrate  $\BDT$  can be ``as fast as'' the deterministic protocols in normal cases and can  maintain  robust performance in  the worst case.

To understand  more scenarios in-between, we then conduct  evaluations in a controlled test environment that can simulate fluctuating network (i.e., switching between ``good'' and ``bad'' networks).
In the simulated good network   (i.e., 50 ms packet delay and 200 Mbps peer-to-peer link), $\BDT$ is almost as fast as 2-chain HotStuff; while in the simulated bad network (i.e., 300 ms packet delay and 50 Mbps peer-to-peer link), $\BDT$ can closely track  the performance of underlying pessimistic asynchronous protocol, though HotStuff might grind to a halt due to inappropriately chosen timeout parameter.

See Section \ref{evaluation} for  more detailed experiment  setup and   results.

\begin{table}[htp]
	\vspace{-0.15cm}
	\captionsetup{font={normalsize}}
	\caption{$\mathsf{BDT}$, 2-chain $\mathsf{HotStuff}$ and $\dumbo$-{\sf BFT}  running over {\bf 100}   EC2 c5.large servers across {\bf 16} regions in {\bf 5} continents}
	\vspace{-0.45cm}
	\label{tab:experiement-highlight}
	\resizebox{0.495\textwidth}{!}{%
		
		\begin{tabular}{c|c|c|c|c|c}
			
			 \hline\rule{0pt}{18pt} 
			\multirow{2}{*}{}                                                  & \multicolumn{3}{c|}{\textbf{\begin{tabular}[c]{@{}c@{}}Good-Case v.s. HotStuff\\ (fastlanes always complete)\end{tabular}}}                   & \multicolumn{2}{c}{\textbf{\begin{tabular}[c]{@{}c@{}}Worst-Case  v.s. Dumbo\\ (fastlanes always fail)\end{tabular}}} \\ \cline{2-6} \rule{0pt}{15pt}
			& \begin{tabular}[c]{@{}c@{}}\hspace{-0.15cm}$\mathsf{BDT}$-$\mathsf{sCAST}$$^{**}$\end{tabular} & \begin{tabular}[c]{@{}c@{}}\hspace{-0.15cm}$\mathsf{BDT}$-$\mathsf{sRBC}$$^\dagger$\end{tabular} &\hspace{-0.15cm} $\mathsf{HotStuff}$$^\S$ & \begin{tabular}[c]{@{}c@{}}\hspace{-0.15cm}$\mathsf{BDT}$-$\mathsf{Timeout}$$^\ddagger$\end{tabular}                       & \hspace{-0.15cm}$\mathsf{DumboBFT}$                      \\ \hline\rule{0pt}{10pt}
			\hspace{-0.25cm}Basic latency    (sec)      & 0.44                                                     & 0.67                                                      & 0.42     & 21.95                                                                               & 16.36                      \\ \hline\rule{0pt}{10pt}
			\begin{tabular}[c]{@{}c@{}}\hspace{-0.25cm}Throughput   (tx/sec)$^{*}$ \end{tabular} & 9,253                                                     & 18,234                                                     & 10,805    & 18,806                                                                               & 21,242                      \\ \hline 
		\end{tabular}
	}%
	
	{
		\footnotesize
		\begin{itemize}[leftmargin=0.4cm]
			\item[*] Each transaction has 250 bytes to approximate the basic Bitcoin transaction.
			\item[**] $\mathsf{BDT}$-$\mathsf{sCAST}$ is $\mathsf{BDT}$ with using $\mathsf{Bolt}$-$\mathsf{sCAST}$, where $\mathsf{Bolt}$-$\mathsf{sCAST}$ is the   fastlane instantiation built from pipelined multicasts.
			\item[$^\dagger$] $\mathsf{BDT}$-$\mathsf{sRBC}$  is $\mathsf{BDT}$ with using $\mathsf{Bolt}$-$\mathsf{sRBC}$, where $\mathsf{Bolt}$-$\mathsf{sRBC}$ is the   fastlane instantiation built from sequential reliable broadcasts.
			\item[$^\ddagger$]  $\mathsf{BDT}$-Timeout idles for 2.5 sec in the  fastlane,  and then runs Pace-Sync+$\dumbo$.
			\item[$^\S$] Reasons why our evaluations for HotStuff   different from \cite{yin2018hotstuff-full}: (i)  we   tested  among 16 AWS regions instead of   one AWS region; (ii) we let the  implementation  to   agree  on 250-byte tx instead of 32-byte tx hash; (iii)   we use a single-process   network layer written in Python, differing from    multi-processing network layer in \cite{yin2018hotstuff-full}.
		\end{itemize}
	}

	\label{table:summary}
	\vspace{-0.2cm}
\end{table}

\section{Efficiency Bottleneck of Prior Art and Our Solution in a Nutshell}\label{sec:prior}


\noindent
{\bf Efficiency obstacles in prior art}.
As briefly mentioned, pioneering works of Kursawe-Shoup \cite{kursawe05} (KS02)  and a following improvement of Ramasamy-Cachin \cite{cachin05} (RC05)
initiated the study of {\em optimistic asynchronous atomic broadcast} by adding a deterministic fastlane to fully asynchronous atomic broadcast, and they adopted multi-valued {\em validated} Byzantine agreement ($\mathsf{MVBA}$) to facilitate fallback once the fastlane fails to progress.

Nevertheless, these prior studies are theoretical in   asynchronous networks, as they rely on heavy $\mathsf{MVBA}$ or even heavier full-fledged state-machine replication for fallback.
Serious efficiency hurdle remains in such cumbersome fallback, thus failing to harvest the best of   both paths in practice. 
Let us first overview the remaining   hurdles and  design challenges.

\smallskip
\noindent
\underline{\smash{\em Challenge  and effiency bottleneck lying in pace-synchronization}}.
As Fig. \ref{fig:ks02} illustrates, the fastlane of KS02 and RC05
  directly employs a sequence of some broadcast primitives (the output of which is also called a block for brevity). 
If a party does not receive a block within a period (defined by a timeout parameter), then it  requests   fallback by informing other parties about the index of the block that it just received. 
When the honest parties receive a sufficient number of fallback requests (e.g., $2f+1$ in the presence of $f$ faulty parties), 
they execute the {\em pace-synchronization} mechanism    to decide where to continue   the pessimistic path.

%
Since different honest parties may have different progress in the fastlane when they decide to fall back, e.g., some are now at block 5, some at block 10, thus {pace-synchronization} needs to ensure: (i) all honest parties can eventually enter the pessimistic path from the same block; and (ii) all the ``mess-ups'' (e.g., missing blocks) left by the fastlane can be properly handled. Both requirements should be satisfied in an asynchronous network! These requirements hint that all the parties may need to {\em agree} on a block index that is proposed by some honest party, 
otherwise they might decide to sync up to some blocks that were never delivered. Unfortunately, directly implementing such a functionality 
requires one-shot asynchronous (multi-valued) Byzantine agreement with strong validity
(that means the output must be from some honest party), which is   infeasible because of inherent exponential communication \cite{fitzi2003efficient}.

\begin{figure}[h]
	\vspace{-0.4cm}
	\begin{center}
		\includegraphics[width=8.5cm]{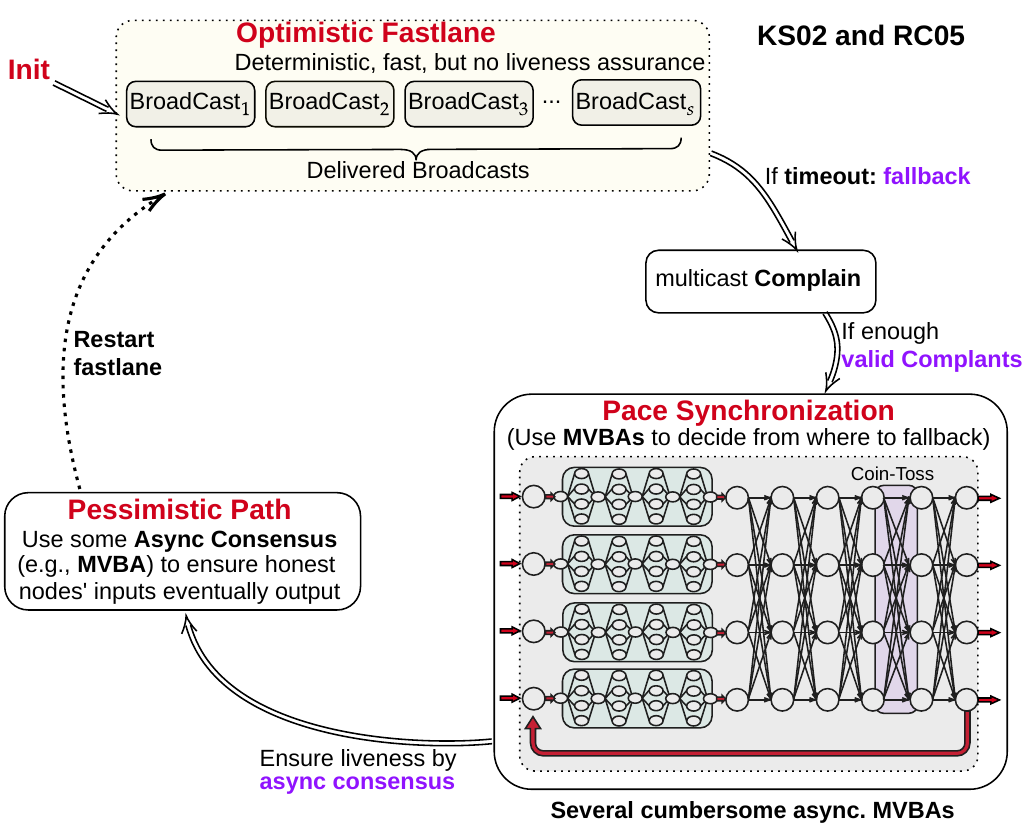}\\
		\vspace{-0.2cm}
		\caption{Execution flow of KS02 \cite{kursawe05} and RC05 \cite{cachin05}. Both rely on cumbersome asynchronous MVBA  to do pace-sync.}	\label{fig:ks02}
	\end{center}
	\vspace{-0.3cm}
\end{figure}


As depicted in Fig. \ref{fig:ks02}, both KS02 and RC05 smartly implement  {pace-synchronization} through asynchronous multi-valued {\em validated} Byzantine agreement ($\mathsf{MVBA}$) to get around the infeasible strong validity.  An $\mathsf{MVBA}$ is a weaker  and implementable form of asynchronous multi-valued Byzantine agreement, the output of which is allowed to be from a malicious party but has to satisfy a predefined predicate. Still, $\mathsf{MVBA}$ is a  cumbersome building block (and can even construct full-fledged asynchronous atomic broadcast directly \cite{cachin01}). 
What's worse, KS02 and RC05 invoke this heavy primitive for both pace-synchronization and pessimistic path, causing at least  $\bigO(n^3)$-bit communication and dozens of rounds. Although we may reduce the $\bigO(n^3)$ communication to $\bigO(n^2)$ by some very recent results (e.g.,  Dumbo-MVBA   \cite{lu2020dumbo}), however, they remain costly in practice due to a large number of extra execution rounds and additional computing costs (e.g., erasure encoding/decoding).

\smallskip
\noindent
\underline{\smash{\em \rev{Slow pace-sync remains in a more general framework  \cite{700aublin}}}}. \rev{Later, \\Aublin et al. \cite{700aublin} studied a   more general framework that is flexible to assemble optimistic fastlanes and  full-fledged BFT   protocols,
	as long as   the underlying modules all satisfy a defined $\mathsf{Abstract}$ functionality.
	To facilitate fallback  when the fastlane fails due to network asynchrony or corruptions, 
	\cite{700aublin} used a stronger version of $\mathsf{Abstract}$ variant with guaranteed liveness (called $\mathsf{Backup}$).
	$\mathsf{Backup}$ can guarantee all parties to output exact $k$ common transactions \cite{700aublin}, 
	so   it can handle fallback by first finishing pace-sync, then deciding some output transactions (i.e., running as the pessimistic path), and finally restarting the fastlane.
}

\rev{ Aublin et al. \cite{700aublin} also pointed out that $\mathsf{Backup}$ (with guaranteed progress) can be obtained from full-fledged BFT protocols. 
	For example, \cite{700aublin} gave exemplary $\mathsf{Backup}$ instantiations based on PBFT \cite{castro2002practical} and Aardvark \cite{clement2009making} in the partially synchronous setting. 
	This indicated another feasible way to implement asynchronous fallback, i.e., implement $\mathsf{Backup}$ by full-fledged  asynchronous BFT protocols.}

\rev{Unfortunately, when $\mathsf{Backup}$ is implemented via full-fledged  asynchronous BFT, it would be  as heavy as $\mathsf{MVBA}$ (or even heavier), since  most existing   performant asynchronous BFT protocols are either constructed from $\mathsf{MVBA}$ \cite{guo2020dumbo,guo2022speeding} or have implicit $\mathsf{MVBA}$ \cite{gelashvili2021jolteon}. 
	That said, though the framework presented in \cite{700aublin} is more general than KS02 and RC05, 
	it is not better than KS02 and RC05 with respect to the efficiency of pace-sync  (and thus has the same efficiency bottleneck lying in pace-sync).}

\rev{	
In contrast, we identify an extra simple property    (not covered by $\mathsf{Abstract}$ \cite{700aublin}), so the new fastlane abstraction (1) enables us to utilize a much simpler asynchronous pace-synchronization, and (2) still can be easily obtained with highly efficient instantiations.}



\smallskip
\noindent
\underline{\smash{\em Consequences of slow pace-synchronization.}}   
The inefficient pace-sync   severely harms the practical effectiveness of adding fastlane.
In particular, when  the network   may  fluctuate as in the real-world Internet, the pace-sync   phase  might be triggered frequently, 
and its high cost   might eliminate the benefits of adding fastlane.

\begin{figure}[h] 
	\vspace{-0.4cm}
	\begin{center}
		\includegraphics[width=8cm]{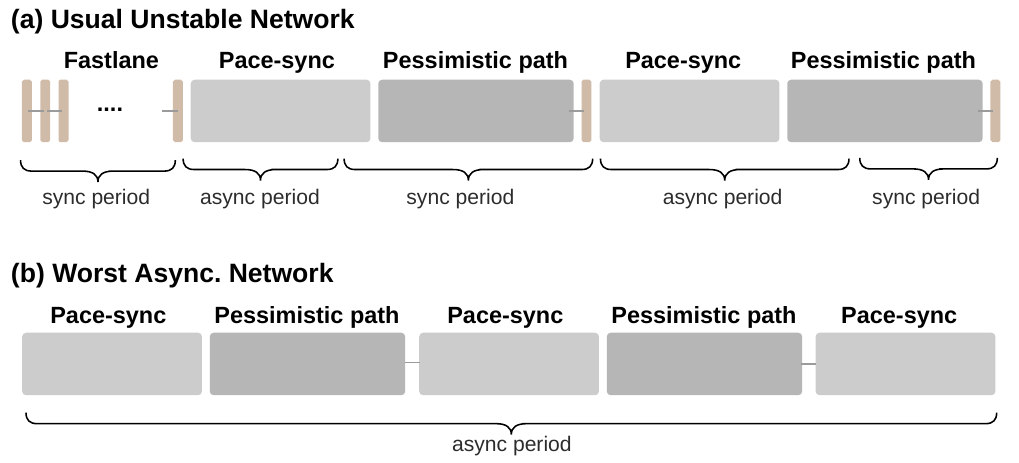}\\
		\vspace{-0.4cm}	\caption{Consequence of slow fallback in KS02/RC05 in fluctuating networks. The length of each phase denotes   latency.}	\label{fig:bottleneck}
	\end{center}
	\vspace{-0.4cm}
\end{figure}

To see the issue,  
consider the heavy pace-sync  of existing work   that is as slow as the asynchronous pessimistic path
and dozens of times slower than the fastlane.\footnote{Actual situation might be much worse in RC02 \cite{kursawe05} because several more MVBA invocations with much larger inputs are executed in the pace-sync.}
As Fig. \ref{fig:bottleneck} (a) exemplifies,
although the network stays  in good conditions for the majority of  time, the overall average latency of the protocol is still way larger than its  fastlane. One  slow fallback could ``waste'' the gain of dozens of optimistic blocks,  and it essentially renders the optimistic fastlane ineffective. In the extreme case shown in Fig. \ref{fig:bottleneck} (b), the fallback  is always triggered because the fastlane leaders are facing adaptive denial-of-service attack,  
it even doubles the cost of simply running the pessimistic asynchronous protocol alone.

%
%

It follows that in the wide-area Internet,    {\em inefficient pace synchronization} in previous theoretical protocols   likely eliminates the potential benefits of    optimistic fastlane, and thus their applicability is limited.  
So a fundamental practical challenge remains to minimize the overhead of pace-sync, such that we can harvest the  best of both paths  in optimistic asynchronous atomic broadcast.

\medskip
\noindent
{\bf Our Solution in a Nutshell}.
Now we walk through how we overcome the above challenge and reduce the   complex pace-sync  problem to only a variant of asynchronous binary agreement.


	\smallskip
	\noindent \underline{\smash{\em First ingredient: a  new abstraction of the fastlane}}.
	We put forth a new simple  fastlane abstraction called {\em notarizable weak atomic broadcast} (\nwabc, with nickname $\Bolt$). In the optimistic case, it performs as a full-fledged atomic broadcast  protocol and can output  a block per $\tau$ clock ticks. But if the synchrony assumption fails to hold, it won't have liveness nor exact agreement, only a  {\em notarizability} property can be ensured: whenever any party outputs a block at position $j$ with   a valid quorum proof, at least $f+1$ honest parties already output   at the position $j-1$, cf. Fig. \ref{fig:nw-abc}. 
	%
	%
	%
	

		\smallskip
	\noindent \underline{\smash{\em How ``notarizability'' better prepares   honest parties?}} To see how      notarizability simplifies pace-sync, 
	let us examine the pattern of the honest parties' fastlane outputs before entering pace-sync. 

	Suppose all honest parties have quit the fastlane,
	exchanged their fallback requests (containing their latest block index and the corresponding quorum proof),
	received such $2f+1$ fallback requests, and thus entered the pace synchronization.
	At the time, let $s$ to be the largest index of all fastlane  blocks with valid proofs.  

	We can   make two easy claims: 
	(i) no honest party can see a valid fallback request with an index equal or larger than $s+1$;
	(ii) all honest parties must see some fallback request with an index equal or larger than $s-1$. 
	%
	%
	If (i) does not hold, following notarizability, at least one party can produce a proof for block $s+1$, which contradicts the definition of $s$. 
	While for (ii), since block $s$ is with a valid proof, at least $f+1$ honest parties received block $s-1$ with valid proof.
	So for any party waits for $2f+1$ fallback requests, it must see at least one fallback sent from some of these $f+1$ honest parties, thus seeing $s-1$; otherwise, there would be $3f+2$ parties. 
	
	The above two   claims    narrow the range of the honest parties' fallback positions  to   $\{s-1,s\}$, i.e., {\em two  unknown   consecutive integers}. 
	
	

	\begin{figure}[h] 
		\vspace{-0.45cm}
		\captionsetup{font={normalsize}}
		\begin{center}
			\includegraphics[width=8.5cm]{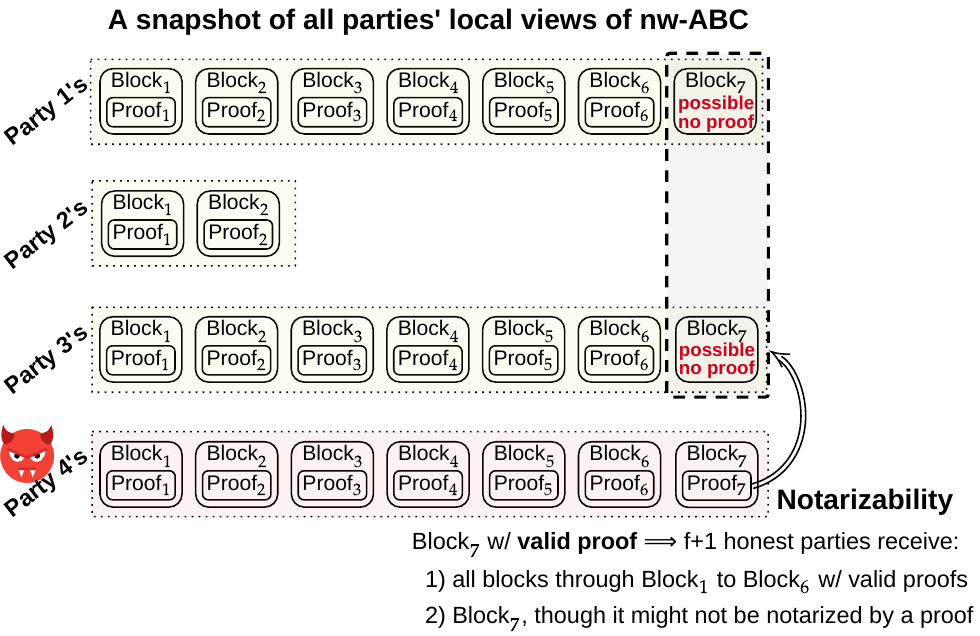}\\
			\vspace{-0.45cm}
			\caption{  {Notarizability}   of   fastlane abstraction ($\mathsf{nw\textnormal{-}ABC}$).}	\label{fig:nw-abc}
		\end{center}
		\vspace{-0.5cm}
	\end{figure}
	
	\smallskip
	\noindent \underline{\smash{\em Second ingredient: async. agreement for   consecutive values}}.
	Pace-sync now is reduced to pick one  value  of two  unknown consecutive integers  $\{s-1,s\}$. To handle the problem, we further define  {\em two-consecutive-valued Byzantine agreement} ($\tcvba$), which can be easily implemented  from any asynchronous {\em binary} Byzantine agreement (cf. Section \ref{sec:fast} for the concrete construction).

	\smallskip
	\noindent\underline{\smash{\em Final piece of the puzzle: adding ``safe-buffer'' to the fastlane.}} When $\tcvba$ outputs $u$,  all honest parties can sync up to block $u$ accordingly. Because no matter $u$ is $s$ or $s-1$, the $u$-th fastlane block is with a valid quorum proof, so it can be retrieved due to notarizability (cf. Fig. \ref{fig:nw-abc}).  
	Nevertheless, a subtle issue remains: $\tcvba$ cannot guarantee $u=s$, and thus $u$ could be $s-1$. 
	This is because an asynchronous adversary can always delay the messages of the honest parties that input $s$ to make them seemingly crashed.
	That means, if a party outputs a fastlane block immediately when seeing its   proof,
	it faces a threat that the pace-sync  returns a   smaller index and revokes the latest output fastlane  block.
	This can even temporarily violate safety requirement, if   other parties output a different block after pace-sync.
	To solve the issue,
	we introduce a ``safe buffer'' to let the newest fastlane block
	 be {\em pending}, i.e., not output until one more fastlane block with valid proof is received.




%



	


\section{Other Related Work}
In the past decades,   asynchronous BFT protocols are mostly theoretical results \cite{rabin1983randomized,ben1983another,canetti1993fast,bracha1987asynchronous,abraham2008almost,patra2009simple,patra2011error,ben2003resilient,benor,correia2006consensus}, 
until several recent progresses such as HoneyBadgerBFT \cite{honeybadger}, BEAT \cite{beat}, $\dumbo$ protocols \cite{guo2020dumbo,lu2020dumbo,guo2022speeding}, VABA \cite{ittai19},  DAG-based asynchronous protocols \cite{danezis2022narwhal,dag}, \rev{and DispersedLedger \cite{yang2021dispersedledger}}. 
Nevertheless, they still have  a   latency much larger than that of good-case   partially synchronous protocols. 
\rev{Besides the earlier   discussed optimistic asynchronous consensus \cite{kursawe05,cachin05} and more general framework \cite{700aublin}}, Spiegelman recently \cite{spiegelman2020search}  used
VABA \cite{ittai19} to instantiate  
pace-sync in optimistic asynchronous atomic broadcast. However,  it is still inefficient, especially when  fallbacks   frequently occur.
\rev{    $\BDT$ framework presents a generic and   efficient solution to add a deterministic fastlane to   most existing asynchronous consensus protocols (except the DAG-based protocols). 
	For example, it is compatible with   two very recent results of DispersedLedger    \cite{yang2021dispersedledger} and $\mathsf{Speeding}$-$\dumbo$ \cite{guo2022speeding}, and can directly employ them to instantiate   more efficient pessimistic path.}

It is well known that partially synchronous protocols \cite{yin2018hotstuff-full,pbft} can be responsive after GST in the absence of failures. Nonetheless,
if some parties are slow or even act maliciously, they might suffer from a worst-case latency related to the  upper bound of network delay. Some recent studies \cite{abraham2021good,abraham2020sync,thunderella,shrestha2020optimality,sbft,momose2020hybrid} also consider synchronous protocols with {\em optimistic responsiveness}, such that when some special conditions were satisfied, they can confirm transactions very quickly (with preserving optimal $n/2$ tolerance). Our protocol is {\em responsive} all the time, because it does not  wait for timeout that is set as large as the upper bound of network delay in all cases. 

Besides,  some literature \cite{blum2019synchronous,blum2021tardigrade,blum2020always,loss2018combining,momosemulti}  
studied how to combine synchronous and asynchronous protocols for stronger and/or flexible security guarantees in varying network environment.  
We instead aim to harvest efficiency from the   deterministic protocols. 

\noindent\underline{\smash{\em Concurrent works}}. 
A concurrent work  \cite{gelashvili2021prepared} considers adding an asynchronous view-change to a variant of HotStuff. Very recently its extended version    \cite{gelashvili2021jolteon} was presented with implementations. They focus on a specific construction of asynchronous fallback tailored for HotStuff by opening up a recent MVBA protocol \cite{ittai19}, thus can have different efficiency trade-offs. On the other hand, they cannot inherit the recent progress of asynchronous BFT protocols to preserve the linear per transaction communication (as we do) in the pessimistic path, or future improvements (since BDT is generic).  Moreover, \cite{gelashvili2021jolteon} essentially still uses an MVBA to handle pace-sync, while we reduce the task to conceptual minimum---a binary agreement, which itself could have more efficient constructions.

\section{Problem Formulation}
\label{sec:def}



 
\smallskip
\noindent{\bf Transaction}. Without loss of generality, we let a transaction denoted by $\tx$ to represent a string of $|m|$ bits. 

\smallskip
\noindent{\bf Block structure}. 
A block is a tuple in form of
$\block:= \langle epoch, $ $slot,  \payload,  \prf \rangle$, 
where $epoch$ and $slot$ are natural numbers, 
$\payload$ is  a sequence of transactions also known as the payload.  Throughout the paper, we assume  $|\payload|=B$, 
where $B$
be the batch size parameter. The batch size can be chosen to saturate the network's available bandwidth in practice.
$\prf$ is a 
quorum proof
  attesting that at least $f+1$ honest parties indeed vote the $\block$ by signing it.

\begin{figure}[h] 
	\vspace{-0.3cm}
	\captionsetup{font={normalsize}}
	\begin{center}
		\includegraphics[width=7cm]{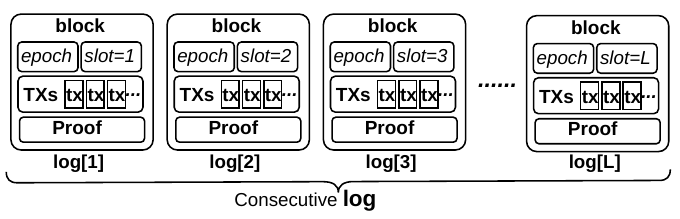}\\
		\vspace{-0.5cm}\caption{Block and  output log due to our terminology.}	\label{fig:block}
	\end{center}
	\vspace{-0.4cm}
\end{figure}

\smallskip
\noindent{\bf Blocks as output $\out$}.
Throughout the paper, a $\out$ (or interchangeably called as $\blocks$) refers to an indexed sequence of blocks. For   $\out$ with length $L:=|\out|$, 
we might use hereunder  notations.
(1) $\out[i]$ denotes the $i$-th block in $\out$. For example, $\out[1]$ is the first block of $\out$, 
	 $\out[-1]$ is the alias of the last block in $\out$,  and $\out[-2]$ represents the  second-to-last  block in $\out$, and so forth.
(2)  $\out.append(\cdot)$  can append some $\block$ to   $\out$. For example, when $\out.append(\cdot)$ takes a $\block \ne \emptyset$ as input, $|\out|$  increases by one and $\out[-1]$ becomes this newly appended $\block$; when $\out.append(\cdot)$ takes a sequence of non-empty blocks $[\block_{x+1}, \dots , \block_{x+k}]$ as input, $|\out|$ would increase by $k$, and $\out[-1]=\block_{x+k}$, $\out[-2]=\block_{x+k-1}$ and so on after the operation; when $\out.append(\cdot)$ takes  an empty block   $\emptyset$ as input, the $append$ operation does nothing.

\smallskip
\noindent{\bf Consecutive output $\out$}.
 An output  $\out$ consisting of $L$ blocks is said to be consecutive if it 
 satisfies: for any two successive blocks $\out[i]$ and $\out[i+1]$ included by $\out$, either of the following two cases is satisfied: (i) 
 $\out[i].epoch = \out[i+1].epoch$ and $\out[i].slot  + 1 = \out[i+1].slot$; or (ii) $\out[i].epoch + 1 = \out[i+1].epoch $ and $\out[i+1].slot = 1$. Without loss of generality,
 we let all   $\out$s to be {\em consecutive}  throughout the paper for presentation simplicity.

\subsection{Modeling the system and threats}
We consider the  standard  asynchronous message-passing  system   with trusted setup,
%
which can be detailed as follows.

	

	\smallskip
	 \noindent{\bf  Known identities and trusted setup}. There are $n$ designated parties, each of which has a unique identity (i.e., $\node_1$ through $\node_n$)  known by everyone else. 
	All involved threshold cryptosystems   are properly set up, so all   parties  can get and only get their own secret keys in addition to   relevant public keys. The setup   can be done by  a trusted dealer or distributed key generation   \cite{pedersen1991threshold,gennaro1999secure,kate2009distributed,gao2021efficient,renavss,kokoris2020asynchronous,abraham2021reaching,das2022practical}. 

\smallskip
	 \noindent{\bf  Byzantine corruptions}.  
	The adversary can  choose up to $f$    parties to   fully control   before the protocol starts. Our instantiations focus on   static corruptions, which is same to all recent {\em practical} asynchronous atomic broadcast \cite{honeybadger,beat,guo2020dumbo,guo2022speeding,yang2021dispersedledger}. 
	Also, no asynchronous BFT can tolerate more than $f = \lfloor (n-1)/3 \rfloor$ Byzantine corruptions.
	Through   the paper, we stick with this optimal resilience. 
	
	\smallskip
	 \noindent{\bf  Fully-meshed reliable asynchronous network}.
	There exists a reliable asynchronous  peer-to-peer channel between any two   parties.
	%
The adversary can arbitrarily delay or reorder messages, 
	 but cannot drop or modify   messages sent among honest parties. 

	\smallskip
	 \noindent{\bf  Computationally-bounded adversary}.
	We  consider  computationally bounded adversary that can   perform  some  probabilistic computing steps
	bounded by   polynomials	in the number of message bits generated by honest parties, which is standard cryptographic practice in the asynchronous network.

	\smallskip
	 \noindent{\bf  Adversary-controlling local ``time''}. 
	It is impossible to implement  global time   in the   asynchronous model. Nevertheless,  we do not require any global wall-clock for securities.
	Same to  \cite{cachin2002asynchronous,kursawe05}, it is still feasible to  let each party keep an adversary-controlling local ``clock'' that elapses at the speed of the actual network delay $\delta$:
	each party sends a  ``tick'' message      to itself via the adversary-controlling network, then whenever  receiving a ``tick'', it increases its local ``time'' by one and resends a new ``tick'' to itself via the adversary.	
	Using the adversary-controlling  ``clock'', each party can     maintain a  timeout mechanism, for example,	let $timer(\tau).start()$ to  denote that a local timer is initialized and  will ``expire''  after $\tau$ clock ticks, 	and let $timer(\tau).restart()$  denote to  reset the timer.	
	
	


\subsection{Security goal: async. atomic broadcast}
 
Our  primary goal is to develop an asynchronous atomic broadcast protocol defined as follows
to attain high robustness against   unstable or even hostile   network environment.
%

\begin{definition}\label{def:abc}
	In {\bf atomic broadcast} ($\mathsf{ABC}$),  
	each party is with  an implicit  queue  of input transactions (i.e., the input backlog)  and  outputs a  $\out$ of blocks.
	Besides the   syntax, the $\mathsf{ABC}$ protocol shall  satisfy the  following properties with all but negligible probability:
	\begin{itemize}[leftmargin=0.5cm]
		\item {\em Total-order}. If an honest party outputs a $\out$, and another honest   party outputs another $\out'$, then $\out[i]=\out'[i]$ for every $i$ that $1 \le i \le \min\{|\out|,|\out'|\}$.
		
		\item {\em Agreement}. If an honest party adds a $\block$ to its $\out$,   all honest parties would eventually add the   $\block$ to their $\out$s.

		\item {\em Liveness} (adapted from \cite{cachin01}). If all honest parties input a transaction $\tx$,  
		$\tx$ would    output   within some   asynchronous rounds (bounded by   polynomials in   security parameters).


	\end{itemize}
\end{definition}

\noindent
{\em Remarks on the definition of $\mathsf{ABC}$.} 	
Throughout the paper, 	we  let safety   refer to the union of total-order and agreement. Besides,
we   insist on the   liveness notion from \cite{cachin01} to ensure  that each input transaction can output reasonably quickly instead of eventually. 
This reasonable aim can   separate  some   studies that have {\em exponentially} large confirmation latency \cite{ben1983another}.
Moreover, the protocol must terminate in polynomial number of   rounds to restrict the computing steps of   adversary in the computationally-secure model  \cite{cachin01,ittai19,pass2017sleepy}, otherwise      cryptographic primitives are potentially insecure.

\subsection{Performance  metrics and preliminaries}

We are particularly interested in   {\em practical} asynchronous protocols, and therefore, 
 consider the following critical efficiency metrics:

\begin{itemize}[leftmargin=0.5cm]
	\item {\em Communication complexity}. We primarily focus on the (average) bits of all messages associated to output each block. 
	Because the communicated bits per block essentially reflects the (amortized) communication per delivered transaction, 
	in particular when   each block   includes   $\mathcal{O}(B)$-sized transactions, where  $B$ is a specified batch-size  parameter.
	
	
	\item {\em Message complexity}. This characterizes     the   number of messages exchanged among  honest parties to produce a block.

	\item {\em Asynchronous  round complexity}. 
	The eventual delivery in   asynchronous network    causes   the protocol execution   independent to ``real time''.
	Nevertheless, it is still needed to characterize the   running time, and 
	a standard way  to do so  is    counting  asynchronous ``rounds'' as in \cite{canetti1993fast,cachin01}.
	%
	

\end{itemize}





\noindent{\bf Cryptographic abstractions}. $\hash$ denotes a collision-resistant hash function. 
$\TSIG$ and $\TPKE$ denote threshold signature  and  threshold encryption, respectively. Established $\TSIG$ is a tuple  of algorithms $(\SignShare_t, \VrfyShare_t,\Combine_t, \Vrfy_t)$,  and throughout the paper, we call the signature share output from $\SignShare_t$ the {\em partial signature}, and   call the output of $\Combine_t$  the {\em full signature}. Established $\TPKE$ consists  three algorithms $(\Enc_t, \DecShare_t, \Dec_t)$. In all notations, the subscript $t$ represents the threshold, cf. some classic literature such as \cite{honeybadger}. The cryptographic security parameter is denoted by $\lambda$,   capturing the bit-length of signatures and hashes.

\medskip
 \noindent
{\bf Building blocks}. We might use the following asynchronous broadcast/consensus protocols in the black-box manner.

\vspace{-0.1cm}
\begin{definition}{\bf Reliable broadcast ($\mathsf{RBC}$)} has a designated sender who aims to send its input to all parties,  
	and satisfies the  next properties  except with negligible probability:
	%
		%
		(i) {\em Validity}. If the sender is honest and inputs $ v $, then all honest parties output $ v $;
		%
	    (ii) {\em Agreement}. The outputs of any two honest parties are same;
		%
		(iii) {\em Totality}. If an honest party outputs $ v $, then all honest parties output $ v $.
\end{definition}

\vspace{-0.2cm}
\begin{definition}{\bf Asynchronous binary  Byzantine agreement ($\ABA$)}  \cite{MMR15,cachin00,canetti1993fast}  has a syntax that  each party inputs and outputs a single bit $b$, where $b$ ranges over $\{0,1\}$, 
and  
shall guarantee the following properties except with negligible probability:
	(i) {\em Validity}. If any honest party outputs $ b $, then at least one honest party takes $ b $  as input;	
	(ii) {\em Agreement}. The outputs of any two honest parties are same;
	(iii) {\em Termination}. If all honest parties activate the protocol with taking a bit as input, then all honest parties would output a bit in the protocol.
\end{definition}


\vspace{-0.2cm}
\begin{definition}{\bf Asynchronous common subset ($\acs$)} \cite{benor} has a syntax that 
each party input a value and output a set of values, where $n$ parties participate in this protocol and up to $f$  corruption parties.
It      satisfies   the next properties   except with negligible probability:
	(i) {\em Validity}. The output set $\bf S$ of an honest party     contains the inputs of at least $n - 2f$ honest parties;
	(ii) {\em Agreement}. The outputs of any two honest parties are same;
	(iii) {\em Termination}. If all honest parties activate the protocol,  then all honest parties would output.
\end{definition}

\section{AllSpark: Fastlane Abstraction and  Two-Consecutive-Value BA}\label{sec:fast}


The simple and efficient pace-synchronization is the crux of making $\mathsf{BDT}$ practical,  
and this becomes possible for two critical ingredients, i.e., a novel fastlane abstraction ($\mathsf{nw\textnormal{-}ABC}$) and a new variant of binary Byzantine agreement ($\mathsf{tcv\textnormal{-}BA}$). Specifically,

\begin{itemize}[leftmargin=0.6cm]
	\item   $\mathsf{nw\textnormal{-}ABC}$ ensures that all parties' fastlane outputs are somewhat weakly consistent, namely, 
	if the $(s)$-th block is the latest block with valid quorum proof,
	then at least $f+1$ honest parties must already output the $(s-1)$-th block   with the valid quorum proof (cf. Fig. \ref{fig:nw-abc}).
	\item Considering the above property of $\mathsf{nw\textnormal{-}ABC}$ fastlane, we  can conclude that: after exchanging timeout requests, all honest parties either know $s$  or $s-1$. 
	We thus lift the conventional binary agreement to a special variant ($\mathsf{tcv\textnormal{-}BA}$) for deciding a common value out of $\{s-1, s\}$, despite that 
	the adversary might input arbitrarily, say  $s-2$ or $s+1$.
\end{itemize}





\vspace{-0.2cm}

\subsection{Abstracting and constructing    the fastlane}


We first put forth  notarizable weak atomic broadcast ($\mathsf{nw\textnormal{-}ABC}$)   to better prepare the fastlane for more efficient pace-sync.

\begin{definition} {{\bf Notarizable weak atomic broadcast}   ($\mathsf{nw\textnormal{-}ABC}$, nicknamed by $\tobc$)}. 
	In the   protocol with an identification $\id$, 
	 each party takes  a  transaction buffer   as input and  outputs a   $\out$ of blocks, where each block $\out[j]$ is in form of $\langle \id, j, \payload_j, \prf_{j} \rangle$. 
	 There also exists two external functions $\LogVerify$  and $\Extract$ taking $\id$, slot $j$ and $\prf_{j}$ as input (whose outputs and functionalities would soon be explained below). We require that $\tobc$   satisfies the following properties except with negligible probability:
	\begin{itemize}[leftmargin=0.6cm]
		\item {\em Total-order}. 
		Same to   atomic broadcast.
		
		\item {\em Notarizability}. If any (probably malicious) party outputs  $\out[j]:=\langle \id, j, \payload_j, \prf_j \rangle$   s.t.   $\LogVerify(\id, j, \prf_j)=1$,  then: there exist at least $f+1$ honest parties, each of which either already outputs $\out[j]$, or  already outputs $\out[j-1]$ and can    invoke  $\Extract$ function with valid $\prf_j$ to
		extract $\out[j]$  from   received protocol scripts.
		
		\item {\em Abandonability}.  An honest party  will not output any block in $\tobc[\id]$ after invoking $\abandon(\id)$. In addition, if $f+1$ honest parties invoke $\abandon(\id)$ before   output $\out[j]$, then no party can output valid $\out[j+1]$.
		
		\item {\em Optimistic liveness}. There exist  a non-empty collection of optimistic conditions  to specify the honesty of   certain   parties, s.t. once an honest party outputs $\out[j]$, it will output   $\out[j+1]$ in $\kappa$ asynchronous rounds, where  $\kappa$ is a   constant.

	\end{itemize}
\end{definition}

Comparing to    $\mathsf{ABC}$, 
$\mathsf{nw}$-$\mathsf{ABC}$ does not have  the exact agreement and liveness properties:
(i) notarizability   compensates the lack of agreement, as it ensures that  whenever a party outputs a block $\out[j]$ at position $j$,   at least $f+1$ honest parties already output  at the position $j-1$, and in addition, $f+1$ honest parties   already receive the protocol scripts carrying the payload of $\out[j]$, so they can   extract the block $\out[j]$ once seeing   valid $\prf_j$; 
(ii) liveness is in an optimistic form, which enables       simple deterministic   implementations of $\mathsf{nw}$-$\mathsf{ABC}$ in the asynchronous setting. 
%

Careful readers might   notice  that the above fastlane abstraction, in particular the notarizability property, share similarities with the popular lock-commit paradigm widely used in   (partially) synchronous byzantine/crash fault tolerant protocols \cite{partialsync,tendermint,pbft,sbft,yin2018hotstuff-full}. 
For example, when any honest party outputs some value (i.e. ``commit''), then at least $f+1$ honest parties shall receive and already vote this output (i.e. ``lock'').
In such a sense, the fastlane can be easily instantiated in many ways through the lens of  (partially) synchronous protocols. 
Unsurprisingly, one candidate is the fastlane  used in KS05 \cite{cachin05}.
Here  we present two more   exemplary  $\tobc$ constructions.

\smallskip
\noindent
\underline{\smash{\em \rev{Comparing with the  $\mathsf{Abstract}$ component in  \cite{700aublin}}}}. 
\rev{As aforementioned, \cite{700aublin} defined $\mathsf{Abstract}$ as a basic component to compose full-fledged BFT consensus with optimistic fastlane.
$\mathsf{Abstract}$ was defined to capture  a very broad array of optimistic conditions (including very optimistic cases such as no fault at all), 
such that a fastlane satisfying $\mathsf{Abstract}$ definition could be designed as simple as possible (with the price that no guarantee of similar progress among honest parties, as we have, before triggering fallback).
For example, \cite{700aublin} presented  $\mathsf{Quorum}$, an implementation of $\mathsf{Abstract}$ that only involves one round trip (with an optimistic condition allowing no fault),
but $\mathsf{Quorum}$ cannot meet the critical notarizability property of $\mathsf{nw}$-$\mathsf{ABC}$ though satisfying $\mathsf{Abstract}$.
Taking $\mathsf{Quorum}$ as example, the  weakening of $\mathsf{Abstract}$ prevents us from using binary agreement to handle some failed $\mathsf{Abstract}$ fastlanes, because the parties cannot reduce the failed position of the fastlane to two consecutive numbers.
This corresponds to the necessity of our stronger $\mathsf{nw}$-$\mathsf{ABC}$ definition in the context of facilitating a simplest possible pace-sync in the asynchronous setting.
}

\smallskip
\noindent
{\bf $\mathsf{Bolt}$ from sequential multicasts}. 
  As shown in   Fig. \ref{fig:multicast}, $\tobc$ can be easily constructed  from pipelined multicasts with using threshold signature, 
  and we call it $\mathsf{Bolt}$-$\mathsf{sCAST}$. The   idea  is as simple as: the leader proposes a batch of transactions via   multicast, 
  then all parties send back their signatures on the proposed batch as their votes,
  once the leader  collects enough votes from distinct parties (i.e., $2f+1$), it uses the votes to form a quorum proof for its precedent proposal,
   and then repeats to multicast a new proposal of transactions (along with the proof). Upon receiving the new proposal and the precedent proof, the parties   output the precedent proposal and the proof (as a block), and then    vote  on the new proposal. Such   execution is   repeated until the abandon interface is invoked.

\vspace{-0.1cm}
\begin{figure}[h]
	\captionsetup{font={normalsize}}
	\begin{footnotesize}
		\vspace{-0.1cm}
		\fbox{%
			\parbox{8.2cm}{				
				
				let $\id$ be the session identification of $\tobc[\id]$, $\buf$ be a FIFO queue of input, $B$ be the batch parameter,  and  $\node_\ell$ be the leader (where $\ell=(\id\mod n) +1$)
				
				\vspace{0.1cm}
				$\node_i$ initializes $s=1$, $\sigma_0=\bot$  and runs the protocol in  consecutive slot number $s$ as:
				
				\begin{itemize}[leftmargin=0.3cm]			
					\item	\textbf{Broadcast}. if $\node_i$ is the leader $\node_\ell$
					\begin{itemize}[leftmargin=0.25cm]
						
						
						\item	if $s > 1$ then: 
						\begin{itemize}[leftmargin=0.25cm]
							\item wait for $2f+1$ $\vote(\id, s-1,\sigma_{s-1,i})$ from distinct parties $\node_i$, where $\sigma_{s-1,i}$ is the valid partial signature signed by $\node_i$ for $\langle \id, s-1, \hash(\payload_{s-1}) \rangle$ 
							
							\item compute $\sigma_{s-1}$,   the   full-signature for $\langle \id, s-1, \hash(\payload_{s-1}) \rangle$,
							by aggregating the $2f+1$ received valid partial signatures
							

						\end{itemize}
						
						
						\item multicast $\proposal(\id, s, \payload_s,\sigma_{s-1})$, where $\payload_s \leftarrow \buf[:B]$ 
						
					\end{itemize}			
					
					\item	 \textbf{Commit and Vote}. upon receiving    $\proposal(\id, s, \payload_s,\sigma_{s-1})$ from   $\node_\ell$ 
					\begin{itemize}[leftmargin=0.3cm]
						
						
						\item	if $s > 1$ then:
						
						\begin{itemize}[leftmargin=0.25cm]
							\item  proceed only if $\sigma_{s-1}$ is valid full signature that aggregates $2f+1$ partial signatures for $\langle \id, s-1, \hash(\payload_{s-1}) \rangle$,  otherwise abort
							
							
							\item  output $\block$:=$ (\id, s-1, \payload_{s-1}, \prf_{s-1})$, where $\prf_{s-1}:=\langle \hash(\payload_{s-1}), \sigma_{s-1} \rangle$

						\end{itemize}			
						
						\item  send $\vote(\id, s, \sigma_{s,i})$ to the leader $\node_\ell$, where $\sigma_{s,i}$ is the partial signature for $\langle \id, s,  \hash(\payload_s) \rangle$, then let $ s\leftarrow s+1$
						
						
					\end{itemize}
					
					\item \textbf{Abandon}. upon $\abandon(\id)$ is invoked then:   abort the above execution
					
				\end{itemize}
				\vspace{-0.1cm}
			}
		}
		\vspace{-0.3cm}
		\caption{$\tobc$ from sequential multicasts ($\mathsf{Bolt}$-$\mathsf{sCAST}$). The external functions are presented in Fig. \ref{fig:external}.}
		
		\label{fig:multicast}
	\end{footnotesize}
	 \vspace{-0.2cm}
\end{figure}

 \smallskip
\noindent
{\bf $\mathsf{Bolt}$ from sequential reliable broadcast}.
As shown in Fig. \ref{fig:prbc}, we can also use sequential $\rbc$ instances   to implement $\tobc$.
In the implementation, a designated fastlane leader can reliably broadcast its proposed transaction batches one by one.
For each party receives a batch from some $\rbc$, it signs the   batch and  $\rbc$'s identifier, and multicasts the signature as vote, 
then wait for $2f+1$ valid votes to form a quorum proof, such that the batch and the proof assemble  an output block, and the party  proceeds into the next $\rbc$.
Note that a $\RBC$ implementation \cite{honeybadger} can use the technique of verifiable information dispersal \cite{cachin2005asynchronous} for communication efficiency as well as balancing network workload, such that the leader's bandwidth usage is at the same order of other parties'.
In contrast, $\mathsf{Bolt}$-$\mathsf{sCAST}$ might cause the  leader's bandwidth usage   $n$ times more than the other parties', 
unless an additional mempool layer is implemented to further decouple the dissemination of transactions from $\Bolt$.


\begin{figure}[h]
	\vspace{-0.15cm}
	\captionsetup{font={normalsize}}
	\begin{footnotesize}		
		\fbox{%
			\parbox{8.2cm}{				
				let $\id$ be the session identification of $\tobc[\id]$, $\buf$ be a FIFO queue of input, $B$ be the batch parameter,  and  $\node_\ell$ be the leader (where $\ell=(\id\mod n) +1$)
				
				\vspace{0.1cm}
				$\node_i$ initializes $s=1$, $\sigma_0=\bot$ and runs the protocol in  consecutive slot number $s$ as:
				
				\begin{itemize}[leftmargin=0.3cm]
					
					\item	\textbf{Broadcast}. if $\node_i$ is the leader $\node_\ell$,   activates $\rbc[\left \langle \id,s \right \rangle]$ with input $\payload_s \leftarrow \buf[:B]$; else activates $\rbc[\left \langle \id,s \right \rangle]$ as non-leader party
					
					\item  \textbf{Vote}.  upon $\rbc[\left \langle \id,s \right \rangle]$ returns $\payload_s$ 

					\begin{itemize}[leftmargin=0.25cm]
						
						\item  send $\vote(\id, s, \sigma_{s,i})$ to all, where $\sigma_{s,i}$ is the partial signature for $\langle \id, s,  \hash(\payload_s) \rangle$
						
						
					\end{itemize}			
					
					\item  \textbf{Commit}.  upon receiving $2f+1$  $\vote(\id, s,\sigma_{s,i})$ from distinct parties $\node_i$, where
					$\sigma_{s,i}$ is the valid partial-signature signed by $\node_i$ for $\langle \id, s, \hash(\payload_{s}) \rangle$
					
					\begin{itemize}[leftmargin=0.25cm]
						
						\item  output $\block$:=$ ( \id, s, \payload_{s}, \prf_{s})$, where $\prf_{s}:=\langle \hash(\payload_{s}), \sigma_{s} \rangle$ and $\sigma_{s}$ is the valid full signature that aggregates the $2f+1$ partial signatures for $\langle \id, s, \hash(\payload_{s}) \rangle$,  then let $ s\leftarrow s+1$ 
						
					\end{itemize}			
					
					\item \textbf{Abandon}. upon $\abandon(\id)$ is invoked then:   abort the above execution
					
				\end{itemize}
                \vspace{-0.1cm}	 
			}
		}
		\vspace{-0.3cm}
		\caption{$\tobc$ from sequential $\rbc$s ($\mathsf{Bolt}$-$\mathsf{sRBC}$). The external functions are presented in Fig. \ref{fig:external}.}
		\label{fig:prbc}
	\end{footnotesize}
	\vspace{-0.4cm}
\end{figure}

\begin{figure}[h]
	\vspace{-0.2cm}
	\captionsetup{font={normalsize}}
	\begin{footnotesize}		
		\fbox{%
			\parbox{8.2cm}{				
				{\color{orange}// Validate $\prf_{s}$ to check whether at least $f+1$ honest parties   output the $s$-th $\block$ or can extract it (according to the next $\Extract$ function)}
				
				\textbf{external function} $\LogVerify(\id, s, \prf_{s})$:  
				
				~~~~  \hspace{1em}parse $\prf_{s}$ as $\langle h_s, \sigma_s\rangle$
				
				~~~~  \hspace{1em}return $\TSIG.\Vrfy_{2f+1}(\langle \id, s, h_s\rangle, \sigma_s)$ 
				
				\dotfill
				
				{\color{orange}// Leverage the valid $\prf_{s}$ to extract the $s$-th $\block$ from some received protocol messages (though the block was not output yet).}
				
				\textbf{external function} $\Extract(\id, s, \prf_{s})$: 
				
				~~\hspace{0.8em} if $\LogVerify(\id, s, \prf_{s})=1$, then parse $\prf_{s}$ as $\langle h_s, \sigma_s\rangle$
				
				~~~~ \hspace{1em}if $\payload_{s}$ was received during executing $\tobc$ s.t. $ h_s=\hash(\payload_{s})$, then: 
				
				~~~~~~~~\hspace{2em}return $\block$:=$ ( \id, s, \payload_{s}, \prf_{s})$, where $\prf_{s}:=\langle h_s, \sigma_s\rangle$
				
				~~~~ \hspace{1em}return $\block$:=$ ( \id, s, \bot, \bot)$ 
				\vspace{-0.1cm}
			}
		}
		\vspace{-0.3cm}
		\caption{Invocable external functions for $\tobc$ instantiations}
		\label{fig:external}
	\end{footnotesize}
	\vspace{-0.2cm}
\end{figure}

\smallskip
\noindent
{\bf Analysis of the $\mathsf{Bolt}$ constructions}.
The security  analyses of $\mathsf{Bolt}$-$\mathsf{sCAST}$ and $\mathsf{Bolt}$-$\mathsf{sRBC}$ are simple by nature (cf. Appendix \ref{append:bolt}).
Their  complexities   can be easily counted as well (cf.  Appendix \ref{sec:complex}).


\subsection{Two-consecutive-value BA}
\label{sec:tcvba}

Another critical ingredient  is a variant of binary agreement that 
can help the honest parties to choose one common integer out of two unknown but consecutive numbers. 
%
Essentially, $\tcvba$ extends the conventional binary  agreement and can   be formalized as follows.

\begin{definition}  {\bf Two-consecutive-value Byzantine agreement}  ($\tcvba$) satisfies   termination, agreement and validity  (same to those of asynchronous binary agreement) with overwhelming probability, if all honest parties input a value in $\{v, v+1\}$ where $v \in \N$.
\end{definition}

\begin{figure}[h]
	\captionsetup{font={normalsize}}
	\begin{footnotesize}
		\centering
		\fbox{%
			\parbox{8.2cm}{%
				For each party $\node_{i}$, make the following modifications to the $\ABA$ code in Alg. 7 of \cite{guo2020dumbo} (originally from \cite{MMR15} but with some adaptions to use Ethan MacBrough's suggestion  \cite{ababug} to fix the potential liveness issues of \cite{MMR15}):
				
				\vspace{0.1cm}
				Replace line 13-23 of  Algorithm 7 in \cite{guo2020dumbo} with the next instructions:			
				\begin{itemize}[leftmargin=0.3cm]
					\item $c\leftarrow \coin_r.\getcoin()$
					\begin{itemize}[leftmargin=0.25cm]
						\item if $S_r= \{v\}$ then:									
						
						\begin{itemize}[leftmargin=0.25cm]
							\item {\color{orange} if $v\%2 = c\%2$ }
							\begin{itemize}[leftmargin=0.25cm]
								\item  if $decided$ = false then: output $v$; $decided$ = true
								\item   else (i.e, $decided$ = true) then: halt						
							\end{itemize}									
							\item $\est_{r+1} \leftarrow v$
						\end{itemize}
						
						\item if $S_r= \{v_1, v_2\}$ then:
						\begin{itemize}[leftmargin=0.25cm]
							\item {\color{orange}if $v_1\%2 = c\%2$, then $\est_{r+1} \leftarrow v_1$}
							\item {\color{orange}else (i.e, $v_2\%2 = c\%2$), then $\est_{r+1} \leftarrow v_2$}
						\end{itemize}
						
					\end{itemize}
				\end{itemize}
				\vspace{-0.1cm}	 
			}
		}
		\vspace{-0.3cm}
		\caption{ $\tcvba$ protocol. Lines different to Alg. 7 in  \cite{guo2020dumbo} are    in {\color{orange} orange texts}.}
		\label{fig:tcvba}
		
	\end{footnotesize}
	\vspace{-0.2cm}
\end{figure}

To squeeze extreme performance of pace synchronization, we   give a non-black-box construction $\tcvba$ that only has to revise three lines of code of the practical $\ABA$ construction  adapted from \cite{MMR15}. This non-black-box construction basically reuses the protocol pseudocode 
except several if-else  checking (see Fig. \ref{fig:tcvba}) and hence has the same performance of this widely adopted   $\ABA$ protocol.

In addition, $\tcvba$ can be constructed    from any  $\ABA$  with only one more    ``multicast'' round, cf. Figure \ref{fig:blackaba}. This black-box construction  provides us a convenient way to inherit any potential improvements of   underlying $\ABA$ primitives \cite{crain2020two, das2021practical}.

\vspace{-0.1cm}
\begin{figure}[h]
	\vspace{-0.1cm}
	\begin{footnotesize}
		\centering
		\fbox{%
			\parbox{8.2cm}{%
				Let $\ABA$  be any asynchronous binary agreement, then   party $\node_{i}$ executes:				
				
				\vspace{0.1cm}Upon receiving input $R$ then: 
				\begin{itemize}[leftmargin=0.3cm]
					\item multicast $\val(\id,R)$
				
					\item { upon receiving $\val(\id,R')$ from $f+1$ parties containing the same $R'$}
				\begin{itemize}[leftmargin=0.25cm]
					\item  if $\val(\id,R')$ has not been sent before, then: multicast $\val(\id,R')$				
				\end{itemize}
					
					\item wait for receiving $2f+1$ $\val(\id,v)$ messages from distinct parties carrying the same $v$, and activate $\ABA[\id]$ with   $v\%2$ as input
					
					\item wait for $\ABA[\id]$ returns $b$:
					\begin{itemize}
						\item if $v\%2 = b$, then: return $v$
						\item else: wait for receiving $f+1$ $\val(\id,v')$ messages from distinct parties containing the same $v'$ such that $v'\%2=b$, then return $v'$
					\end{itemize}
					
				\end{itemize}
				\vspace{-0.1cm}	 
			}
		}
	\vspace{-0.3cm}
		\caption{ $\tcvba$ protocol built from any $\ABA$ ``black-box''}
		\label{fig:blackaba}
	\end{footnotesize}
	  \vspace{-0.2cm}
\end{figure}

\section{$\sys$ framework}
\label{sec:bdt}


As Fig. \ref{fig:hybrid} outlines, the fastlane of $\mathsf{BDT}$   is a $\tobc$ instance wrapped by a timer.
If honest parties can   receive a new $\tobc$ block in time, they would restart the timer to wait for the next $\tobc$ block.
Otherwise, the timer   expires, and the honest parties   multicast a fallback request containing the latest $\tobc$ block's quorum proof that they can see. 

\begin{figure}[htbp]
	\vspace{-0.2cm}
	\begin{center}	
		\includegraphics[width=8cm]{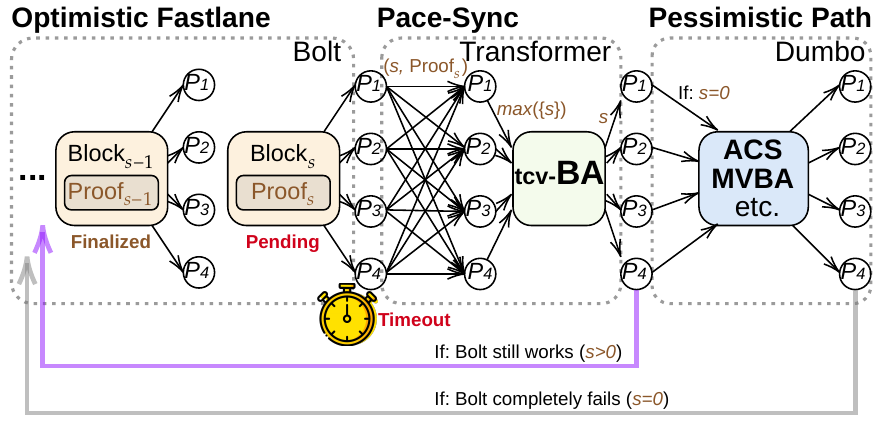}
	\end{center}
	\vspace{-0.25cm}
	\caption{The execution flow of $\sys$ 
	}	
	\label{fig:hybrid}
	\vspace{-0.2cm}
\end{figure}

After timeout, each party   waits for $n-f$   fallback requests with valid $\tobc$ block proofs, and enters pace-sync.
They invoke $\mathsf{tcv\textnormal{-}BA}$ with using the maximum block index (slot) in the received fallback requests as input.
Eventually, the honest parties enter the $\mathsf{tcv\textnormal{-}BA}$ and decide to either retry the fastlane or start the pessimistic path. As we briefly mentioned before, the reason we can use such a simple version of {\em binary} agreement is that via a careful analysis, we can find that \nwabc\/ prepares all honest parties will enter $\tcvba$ with one of neighboring indices as input.



The remaining non-triviality is that $\mathsf{tcv\textnormal{-}BA}$ cannot  ensure its output to always be the larger number out of the two possible inputs,
that means the globally latest $\tobc$ block can be revoked after pace-sync. Hence, the latest fastlane block is marked as ``pending'', 
and a ``pending'' block is finally output until the fastlane returns another new block.
This pending fastlane block ensures safety in $\mathsf{BDT}$.


\begin{figure*}[!b]
	\vspace{-0.2cm}
	\captionsetup{font={normalsize}}
	\begin{footnotesize}
		\centering
		\fbox{%
			\parbox{17.5cm}{%
				
				{\color{orange}// {\bf Optimistic Path} (also the BDT protocol's main entry)}					
				
				Every party $\node_i$ runs the protocol in  consecutive epoch numbered $e$ (initialized as 1) as follows:
				\begin{itemize}
					\item initialize: $p_e\leftarrow 0$, $\prf_{e}\leftarrow\bot$,  $\mathsf{Paces}_e\leftarrow\{\}$,	$\notarized_e\leftarrow\emptyset$ 
					\item activate $\tobc[e]$ instance, and start a timer that expires if not being restarted after $\tau$ clock ``ticks'' (i.e., invoke $timer(\tau).start()$)
					\item upon $\tobc[e]$ delivers a $\block$: 
					\begin{itemize}
						\ignore{
							\item parse $\block$:=$ \langle e, p, \payload_{p}, \prf_{p} \rangle $, where $p$ is the ``slot'' number
							\item $\out.append(\notarized_e)$, $\buf\leftarrow\buf \setminus \{\payload \textrm{ in } \notarized_e\}$ {\color{orange} ~~~//  the elder pending fastlane block becomes finalized, thus being output}
							\item 
							$\notarized_e\leftarrow \block$, $p_e\leftarrow p$, $\prf_{e}\leftarrow\prf_{p}$ {\color{orange} ~~~//  the newly output fastlane block becomes the pending block}
							\item $timer(\tau).restart$ {\color{orange} ~~~//   $\tobc[e]$ makes progress in time, so restart the ``heartbeat'' timer}
						}
						
						\item parse $\block$:=$ \langle e, p, \payload_{p}, \prf_{p} \rangle $, where $p$ is the ``slot'' number
						\item $\out.append(\notarized_e)$, $\buf\leftarrow\buf \setminus \{\payload \textrm{ in } \notarized_e\}$, $\notarized_e\leftarrow \block$ {\color{orange} ~~~//finalized  the elder pending block, pending the newly fastlane block}
						\item  $p_e\leftarrow p$, $\prf_{e}\leftarrow\prf_{p}$, $timer(\tau).restart$ {\color{orange} ~~~//   $\tobc[e]$ makes progress in time, so restart the ``heartbeat'' timer}
						
					\end{itemize}
					\item upon $timer(\tau)$ expires  or the front $\tx$ in the backlog $\buf$ was buffered   $T$ clock ``ticks'' ago:
					\begin{itemize}
						\item invoke  $\tobc[e]$.$\abandon()$  and multicast $\view(e, \pace_e, \prf_{e})$ {\color{orange} ~~~ // the fastlane is probably stucking or censoring certain transactions}
					\end{itemize}					
					
					\item upon receiving  message $\view(e, \pace_{e}^j, \prf_{e}^j)$ from  $\node_j$ for the first time:
					\begin{itemize}
						\item if $\LogVerify(e, \pace_{e}^j, \prf_{e}^j) = 1$: $\mathsf{Paces}_e \leftarrow \mathsf{Paces}_e\cup \pace_{e}^j$
					\end{itemize}
					
					\begin{itemize}
						\item if $|\mathsf{Paces}_e|=n-f$:  {\color{orange} ~~~// enough parties have already quitted the fastlane}
						
						\begin{itemize}
							\item invoke  $\mathsf{\textbf{Transformer(e)}}$ and wait for its return to continue {\color{orange} ~~~// enter into the pace-synchronization phase}
							\item proceed to the next epoch $e \leftarrow e + 1$  {\color{orange} ~~~// restart the fastlane of next epoch}
						\end{itemize}					
					\end{itemize}			
				\end{itemize}	
				
				\vspace{-0.1cm}
				\dotfill
				
				{\color{orange}// \bf Pace Synchronization}
				
				\textbf{internal function} $\trans$($e$): {~~~\color{orange}// {\bf internal} function shares   all internal states of the BDT protocol}
				
				\begin{itemize}
					\item let $\maxview_e \leftarrow \max(\mathsf{Paces}_e)$ and then $\mathsf{syncPace}_e  \leftarrow \tcvba[e]$($\maxview_e$)  {~~~\color{orange} // see Fig. \ref{fig:tcvba} or \ref{fig:blackaba} for concrete implementations of $\tcvba$}
					
					\item  if $\mathsf{syncPace}_e > 0$:  
					\begin{itemize}
						\item send $\view(e, \mathsf{syncPace}_e, \prf)$ to all if  $\mathsf{syncPace}_e\in\mathsf{Paces}_e$, where $\LogVerify(e, \mathsf{syncPace}_e, \prf) = 1$
						
						\item if $\mathsf{syncPace}_e = \pace_e$: $\out.append(\notarized_e)$ and $\buf=\buf \setminus \{\payload \textrm{ in } \notarized_e\}$ 
						
						\item if $\mathsf{syncPace}_e = \pace_e+1$:
						
						\begin{itemize}
							\item wait for a valid $\view(e, \mathsf{syncPace}_e, \prf)$, 
							then $\block'\leftarrow\Extract(e, \mathsf{syncPace}_e, \prf)$ {\color{orange} // try to extract the missing block}
							\item if $\block'$ is in form of $ ( e, \mathsf{syncPace}_e, \bot, \bot)$, then:  
							$\block'\leftarrow \gethelp(e, \pace_e, 1)$ {\color{orange} // failed to extract, have to rely on other   parties to fetch, cf. Fig. \ref{fig:help}}
							
							\item $\out.append(\notarized_e).append(\block')$ and $\buf=\buf \setminus \{\payload \textrm{ in } \notarized_e \textrm{ and } \block'\}$
							
						\end{itemize}
						
						\item if $\mathsf{syncPace}_e > \pace_e+1$:   $\blocks \leftarrow \gethelp(e, \pace_e, \mathsf{syncPace}_e - \pace_e)$      {\color{orange} // contact other parties to fetch   missing fastlane blocks, cf. Fig. \ref{fig:help}}
						\begin{itemize}
							\item $\out.append(\notarized_e).append(\blocks)$ and $\buf=\buf \setminus \{\payload \textrm{ in } \notarized_e \textrm{ and } \blocks\}$  
						\end{itemize}
						\item continue \textbf{Optimistic Path} with $e \leftarrow e+1$
					\end{itemize}
					
					\item if $\mathsf{syncPace}_e=0$: 
					invoke  $\mathsf{\textbf{Pessimistic(e)}}$  and wait for  its return, then continue \textbf{Optimistic Path} with $e \leftarrow e+1$
				\end{itemize}
				
				\vspace{-0.1cm}
				\dotfill
				
				{\color{orange}// {\bf Pessimistic Path}}
				
				\textbf{internal function} $\Pessimistic$($e$): {~~~\color{orange}// {\bf internal} function shares   all internal states of the BDT protocol}
				
				\begin{itemize}			
					\item $\txs_i$ $\leftarrow$  randomly select $\lfloor B/n \rfloor$-sized transactions from the first $B$-sized transactions at the top of $\buf$
					\item $x_i \leftarrow \TPKE.\Enc(epk, \txs_i)$, namely, encrypt $\txs_i$ to obtain $x_i$  
					\item $\{x_j\}_{j \in S } \leftarrow  \ACS[e](x_i)$, where $S \subset [n]$ and $|S|\ge n-f$ 
					\item 
					For each $j\in S$, jointly decrypt the ciphertext $x_j$ to obtain $\txs_j$, so the payload $\payload=\bigcup_{j \in S} \txs_j\leftarrow \{\TPKE.\Dec(epk, x_j)\}_{j \in S } $
					
					\item let $\block := \langle e, 1, \payload, \bot \rangle$, then $\out.append(\block)$ and $\buf=\buf \setminus \{\payload \textrm{ in } \block\}$ 
					
				\end{itemize} 
			}
		}
		\vspace{-0.4cm}
		\caption{The $\sys$ (BDT) protocol}
		
		\label{fig:mule}
	\end{footnotesize}
	  \vspace{-0.1cm}
\end{figure*}

\begin{figure*}[!b]
	\captionsetup{font={normalsize}}
 	\vspace{-0.1cm}
	\begin{footnotesize}
		\centering
		\fbox{%
			\parbox{17.5cm}{%
				
				{\color{orange}// {\bf The $\help$ daemon process}}\\
				$\help$: It is a daemon process that can read the finalized output $\out$ of $\sys$, 
				and it listens  to the down below event:
				\begin{itemize}
					\item  upon receiving message $\ghlp(e, \tip, \gap)$ from party $\node_j$ for the first time:
					\begin{itemize}
						\item assert $1 \leq \gap \le Esize$
						\item wait for $\out$ containing the $\block := \langle e, \tip + \gap, *, * \rangle$
						\item let $M \leftarrow$ retrieve all blocks in $\out$ from $\block := \langle e, \tip + 1, *, * \rangle$ to $\block := \langle e, \tip + \gap, *, * \rangle$

						\begin{itemize}
							\item let $\{m_k\}_{k \in [n]}$   be the fragements of a $(n-2f, n)$-erasure code   applied to $M$  and $h$   be a Merkle tree root computed over $\{m_k\}_{k \in [n]}$ 
							\item send $\hlp(e, \tip, \gap, h, m_i, b_i)$ to   $\node_j$ where $m_i$ is the $i$-th erasure-code fragement of $M$  and $b_i$ is the $i$-th Merkle tree branch 
						\end{itemize}
					\end{itemize}
				\end{itemize}
				
				\vspace{-0.2cm}
				\hrulefill
				
				{\color{orange}// {\bf The $\gethelp$ function}}
				
				{\textbf{external function} $\gethelp(e, \tip, \gap)$:} 
				\begin{itemize}
					\item let $F\leftarrow[\ ]$ to be a dictionary structure such that $F[h]$ can store all leaves committed to Merkle tree root $h$ 
					\item multicast message $\ghlp(e, \tip, \gap)$
					\item upon receiving the message $\hlp(e, \tip, \gap, h, m_j, b_j)$ from party $\node_j$ for the first time:
					\begin{itemize}
						\item if $b_j$ is a valid Merkle branch for root $h$ and leaf $m_j$ then: $F[h] \leftarrow F[h] \cup (j, m_j) $;
						otherwise discard the   message
						\item if $|F[h]| = n-2f$ then:
						\begin{itemize}
							\item interpolate the $n-2f$ leaves stored in $F[h]$ to reconstruct $M$, then parse $M$ as a sequence of $\block s$ and return $\block s$
						\end{itemize} 
					\end{itemize}
					
				\end{itemize}
			}
		}
	\vspace{-0.4cm}
		\caption{$\help$ and $\gethelp$. $\help$ is a daemon process having access to the output $\out$, and $\gethelp$ is a function   to call   $\help$}
		\label{fig:help}
	\end{footnotesize}
	 \vspace{-0.1cm}
\end{figure*}

\smallskip
\noindent
{\bf Protocol details.}
$\mathsf{BDT}$   is formally illustrated in Fig. \ref{fig:mule}.
It  employs a  reduction to $\mathsf{nw\textnormal{-}ABC}$, $\tcvba$, and some asynchronous consensus (e.g., $\ACS$). 
Informally, it proceeds as follows by successive epochs: 
\begin{enumerate}[leftmargin=0.6cm]
	\vspace{-0.1cm}
	\item
	{\em $\tobc$ phase.} 
	When an honest party enters an epoch $e$, 
	it   activates a $\tobc[e]$ instance, and locally starts an adversary-controlling  ``timer'' that expires after $\tau$ clock ``ticks'' and resets once hearing the ``heartbeat'' of $\tobc[e]$ (e.g., $\tobc[e]$ returns a new block). If one party receives a new   $\tobc[e]$ block in time without ``timeout'', it  temporarily records the block as  $\notarized$, finalizes the previous (non-empty) $\notarized$ block   as $\mathsf{BDT}$'s output, and sets its ``pace'' $p_e$ 	to the new block's slot number. Otherwise, the ``timeout'' mechanism interrupts, 
	and the party   abandons $\tobc[e]$. 
	%
	%
	Beside the above ``timeout'' mechanism to ensure $\tobc$ progress in time, we also consider   that some transactions are probably censored: if the oldest transaction (at the top of the input backlog) is not   output  for a  duration $T$, an interruption is also raised to abandon $\tobc[e]$.
	%
	Once a party  abandons $\tobc[e]$ for any    above reason, it immediately    multicasts   latest ``pace'' $p_e$   with  the corresponding block's proof  via  a $\view$ message. 

	\item
	 {\em $\trans$ phase.} 
	If an honest party receives   ($n-f$) valid $\view$ messages from distinct parties w.r.t.  $\tobc[e]$,    it enters   $\trans$. 
	In the   phase, the party chooses the maximum ``pace'' $\maxview$ out of the $n-f$  ``paces'' sent from distinct parties,
	and it would use this $\maxview$ as input to invoke the  $\tcvba[e]$ instance. 
	When $\tcvba[e]$ returns a value $\mathsf{syncPace}$,  all parties 
	agree to continue from the $\mathsf{syncPace}$-th block in $\tcvba[e]$.
	In some worse case that a party did not yet receive all blocks up to $\mathsf{syncPace}$,  it can fetch  the missing blocks from other  parties by calling  the
	 $\gethelp$ function    (cf. Fig. \ref{fig:help}).
	
	\item
	 {\em $\Pessimistic$ phase.} This phase may not be executed unless the optimistic fastlane of the current epoch $e$ makes no   progress at all, i.e,  $\mathsf{syncPace}=0$. 	In the worst case,  
	 $\Pessimistic$ is invoked to guarantee that some blocks (e.g., one) can be generated despite an adversarial   network   or   corrupt leaders, which becomes   the last line of defense   to ensure the critical liveness.
	
\end{enumerate}

\noindent
{\bf $\help$ and $\gethelp$.} 
%
	%
	 Besides the above main protocol procedures,
	 a party might  invoke the 	$\gethelp$ function to broadcast a $\ghlp$ message, when it realizes that some fastlane blocks are missing. 
	As Fig. \ref{fig:help} illustrates,   $\ghlp$ messages specify which blocks to retrieve,
	and every party also runs a $\help$ daemon   to  handle  $\ghlp$ messages.
	Actually, any honest party that invokes  $\gethelp$ can eventually retrieve the missing  blocks, because at least $f+1$ honest parties indeed   output   the  blocks under request. The $\help$ daemon can also use the  techniques of erasure-code and Merkle commitment tree in verifiable information dispersal \cite{cachin2005asynchronous,honeybadger},  such that it only  responds with a coded fragment of the requested blocks, thus saving the overall communication  cost by an $\bigO(n)$ order.

\smallskip
\noindent
{\bf Alternative pessimistic path.}
The exemplary $\Pessimistic$ path invokes $\dumbo$ to output one single block. 
Nonetheless, this is not the only design choice. First, \rev{$\BDT$ is a generic framework, and thus it is compatible with many recent asynchronous BFT protocols such as   DispersedLedger \cite{yang2021dispersedledger} and not restricted to $\dumbo$}. Second, there could be some global heuristics to estimate how many  blocks needed to generate during the pessimistic path according to some public   information (e.g., how many times the fastlane completely fails   in a stream). Designing such  heuristics to better fit   real-world Internet environments could be an interesting engineering question to explore in the future  but   does not impact any security analysis.




 \smallskip
\noindent
{\bf Security intuitions.} We brief the   security  intuitions of      $\BDT$   in the following, and defer   detailed proofs to Appendix \ref{append:bdt} for space limit. 

	 	\noindent
	 \underline{\smash{\em Safety}}. The core ideas of proving agreement and total-order are:
	 \vspace{-0.1cm}
	 \begin{itemize}[leftmargin=0.6cm]
	 	\item {\em $\trans$ returns a common index}. All honest parties must obtain the same block index from $\trans$, so they always agree the same fastlane block to continue the pessimistic path (or retry the fastlane). This is ensured by   $\tcvba$'s agreement.
	 	\item {\em $\trans$ returns an index   not ``too large''}. For the index returned from $\trans$, at least $f+1$ honest parties did receive all blocks (with valid proofs) up to this index. As such, if any party misses some blocks, it can easily fetch the correct blocks from these $f+1$ parties. 
	 	This is because the notarizability of $\tobc$ prevents the adversary from forging a proof for a fastlane block with an index higher than the actually delivered block. So no honest party would input some index of an irretrievable block to $\tcvba$, and then the validity of $\tcvba$ simply guarantees the claim.
	 	
	 	\item {\em $\trans$ returns an index   not ``too small''}. No honest party would revoke any fastlane block that was already committed as a finalized output. Since each honest party waits for $2f+1$  $\view$ messages   from distinct parties, then due to the  notarizability of $\tobc$, there is at least one $\view$ message contains $s-1$, where $s$ is the latest fastlane block (among all parties). So every honest party at least inputs $\tcvba$ with $s-1$. The validity of $\tcvba$ then ensures the output at least to be $s-1$ as well.
	 	Recall that there is   a ``safe buffer'' to hold the latest fastlane block as a pending one, the claim is then correct.
	 	
	 	\item {\em Pessimistic path and fastlane are safe}. Pessimistic path is trivially safe due to its agreement and total order. Fastlane  has total-order by definition, and its weaker agreement (notarizability) is complemented by $\trans$ as argued above.
	 \end{itemize}
	
	\noindent
	 \underline{\smash{\em Liveness}}. This stems from  the liveness of all three phases.	 
	 The liveness of fastlane is guaranteed by the ``timeout'' parameter $\tau$. That means, all honest parties can leave the fastlanes without  ``getting stuck''.  
	After that, all parties would invoke $\tcvba$  and   obtain $\mathsf{syncPace}$ as the $\tcvba$  output due to the termination   of   $\tcvba$; moreover, if any honest party realizes that it misses some fastlane blocks after obtaining $\mathsf{syncPace}$, it can   sync up to   $\mathsf{syncPace}$ within only two asynchronous rounds, because at least $f+1$ honest parties can help it to fetch the missing blocks. So no honest party would ``stuck'' during the $\trans$ phase.
	Finally, the honest parties would enter the $\Pessimistic$ phase if the fastlanes completely fail to output nothing. After that, the protocol must   output expected $\bigO(B)$-sized transactions, and ensures that  any transactions (at the $B$-top of all honest parties' backlogs) can output with a constant probability, 
	thus ensuring liveness even if in the worst case.

\smallskip
	\noindent{\bf Efficiency analysis.} The complexities can be analyzed by counting these of each underlying module.
	Overall, $\mathsf{BDT}$ would cost (expected) $\bigO(n)$ communicated  bits per output transaction, 
	and the  latency of each output block is of expected constant rounds. 
	These  complexities hold in all cases (no matter the network   is synchronous or asynchronous). 
	We defer such tedious counting on the number  of exchanged bits and  execution rounds to Section \ref{sec:complex}.

\smallskip
\noindent{\bf Optimistic conditions.} 
$\BDT$ has a simple and efficient deterministic fastlane that might  keep on progressing
under certain optimistic conditions, which intuitively are: (i) the actual network delay is smaller than some guessed timing parameter and (ii) the leader of fastlane is honest.
In Appendix \ref{app:optimistic}, we discuss why such optimistic conditions  can ensure  the progress of fastlanes.

 \vspace{-0.1cm}
\section{Performance Evaluation}
\label{evaluation}

\noindent
{\bf Implementation details}.
%
We   program the proof-of-concept implementations of $\mathsf{BDT}$, $\dumbo$ and 2-chained $\mathsf{HotStuff}$ in the same language (i.e. Python 3), with using the  same libraries and security parameters
for all   cryptographic implementations.
The   BFT protocols are implemented by single-process code.
Besides, a common     network layer is programmed   by using  unauthenticated TCP sockets.
The network layer is implemented as a separate   Python process to provide    non-blocking communication interface.

For common coin, it is realized by hashing   Boldyreva's pairing-based unique   threshold signature \cite{boldyreva2003threshold} (implemented over MNT224 curve).
For quorum proofs, we concatenate ECDSA signatures (implemented over secp256k1 curve).
%
For threshold public key encryption, the hybrid encryption approach implemented in HoneyBadger BFT is used  \cite{honeybadger}.
For  erasure coding, the Reed-Solomon implementation in the zfec library is adopted.
For timeout mechanism, we  use the    clock  in each EC2 instance to implement the local time in lieu of the adversary-controlling ``clock''  in our formal security model.
\rev{Our proof-of-concept codebase is available at https://github.com/yylluu/BDT}.

For notations,  $\mathsf{BDT}$-$\mathsf{sCAST}$ denotes $\mathsf{BDT}$     using       $\tobc$-$\mathsf{sCAST}$ as fastlane, 
while    $\mathsf{BDT}$-$\mathsf{sRBC}$ denotes the other instantiation   using $\tobc$-$s\RBC$. 
In addition, $\mathsf{BDT}$-$\mathsf{Timeout}$ denotes to use an idle fastlane that   just waits for timeout, which  can be used as benchmark to ``mimic'' the worst case that the fastlanes always  output  nothing due to constant denial-of-service attacks.

Due to space limitation, we only presented most intuitive experiments here, more evaluations can be found in Appendix \ref{app:data}.

\subsection{Evaluations in   wide-area network}

\noindent
{\bf Setup on Amazon EC2}. 
To demonstrate the practicability of $\BDT$ in realistic wide-area network (WAN),
we evaluate it among   Amazon EC2 c5.large instances (2 vCPUs and 4 GB   RAM) for $n$=64 and 100 parties, and also test $\dumbo$ and $\mathsf{HotStuff}$ in the same setting as reference points. 
All  EC2 instances  are evenly distributed in 16 AWS  regions, i.e., Virginia, Ohio, California, Oregon, Central Canada, S\~ao Paulo, Frankfurt, Ireland, London, Paris, Stockholm, Mubai, Seoul, Singapore, Tokyo and Sydney. 
\rev{All evaluation results in the WAN setting are measured back-to-back and averaged over two executions (each run for 5-10 minutes).} 

In the  WAN setting  tests, we might fix some parameters of $\mathsf{BDT}$ to intentionally amplify the fallback cost.
For example,   let each fastlane interrupt  after  output   only 50   blocks, so     $\trans$   is   frequently invoked.
We also set the fastlane's timeout parameter $\tau$ as large as 2.5 sec (nearly twenty times of the  one-way network latency in our test environment), so all fallbacks  triggered by    timeout would incur a 2.5-second overhead in addition to the $\trans$'s latency.

\smallskip
\noindent
{\bf Basic latency}. 
We firstly measure   the basic latency  to reflect how fast the protocols are  (in the good cases without faults or timeouts), if all blocks have nearly zero payload (cf. Fig. \ref{fig:blatency}). This     provides us the baseline   understanding about how fast BDT, $\mathsf{HotStuff}$ and $\dumbo$ can be to  handle the scenarios favoring low-latency.

\begin{figure}[h] 
	\captionsetup{font={normalsize}}
	\vspace{-0.2cm}
	\begin{center}
		\includegraphics[width=7.5cm]{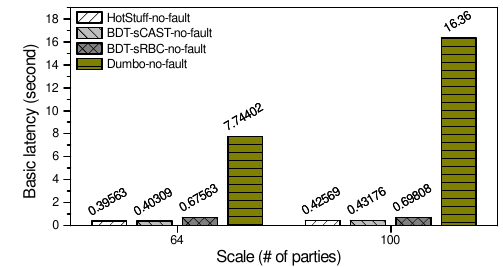}\\
		\vspace{-0.4cm}
		\caption{Basic latency   in experiments over WAN for two-chain $\mathsf{HotStuff}$, $\mathsf{BDT}$-$\mathsf{sCAST}$, $\mathsf{BDT}$-$\mathsf{sRBC}$ and $\dumbo$.}	
		\label{fig:blatency}
		\vspace{-0.2cm}
	\end{center}
\end{figure}

When $n=100$, $\mathsf{BDT}$-$\mathsf{sCAST}$ is 36x faster than $\dumbo$, and  $\mathsf{BDT}$-$\mathsf{sRBC}$ is 23x faster than $\dumbo$;  when $n=64$,
$\mathsf{BDT}$-$\mathsf{sCAST}$ is 18x faster than $\dumbo$, and  $\mathsf{BDT}$-$\mathsf{sRBC}$ is 10x faster than $\dumbo$; moreover, the execution speed of both $\mathsf{BDT}$-$\mathsf{sCAST}$ and  $\mathsf{BDT}$-$\mathsf{sRBC}$ are at the same magnitude of $\mathsf{HotStuff}$. In particular, the basic latency of $\mathsf{BDT}$-$\mathsf{sCAST}$ is almost as same as that of 2-chain  $\mathsf{Hotstuff}$, which is because the fastlane of $\mathsf{BDT}$-$\mathsf{sCAST}$ can be thought of a stable-leader 2-chain  $\mathsf{Hotstuff}$ and its optimistic latency has five rounds \footnote{The five-round latency of $\mathsf{BDT}$-$\mathsf{sCAST}$ in the best cases can be   counted as follows: one round for the leader to multicast the proposed batch of transactions, one round for the parties to vote (by signing), one round for the leader to multicast the quorum proof   (and thus all parties can get a   pending block), and finally two more rounds for every parties to receive one more   block and therefore output the earlier pending block. The concrete of rounds of $\mathsf{BDT}$-$\mathsf{sRBC}$ in the best cases  can be counted similarly.}, i.e., same to that of 2-chain $\mathsf{Hotstuff}$.

\smallskip
\noindent
{\bf Peak throughput}. We then measure throughput in unit of transactions per second (where each transaction is a 250 bytes string to  approximate the size of
a typical Bitcoin transaction).  
The peak throughput is depicted in Fig. \ref{fig:tps}, and 
gives us an insight how well $\mathsf{BDT}$, $\mathsf{HotStuff}$ and $\dumbo$ can handle transaction burst. \footnote{Note that we didn't implement an additional layer of mempool as in \cite{gao2022dumbong} and \cite{danezis2022narwhal}, and we can expect much higher throughout if we adopt their mempool techniques.}

\begin{figure}[h] 
	\vspace{-0.2cm}
	\captionsetup{font={normalsize}}
	\begin{center}
		\includegraphics[width=7.2cm]{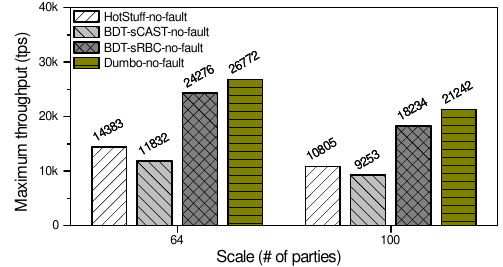}\\
		\vspace{-0.4cm}
		\caption{Peak throughput in experiments over WAN for two-chain $\mathsf{HotStuff}$, $\mathsf{BDT}$-$\mathsf{sCAST}$, $\mathsf{BDT}$-$\mathsf{sRBC}$ and $\dumbo$.}\label{fig:tps}	
	\end{center}
	\vspace{-0.2cm}
\end{figure}

$\mathsf{BDT}$-$\mathsf{sCAST}$ realizes a peak throughput   about 85\% of $\mathsf{HotStuff}$'s  when either $n$ is 100 or 64,  
$\mathsf{BDT}$-$\mathsf{sRBC}$ achieves a peak throughput that is as high as around 90\% of $\dumbo$'s   for  $n=64$ case and about 85\% of $\dumbo$'s    for  $n=100$ case. All these throughput numbers are achieved despite frequent $\trans$ occurrence, as we intend to let each fastlane to  fallback after  output mere   50 blocks.

\medskip
\noindent
{\bf Overhead of $\trans$}. It is critical for us to understand the  practical cost   of $\trans$.
%
We estimate such overhead from two different perspectives as shown in Fig. \ref{fig:vc} and \ref{fig:overhead}.

\begin{figure}[h] 
	\vspace{-0.2cm}
	\captionsetup{font={normalsize}}
	\begin{center}
		\includegraphics[width=7.5cm]{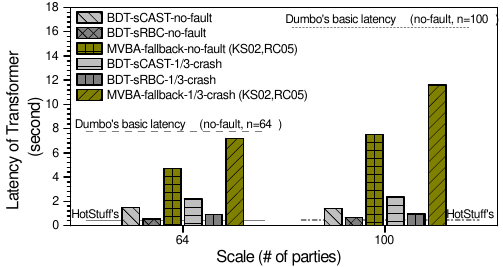}\\
		\vspace{-0.43cm}
		\caption{Latency of   $\trans$ for pace-sync  in $\mathsf{BDT}$-$\mathsf{sCAST}$ and $\mathsf{BDT}$-$\mathsf{sRBC}$ (when no fault and 1/3 crash, respectively). MVBA fallback in RC05 is also tested   as a reference point.}	
		\label{fig:vc}
	\end{center}
	\vspace{-0.2cm}
\end{figure}

As shown in Fig. 	\ref{fig:vc}, we   measure the execution time of   $\trans$   in various settings
by taking combinations of the following setups: (i)   $\mathsf{BDT}$-$\mathsf{sCAST}$ or $\mathsf{BDT}$-$\mathsf{sRBC}$; (ii)    $1/3$ crashes on or off; (iii) 64 EC2 instances or 100 EC2 instances. \rev{Moreover, in order to comprehensively compare $\trans$ with the prior art \cite{kursawe05,cachin05,700aublin},	
	we   also measure the latency of $\mathsf{MVBA}$   pace-sync (which instantiates the $\mathsf{Backup}$/$\mathsf{Abstract}$ primitive in \cite{700aublin} to combine the fastlane and $\mathsf{Dumbo}$ \footnote{\rev{Following \cite{700aublin} that used full-fledged $\mathsf{SMR}$ to instantiate Backup for fallback, 
			one can 
			combine stable-leader 2-chain $\mathsf{HotStuff}$ (the fastlane of $\mathsf{BDT}$-$\mathsf{sCAST}$) and $\mathsf{Dumbo}$ by a single block of asynchronous $\mathsf{SMR}$.
			This intuitive idea can be realized from $\mathsf{MVBA}$  \cite{cachin01,cachin05} as follows after the fastlane times out: each party   signs and multicasts the highest quorum proof  received from $\mathsf{HotStuff}$, then waits for $n-f$ such signed   proofs from distinct parties, and takes them as $\mathsf{MVBA}$ input;   $\mathsf{MVBA}$ thus would output $n-f$ valid $\mathsf{HotStuff}$ quorum proofs (signed by $n-f$   parties), and the highest quorum proof in the   $\mathsf{MVBA}$ output can represent the $\mathsf{HotStuff}$ block   to continue $\mathsf{Dumbo}$.}}) as a basic reference point, cf. Section 2 for the idea of using $\mathsf{Backup}$/$\mathsf{Abstract}$ for asynchronous fallback \cite{700aublin}.
}    
The comparison indicates that $\trans$ is much cheaper in contrast to the high cost of $\mathsf{MVBA}$ pace-sync.
For example,   $\trans$     always costs less than 1 second in $\mathsf{BDT}$-$\mathsf{sRBC}$, despite $n$ and on/off of crashes, while $\mathsf{MVBA}$ pace-sync is about 10 times slower.

\begin{figure}[h] 
	\captionsetup{font={normalsize}}
	\vspace{-0.3cm}
	\begin{center}
		\includegraphics[width=7.5cm]{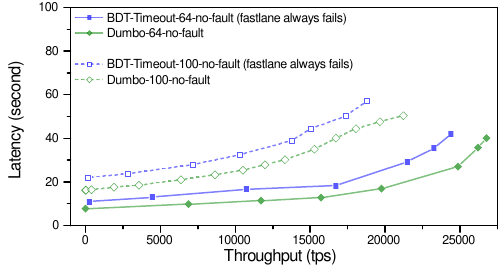}\\
		\vspace{-0.3cm}
		\caption{Latency v.s. throughput for experiments of $\mathsf{BDT}$ with   idling fastlane (i.e.,   fastlane just timeouts after 2.5 sec).}	
		\label{fig:overhead}
	\end{center}
	\vspace{-0.4cm}
\end{figure}

As illustrated in Fig. \ref{fig:overhead}, 
we measure the latency-throughput tradeoffs for $\mathsf{BDT}$-$\mathsf{Timeout}$, namely, to see how   $\mathsf{BDT}$   worse than $\dumbo$ when $\mathsf{BDT}$'s fastlane is under denial-of-service.
This is arguably the   worst-case  test vector for  $\mathsf{BDT}$, since relative to $\mathsf{Dumbo}$, it always costs extra 2.5 seconds to timeout and then executes the    $\trans$ subprotocol. Nevertheless,    the performance of $\mathsf{BDT}$  is still close to $\dumbo$. In particular, to realize the same throughput, $\mathsf{BDT}$    spends only a few additional seconds (which is  mostly caused by our conservation 2.5-second timeout parameter).

\medskip
\noindent
{\bf Latency-throughput trade-off}.  \rev{Figure \ref{fig:nofault} plots   latency-throughput trade-offs of $\mathsf{BDT}$-$\mathsf{sCAST}$, $\mathsf{BDT}$-$\mathsf{sRBC}$, $\mathsf{HotStuff}$ and $\mathsf{Dumbo}$ in the WAN setting for $n$= 64 and 100 parties. This illustrates that  $\mathsf{BDT}$ has low latency close to that of $\mathsf{HotStuff}$ under varying system load.}

\begin{figure}[h] 
	\vspace{-0.4cm}
	\captionsetup{font={normalsize}}
	\begin{center}
		\includegraphics[width=7.5cm]{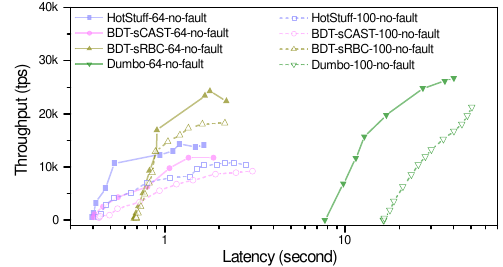}\\
		\vspace{-0.3cm}
		\caption{Throughput v.s. latency for experiments  over WAN  when $n=64$ and $100$, respectively (in case of periodically running pace-sync in $\mathsf{BDT}$ per only 50 fastlane blocks).}	
		\label{fig:nofault}
	\end{center}
	\vspace{-0.3cm}
\end{figure}

Either $\mathsf{BDT}$-$\mathsf{sCAST}$ or $\mathsf{BDT}$-$\mathsf{sRBC}$ is much faster than $\dumbo$ by several orders of magnitude in all cases, while the two $\mathsf{BDT}$ instantiations have different favors  towards distinct scenarios.
$\mathsf{BDT}$-$\mathsf{sCAST}$ has a latency-throughput trade-off similar to that of 2-chain $\mathsf{HotStuff}$, and their small variance in latency is because we intentionally trigger timeouts in $\mathsf{BDT}$-$\mathsf{sCAST}$ after each 50 fastlane blocks.
%
$\mathsf{BDT}$-$\mathsf{sRBC}$ has a latency-throughput trend quite different from $\mathsf{HotStuff}$ and $\mathsf{BDT}$-$\mathsf{sCAST}$. 
Namely, when fixing larger throughput, $\mathsf{BDT}$-$\mathsf{sRBC}$ has a latency less than $\mathsf{BDT}$-$\mathsf{sCAST}$'s;
when fixing   small throughput,  $\mathsf{BDT}$-$\mathsf{sRBC}$ could be slower.
This separates them clearly in terms of application scenarios, since $\mathsf{BDT}$-$\mathsf{sRBC}$ is   a better choice for large throughput-favoring cases and $\mathsf{BDT}$-$\mathsf{sCAST}$ 
is more suitable for latency-sensitive scenarios. 

\medskip
\noindent
{\bf Summary of evaluations in the  WAN setting}.
The above   results clearly demonstrate the   efficiency of our pace-synchronization---$\trans$.
And thanks to that, $\BDT$  in the WAN setting is:
\begin{itemize}[leftmargin=0.6cm]
	\item As fast as 2-chain $\mathsf{HotStuff}$ in the best case (i.e., synchronous network without faulty parties); \footnote{\rev{As discussed in Footnote 4, $\mathsf{BDT}$-$\mathsf{sCAST}$'s fastlane has a 5-round latency, which is same to that of 2-chain $\mathsf{HotStuff}$. The tiny difference between their evaluated latency is because we periodically trigger $\Trans$  in the experiments of $\mathsf{BDT}$-$\mathsf{sCAST}$.}}
	\item As robust as the underlying asynchronous pessimistic path  in the worst case (i.e., the fastlane always completely fails).
\end{itemize}


\subsection{Evaluation in  controlled  dynamic  network}

\noindent
{\bf Setup on the simulated fluctuating  network}. 
We also deploy our Python-written protocols for  $n$=64 parties in  a high-performance server
 having 4 28-core   Xeon Platinum 8280 CPUs and 1TB RAM. 
The   code is same to the earlier WAN experiments, except that we implement all   TCP sockets   with controllable bandwidth and delay. This allows us to  simulate a dynamic  communication network. 

In particular, we  interleave    ``good'' network (i.e., 50ms delay and 200Mbps bitrate) and   ``bad'' network  (i.e.,  300ms delay and 50Mbps bitrate) in the following experiments to reflect network fluctuation. \rev{Through  the subsection, $\BDT$ refers to $\BDT$-$\mathsf{sCAST}$, the   approach of using  Abstract  primitive \cite{700aublin}  to combine   stable-leader  2-chain $\mathsf{HotStuff}$ ($\BDT$-$\mathsf{sCAST}$'s fastlane)  and $\mathsf{Dumbo}$    is denoted by $\mathsf{HS}$+$\mathsf{Abstract}$+$\mathsf{Dumbo}$ (where   Backup/Abstract   is instantiated by $\mathsf{MVBA}$ as explained in Footnote 5).
	For experiment parameters, the fastlane's timeout   is  set as 1 second, the fastlane block and  pessimistic   block  contain  $10^4$ and  $10^6$ transactions respectively, and we would   report the number of confirmed transactions over time in   random sample executions.}

\smallskip
\noindent
{\bf Good network with  very short fluctuations}. We first examine in a network that mostly stays at good condition except interleaving some short-term bad network condition   that lasts only 2 seconds (which just triggers fastlane timeout). The sample executions in the setting are  plotted in Fig. \ref{fig:change-1}. 
The result indicates that the performance of $\BDT$ does not degrade due to the several short-term network fluctuation, 
and it remains as fast as 2-chain $\mathsf{HotStuff}$. This feature is because $\BDT$ adopts a two-level fallback mechanism, such that   it can just execute the light pace-sync and then immediately retry another fastlane. \rev{In contrast,   using $\mathsf{Backup}$/$\mathsf{Abstract}$ primitive (instantiated by $\mathsf{MVBA}$) as pace-sync     would encounter rather long latency ($\sim$25 sec) to run the heavy pace-sync and   pessimistic path after the short-term network fluctuations.} 


\begin{figure}[h] 
	\captionsetup{font={normalsize}}
	\begin{center}
		\includegraphics[width=7.2cm]{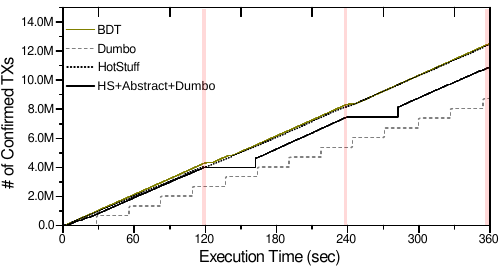}\\
		\vspace{-0.35cm}
		\caption{Sample executions of $\mathsf{BDT}$, 2-chain $\mathsf{HotStuff}$, $\dumbo$, \rev{and the  composition of $\mathsf{HotStuff}$+$\mathsf{Abstract}$+$\mathsf{Dumbo}$} for $n$=64, when facing a few 2-second bad periods. The red region represents the 2-second period of bad network.}	\label{fig:change-1}
	\end{center}
	\vspace{-0.4cm}
\end{figure}

\smallskip
\noindent
{\bf Intermittent network with long  bad time}. We then evaluate the effect of long-lasting bad network condition. We visualize such sample executions in Fig. \ref{fig:change-2}. Clearly, $\BDT$   can closely track the performance of its underlying pessimistic path during the long  periods of bad network condition. Again, this feature is a result of efficient pace-sync, as it adds minimal overhead to the fallback. \rev{In contrast,  using Backup/Abstract primitive (instantiated by $\mathsf{MVBA}$)   to compose stable-leader $\mathsf{HotStuff}$ and $\dumbo$  would incur a latency  $\sim$10 seconds larger than $\BDT$ during the bad network due to its cumbersome pace-sync.}

\begin{figure}[h] 
	\vspace{-0.4cm}
	\captionsetup{font={normalsize}}
	\begin{center}
		\includegraphics[width=7.2cm]{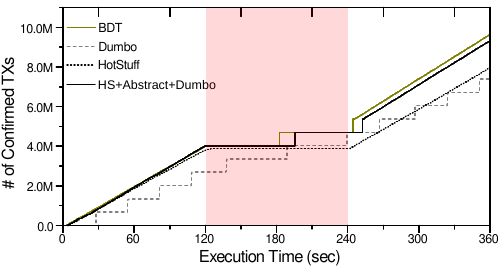}\\
		\vspace{-0.4cm}
		\caption{Sample executions of $\mathsf{BDT}$,  2-chain $\mathsf{HotStuff}$, $\dumbo$, \rev{and the composition of $\mathsf{HotStuff}$+$\mathsf{Abstract}$+$\mathsf{Dumbo}$} for $n$=64, when suffering from 120-second bad network. The red region represents the 120-second period of bad network.}	\label{fig:change-2}
	\end{center}
	\vspace{-0.4cm}
\end{figure}

\smallskip
\noindent
{\bf Summary of evaluations in   fluctuating network}. 
As expected by our efficient pace-sync subprotocol, $\BDT$ also performs well in the fluctuating network environment. Specifically,
\begin{itemize}[leftmargin=0.6cm]
	\item When   encountering short-term network fluctuations, $\BDT$ can quickly finish pace-sync and restart a new fastlane, thus progressing at a speed same to 2-chain $\mathsf{HotStuff}$.
	\item When the network   becomes slow for longer periods (and even $\mathsf{HotStuff}$   grinds to a halt), $\BDT$ still is robust to  progress    nearly as fast as the underlying   asynchronous    protocol.
\end{itemize}

\section{Conclusions}
In this paper, we proposed the first {\em practical} and {\em generic} framework for optimistic asynchronous atomic broadcast $\mathsf{BDT}$, in which we abstract a new and simple deterministic fastlane that enables us to reduce the asynchronous pace-synchronization to the conceptually minimum binary agreement. Several interesting questions remain: theoretically, formally demonstrating efficiency gap between asynchronous consensus and deterministic consensus, e.g., a better lower bound would be very interesting; practically, our current pessimistic path requires asynchronous common subset (ACS), building on top an ABC directly may need some further care; also, the paradigm of adding an optimistic path could be further generalized to provide not only efficiency, but also better resilience, or flexibility.

\begin{acks}
We would like to thank Vincent Gramoli and the anonymous reviewers for their valuable comments.
Yuan is supported in part by   NSFC under Grant  62102404 and the Youth Innovation Promotion Association CAS.
Qiang and Zhenliang are supported in part by research gifts from Ethereum Foundation, Stellar Foundation, Protocol Labs, Algorand Foundation and The University of Sydney.
\end{acks}

\appendix

\bibliographystyle{ACM-Reference-Format}
\bibliography{references}


\begin{thebibliography}{00}


\ifx \showCODEN    \undefined \def \showCODEN     #1{\unskip}     \fi
\ifx \showDOI      \undefined \def \showDOI       #1{#1}\fi
\ifx \showISBNx    \undefined \def \showISBNx     #1{\unskip}     \fi
\ifx \showISBNxiii \undefined \def \showISBNxiii  #1{\unskip}     \fi
\ifx \showISSN     \undefined \def \showISSN      #1{\unskip}     \fi
\ifx \showLCCN     \undefined \def \showLCCN      #1{\unskip}     \fi
\ifx \shownote     \undefined \def \shownote      #1{#1}          \fi
\ifx \showarticletitle \undefined \def \showarticletitle #1{#1}   \fi
\ifx \showURL      \undefined \def \showURL       {\relax}        \fi
\providecommand\bibfield[2]{#2}
\providecommand\bibinfo[2]{#2}
\providecommand\natexlab[1]{#1}
\providecommand\showeprint[2][]{arXiv:#2}

\bibitem[\protect\citeauthoryear{??}{aba}{}]%
        {ababug}
\showarticletitle{Bug in ABA protocol's use of Common Coin}.
\newblock
\showURL{%
\url{https://github.com/amiller/HoneyBadgerBFT/issues/59}}


\bibitem[\protect\citeauthoryear{Abraham, Dolev, and Halpern}{Abraham
  et~al\mbox{.}}{}]%
        {abraham2008almost}
\bibfield{author}{\bibinfo{person}{Ittai Abraham}, \bibinfo{person}{Danny
  Dolev}, {and} \bibinfo{person}{Joseph~Y Halpern}.}
\newblock \showarticletitle{An almost-surely terminating polynomial protocol
  for asynchronous byzantine agreement with optimal resilience}. In
  \bibinfo{booktitle}{{\em Proc. PODC 2008}}. \bibinfo{pages}{405--414}.
\newblock


\bibitem[\protect\citeauthoryear{Abraham, Jovanovic, Maller, Meiklejohn, Stern,
  and Tomescu}{Abraham et~al\mbox{.}}{2021}]%
        {abraham2021reaching}
\bibfield{author}{\bibinfo{person}{Ittai Abraham}, \bibinfo{person}{Philipp
  Jovanovic}, \bibinfo{person}{Mary Maller}, \bibinfo{person}{Sarah
  Meiklejohn}, \bibinfo{person}{Gilad Stern}, {and} \bibinfo{person}{Alin
  Tomescu}.} \bibinfo{year}{2021}\natexlab{}.
\newblock \showarticletitle{Reaching consensus for asynchronous distributed key
  generation}. In \bibinfo{booktitle}{{\em Proc. PODC 2021}}.
  \bibinfo{pages}{363--373}.
\newblock


\bibitem[\protect\citeauthoryear{Abraham, Malkhi, Nayak, Ren, and Yin}{Abraham
  et~al\mbox{.}}{2020}]%
        {abraham2020sync}
\bibfield{author}{\bibinfo{person}{Ittai Abraham}, \bibinfo{person}{Dahlia
  Malkhi}, \bibinfo{person}{Kartik Nayak}, \bibinfo{person}{Ling Ren}, {and}
  \bibinfo{person}{Maofan Yin}.} \bibinfo{year}{2020}\natexlab{}.
\newblock \showarticletitle{Sync hotstuff: Simple and practical synchronous
  state machine replication}. In \bibinfo{booktitle}{{\em 2020 IEEE Symposium
  on Security and Privacy (SP)}}. IEEE, \bibinfo{pages}{106--118}.
\newblock


\bibitem[\protect\citeauthoryear{Abraham, Malkhi, and Spiegelman}{Abraham
  et~al\mbox{.}}{}]%
        {ittai19}
\bibfield{author}{\bibinfo{person}{Ittai Abraham}, \bibinfo{person}{Dahlia
  Malkhi}, {and} \bibinfo{person}{Alexander Spiegelman}.}
\newblock \showarticletitle{Asymptotically Optimal Validated Asynchronous
  Byzantine Agreement}. In \bibinfo{booktitle}{{\em Proc. PODC 2019}}.
  \bibinfo{pages}{337--346}.
\newblock


\bibitem[\protect\citeauthoryear{Abraham, Nayak, Ren, and Xiang}{Abraham
  et~al\mbox{.}}{}]%
        {abraham2021good}
\bibfield{author}{\bibinfo{person}{Ittai Abraham}, \bibinfo{person}{Kartik
  Nayak}, \bibinfo{person}{Ling Ren}, {and} \bibinfo{person}{Zhuolun Xiang}.}
\newblock \showarticletitle{Good-Case Latency of Byzantine Broadcast: A
  Complete Categorization}. In \bibinfo{booktitle}{{\em Proc. PODC 2021}}.
\newblock


\bibitem[\protect\citeauthoryear{Amir, Coan, Kirsch, and Lane}{Amir
  et~al\mbox{.}}{2010}]%
        {amir2010prime}
\bibfield{author}{\bibinfo{person}{Yair Amir}, \bibinfo{person}{Brian Coan},
  \bibinfo{person}{Jonathan Kirsch}, {and} \bibinfo{person}{John Lane}.}
  \bibinfo{year}{2010}\natexlab{}.
\newblock \showarticletitle{Prime: Byzantine replication under attack}.
\newblock \bibinfo{journal}{{\em IEEE transactions on dependable and secure
  computing\/}} \bibinfo{volume}{8}, \bibinfo{number}{4}
  (\bibinfo{year}{2010}), \bibinfo{pages}{564--577}.
\newblock


\bibitem[\protect\citeauthoryear{Amoussou-Guenou, Del~Pozzo, Potop-Butucaru,
  and Tucci-Piergiovanni}{Amoussou-Guenou et~al\mbox{.}}{}]%
        {tendermint}
\bibfield{author}{\bibinfo{person}{Yackolley Amoussou-Guenou},
  \bibinfo{person}{Antonella Del~Pozzo}, \bibinfo{person}{Maria
  Potop-Butucaru}, {and} \bibinfo{person}{Sara Tucci-Piergiovanni}.}
\newblock \showarticletitle{Correctness of tendermint-core blockchains}. In
  \bibinfo{booktitle}{{\em Proc. OPODIS 2018}}.
\newblock


\bibitem[\protect\citeauthoryear{Attiya and Welch}{Attiya and Welch}{2004}]%
        {attiya2004distributed}
\bibfield{author}{\bibinfo{person}{Hagit Attiya} {and}
  \bibinfo{person}{Jennifer Welch}.} \bibinfo{year}{2004}\natexlab{}.
\newblock \bibinfo{booktitle}{{\em Distributed computing: fundamentals,
  simulations, and advanced topics}}. Vol.~\bibinfo{volume}{19}.
\newblock \bibinfo{publisher}{John Wiley \& Sons}.
\newblock


\bibitem[\protect\citeauthoryear{Aublin, Guerraoui, Kne{\v{z}}evi{\'c},
  Qu{\'e}ma, and Vukoli{\'c}}{Aublin et~al\mbox{.}}{2015}]%
        {700aublin}
\bibfield{author}{\bibinfo{person}{Pierre-Louis Aublin},
  \bibinfo{person}{Rachid Guerraoui}, \bibinfo{person}{Nikola
  Kne{\v{z}}evi{\'c}}, \bibinfo{person}{Vivien Qu{\'e}ma}, {and}
  \bibinfo{person}{Marko Vukoli{\'c}}.} \bibinfo{year}{2015}\natexlab{}.
\newblock \showarticletitle{The next 700 BFT protocols}.
\newblock \bibinfo{journal}{{\em ACM Transactions on Computer Systems
  (TOCS)\/}} \bibinfo{volume}{32}, \bibinfo{number}{4} (\bibinfo{year}{2015}),
  \bibinfo{pages}{1--45}.
\newblock


\bibitem[\protect\citeauthoryear{Aublin, Mokhtar, and Qu{\'e}ma}{Aublin
  et~al\mbox{.}}{2013}]%
        {aublin2013rbft}
\bibfield{author}{\bibinfo{person}{Pierre-Louis Aublin},
  \bibinfo{person}{Sonia~Ben Mokhtar}, {and} \bibinfo{person}{Vivien
  Qu{\'e}ma}.} \bibinfo{year}{2013}\natexlab{}.
\newblock \showarticletitle{Rbft: Redundant byzantine fault tolerance}. In
  \bibinfo{booktitle}{{\em 2013 IEEE 33rd International Conference on
  Distributed Computing Systems}}. \bibinfo{pages}{297--306}.
\newblock


\bibitem[\protect\citeauthoryear{Ben-Or}{Ben-Or}{}]%
        {ben1983another}
\bibfield{author}{\bibinfo{person}{Michael Ben-Or}.}
\newblock \showarticletitle{Another advantage of free choice (Extended
  Abstract) Completely asynchronous agreement protocols}. In
  \bibinfo{booktitle}{{\em Proc. PODC 1983}}. \bibinfo{pages}{27--30}.
\newblock


\bibitem[\protect\citeauthoryear{Ben-Or and El-Yaniv}{Ben-Or and
  El-Yaniv}{2003}]%
        {ben2003resilient}
\bibfield{author}{\bibinfo{person}{Michael Ben-Or} {and} \bibinfo{person}{Ran
  El-Yaniv}.} \bibinfo{year}{2003}\natexlab{}.
\newblock \showarticletitle{Resilient-optimal interactive consistency in
  constant time}.
\newblock \bibinfo{journal}{{\em Distributed Computing\/}}
  \bibinfo{volume}{16}, \bibinfo{number}{4} (\bibinfo{year}{2003}),
  \bibinfo{pages}{249--262}.
\newblock


\bibitem[\protect\citeauthoryear{Ben-Or, Kelmer, and Rabin}{Ben-Or
  et~al\mbox{.}}{}]%
        {benor}
\bibfield{author}{\bibinfo{person}{Michael Ben-Or}, \bibinfo{person}{Boaz
  Kelmer}, {and} \bibinfo{person}{Tal Rabin}.}
\newblock \showarticletitle{Asynchronous secure computations with optimal
  resilience}. In \bibinfo{booktitle}{{\em Proc. PODC 1994}}.
  \bibinfo{pages}{183--192}.
\newblock


\bibitem[\protect\citeauthoryear{Bessani, Sousa, and Alchieri}{Bessani
  et~al\mbox{.}}{2014}]%
        {bft-smart}
\bibfield{author}{\bibinfo{person}{Alysson Bessani}, \bibinfo{person}{Jo{\~a}o
  Sousa}, {and} \bibinfo{person}{Eduardo~EP Alchieri}.}
  \bibinfo{year}{2014}\natexlab{}.
\newblock \showarticletitle{State machine replication for the masses with
  BFT-SMaRt}. In \bibinfo{booktitle}{{\em Proc. DSN 2014}}.
  \bibinfo{pages}{355--362}.
\newblock


\bibitem[\protect\citeauthoryear{Blum, Katz, and Loss}{Blum et~al\mbox{.}}{}]%
        {blum2019synchronous}
\bibfield{author}{\bibinfo{person}{Erica Blum}, \bibinfo{person}{Jonathan
  Katz}, {and} \bibinfo{person}{Julian Loss}.}
\newblock \showarticletitle{Synchronous consensus with optimal asynchronous
  fallback guarantees}. In \bibinfo{booktitle}{{\em Proc. TCC 2019}}.
  \bibinfo{pages}{131--150}.
\newblock


\bibitem[\protect\citeauthoryear{Blum, Katz, and Loss}{Blum
  et~al\mbox{.}}{2021}]%
        {blum2021tardigrade}
\bibfield{author}{\bibinfo{person}{Erica Blum}, \bibinfo{person}{Jonathan
  Katz}, {and} \bibinfo{person}{Julian Loss}.} \bibinfo{year}{2021}\natexlab{}.
\newblock \showarticletitle{Tardigrade: An Atomic Broadcast Protocol for
  Arbitrary Network Conditions}. In \bibinfo{booktitle}{{\em International
  Conference on the Theory and Application of Cryptology and Information
  Security}}. Springer, \bibinfo{pages}{547--572}.
\newblock


\bibitem[\protect\citeauthoryear{Blum, Liu-Zhang, and Loss}{Blum
  et~al\mbox{.}}{2020}]%
        {blum2020always}
\bibfield{author}{\bibinfo{person}{Erica Blum}, \bibinfo{person}{Chen-Da
  Liu-Zhang}, {and} \bibinfo{person}{Julian Loss}.}
  \bibinfo{year}{2020}\natexlab{}.
\newblock \showarticletitle{Always have a backup plan: fully secure synchronous
  MPC with asynchronous fallback}. In \bibinfo{booktitle}{{\em Annual
  International Cryptology Conference}}. Springer, \bibinfo{pages}{707--731}.
\newblock


\bibitem[\protect\citeauthoryear{Boldyreva}{Boldyreva}{}]%
        {boldyreva2003threshold}
\bibfield{author}{\bibinfo{person}{Alexandra Boldyreva}.}
\newblock \showarticletitle{Threshold signatures, multisignatures and blind
  signatures based on the gap-Diffie-Hellman-group signature scheme}. In
  \bibinfo{booktitle}{{\em Proc. PKC 2003}}. \bibinfo{pages}{31--46}.
\newblock


\bibitem[\protect\citeauthoryear{Bracha}{Bracha}{1987}]%
        {bracha1987asynchronous}
\bibfield{author}{\bibinfo{person}{Gabriel Bracha}.}
  \bibinfo{year}{1987}\natexlab{}.
\newblock \showarticletitle{Asynchronous Byzantine agreement protocols}.
\newblock \bibinfo{journal}{{\em Information and Computation\/}}
  \bibinfo{volume}{75}, \bibinfo{number}{2} (\bibinfo{year}{1987}),
  \bibinfo{pages}{130--143}.
\newblock


\bibitem[\protect\citeauthoryear{Buterin et~al\mbox{.}}{Buterin
  et~al\mbox{.}}{2014}]%
        {buterin2014next}
\bibfield{author}{\bibinfo{person}{Vitalik Buterin} {et~al\mbox{.}}}
  \bibinfo{year}{2014}\natexlab{}.
\newblock \showarticletitle{A next-generation smart contract and decentralized
  application platform}.
\newblock \bibinfo{journal}{{\em white paper\/}} \bibinfo{volume}{3},
  \bibinfo{number}{37} (\bibinfo{year}{2014}).
\newblock


\bibitem[\protect\citeauthoryear{Cachin, Kursawe, Lysyanskaya, and
  Strobl}{Cachin et~al\mbox{.}}{b}]%
        {cachin2002asynchronous}
\bibfield{author}{\bibinfo{person}{Christian Cachin}, \bibinfo{person}{Klaus
  Kursawe}, \bibinfo{person}{Anna Lysyanskaya}, {and} \bibinfo{person}{Reto
  Strobl}.}
\newblock \showarticletitle{Asynchronous verifiable secret sharing and
  proactive cryptosystems}. In \bibinfo{booktitle}{{\em Proc. CCS 2002}}.
  \bibinfo{pages}{88--97}.
\newblock


\bibitem[\protect\citeauthoryear{Cachin, Kursawe, Petzold, and Shoup}{Cachin
  et~al\mbox{.}}{c}]%
        {cachin01}
\bibfield{author}{\bibinfo{person}{Christian Cachin}, \bibinfo{person}{Klaus
  Kursawe}, \bibinfo{person}{Frank Petzold}, {and} \bibinfo{person}{Victor
  Shoup}.}
\newblock \showarticletitle{Secure and efficient asynchronous broadcast
  protocols}. In \bibinfo{booktitle}{{\em Proc. CRYPTO 2001}}.
  \bibinfo{pages}{524--541}.
\newblock


\bibitem[\protect\citeauthoryear{Cachin, Kursawe, and Shoup}{Cachin
  et~al\mbox{.}}{a}]%
        {cachin00}
\bibfield{author}{\bibinfo{person}{Christian Cachin}, \bibinfo{person}{Klaus
  Kursawe}, {and} \bibinfo{person}{Victor Shoup}.}
\newblock \showarticletitle{Random oracles in constantipole: practical
  asynchronous Byzantine agreement using cryptography}. In
  \bibinfo{booktitle}{{\em Proc. PODC 2020}}. \bibinfo{pages}{123--132}.
\newblock


\bibitem[\protect\citeauthoryear{Cachin and Tessaro}{Cachin and Tessaro}{}]%
        {cachin2005asynchronous}
\bibfield{author}{\bibinfo{person}{Christian Cachin} {and}
  \bibinfo{person}{Stefano Tessaro}.}
\newblock \showarticletitle{Asynchronous verifiable information dispersal}. In
  \bibinfo{booktitle}{{\em Proc. SRDS 2005}}. \bibinfo{pages}{191--201}.
\newblock


\bibitem[\protect\citeauthoryear{Cachin and Vukolic}{Cachin and Vukolic}{}]%
        {cachin2017blockchain}
\bibfield{author}{\bibinfo{person}{Christian Cachin} {and}
  \bibinfo{person}{Marko Vukolic}.}
\newblock \showarticletitle{Blockchain Consensus Protocols in the Wild (Keynote
  Talk)}. In \bibinfo{booktitle}{{\em Proc. DISC 2017}}.
\newblock


\bibitem[\protect\citeauthoryear{Canetti and Rabin}{Canetti and Rabin}{}]%
        {canetti1993fast}
\bibfield{author}{\bibinfo{person}{Ran Canetti} {and} \bibinfo{person}{Tal
  Rabin}.}
\newblock \showarticletitle{Fast asynchronous Byzantine agreement with optimal
  resilience}. In \bibinfo{booktitle}{{\em Proc. STOC 1993}}.
  \bibinfo{pages}{42--51}.
\newblock


\bibitem[\protect\citeauthoryear{Castro and Liskov}{Castro and Liskov}{2002}]%
        {castro2002practical}
\bibfield{author}{\bibinfo{person}{Miguel Castro} {and}
  \bibinfo{person}{Barbara Liskov}.} \bibinfo{year}{2002}\natexlab{}.
\newblock \showarticletitle{Practical Byzantine fault tolerance and proactive
  recovery}.
\newblock \bibinfo{journal}{{\em ACM Transactions on Computer Systems
  (TOCS)\/}} \bibinfo{volume}{20}, \bibinfo{number}{4} (\bibinfo{year}{2002}),
  \bibinfo{pages}{398--461}.
\newblock


\bibitem[\protect\citeauthoryear{Castro, Liskov, et~al\mbox{.}}{Castro
  et~al\mbox{.}}{}]%
        {pbft}
\bibfield{author}{\bibinfo{person}{Miguel Castro}, \bibinfo{person}{Barbara
  Liskov}, {et~al\mbox{.}}}
\newblock \showarticletitle{Practical Byzantine fault tolerance}. In
  \bibinfo{booktitle}{{\em Proc. OSDI 1999}}. \bibinfo{pages}{173--186}.
\newblock


\bibitem[\protect\citeauthoryear{Chan and Shi}{Chan and Shi}{}]%
        {chan2020streamlet}
\bibfield{author}{\bibinfo{person}{Benjamin~Y Chan} {and}
  \bibinfo{person}{Elaine Shi}.}
\newblock \showarticletitle{Streamlet: Textbook streamlined blockchains}. In
  \bibinfo{booktitle}{{\em Proc. AFT 2020}}. \bibinfo{pages}{1--11}.
\newblock


\bibitem[\protect\citeauthoryear{Clement, Wong, Alvisi, Dahlin, and
  Marchetti}{Clement et~al\mbox{.}}{}]%
        {clement2009making}
\bibfield{author}{\bibinfo{person}{Allen Clement}, \bibinfo{person}{Edmund~L
  Wong}, \bibinfo{person}{Lorenzo Alvisi}, \bibinfo{person}{Michael Dahlin},
  {and} \bibinfo{person}{Mirco Marchetti}.}
\newblock \showarticletitle{Making Byzantine Fault Tolerant Systems Tolerate
  Byzantine Faults.}. In \bibinfo{booktitle}{{\em Proc. NSDI 2009}},
  Vol.~\bibinfo{volume}{9}. \bibinfo{pages}{153--168}.
\newblock


\bibitem[\protect\citeauthoryear{Correia, Neves, and Ver{\'\i}ssimo}{Correia
  et~al\mbox{.}}{2006}]%
        {correia2006consensus}
\bibfield{author}{\bibinfo{person}{Miguel Correia},
  \bibinfo{person}{Nuno~Ferreira Neves}, {and} \bibinfo{person}{Paulo
  Ver{\'\i}ssimo}.} \bibinfo{year}{2006}\natexlab{}.
\newblock \showarticletitle{From consensus to atomic broadcast: Time-free
  Byzantine-resistant protocols without signatures}.
\newblock \bibinfo{journal}{{\it Comput. J.}} \bibinfo{volume}{49},
  \bibinfo{number}{1} (\bibinfo{year}{2006}), \bibinfo{pages}{82--96}.
\newblock


\bibitem[\protect\citeauthoryear{Crain}{Crain}{2020}]%
        {crain2020two}
\bibfield{author}{\bibinfo{person}{Tyler Crain}.}
  \bibinfo{year}{2020}\natexlab{}.
\newblock \showarticletitle{Two More Algorithms for Randomized Signature-Free
  Asynchronous Binary Byzantine Consensus with t< n/3 and O(n\^{}2) Messages
  and O(1) Round Expected Termination}.
\newblock \bibinfo{journal}{{\em arXiv preprint arXiv:2002.08765\/}}
  (\bibinfo{year}{2020}).
\newblock


\bibitem[\protect\citeauthoryear{Danezis, Kokoris-Kogias, Sonnino, and
  Spiegelman}{Danezis et~al\mbox{.}}{2022}]%
        {danezis2022narwhal}
\bibfield{author}{\bibinfo{person}{George Danezis}, \bibinfo{person}{Lefteris
  Kokoris-Kogias}, \bibinfo{person}{Alberto Sonnino}, {and}
  \bibinfo{person}{Alexander Spiegelman}.} \bibinfo{year}{2022}\natexlab{}.
\newblock \showarticletitle{Narwhal and Tusk: a DAG-based mempool and efficient
  BFT consensus}. In \bibinfo{booktitle}{{\em Proc. EuroSys 2022}}.
  \bibinfo{pages}{34--50}.
\newblock


\bibitem[\protect\citeauthoryear{Das, Xiang, and Ren}{Das
  et~al\mbox{.}}{2021a}]%
        {renavss}
\bibfield{author}{\bibinfo{person}{Sourav Das}, \bibinfo{person}{Zhuolun
  Xiang}, {and} \bibinfo{person}{Ling Ren}.} \bibinfo{year}{2021}\natexlab{a}.
\newblock \showarticletitle{Asynchronous data dissemination and its
  applications}. In \bibinfo{booktitle}{{\em Proc. CCS 2021}}.
  \bibinfo{pages}{2705--2721}.
\newblock


\bibitem[\protect\citeauthoryear{Das, Yurek, Xiang, Miller, Kokoris-Kogias, and
  Ren}{Das et~al\mbox{.}}{2021b}]%
        {das2021practical}
\bibfield{author}{\bibinfo{person}{Sourav Das}, \bibinfo{person}{Tom Yurek},
  \bibinfo{person}{Zhuolun Xiang}, \bibinfo{person}{Andrew Miller},
  \bibinfo{person}{Lefteris Kokoris-Kogias}, {and} \bibinfo{person}{Ling Ren}.}
  \bibinfo{year}{2021}\natexlab{b}.
\newblock \showarticletitle{Practical asynchronous distributed key generation}.
\newblock \bibinfo{journal}{{\em Cryptology ePrint Archive\/}}
  (\bibinfo{year}{2021}).
\newblock


\bibitem[\protect\citeauthoryear{Das, Yurek, Xiang, Miller, Kokoris-Kogias, and
  Ren}{Das et~al\mbox{.}}{2022}]%
        {das2022practical}
\bibfield{author}{\bibinfo{person}{Sourav Das}, \bibinfo{person}{Thomas Yurek},
  \bibinfo{person}{Zhuolun Xiang}, \bibinfo{person}{Andrew Miller},
  \bibinfo{person}{Lefteris Kokoris-Kogias}, {and} \bibinfo{person}{Ling Ren}.}
  \bibinfo{year}{2022}\natexlab{}.
\newblock \showarticletitle{Practical Asynchronous Distributed Key Generation}.
  In \bibinfo{booktitle}{{\em 2022 IEEE Symposium on Security and Privacy
  (SP)}}. \bibinfo{pages}{2518--2534}.
\newblock


\bibitem[\protect\citeauthoryear{Duan, Reiter, and Zhang}{Duan
  et~al\mbox{.}}{}]%
        {beat}
\bibfield{author}{\bibinfo{person}{Sisi Duan}, \bibinfo{person}{Michael~K
  Reiter}, {and} \bibinfo{person}{Haibin Zhang}.}
\newblock \showarticletitle{BEAT: Asynchronous BFT made practical}. In
  \bibinfo{booktitle}{{\em Proc. CCS 2018}}. \bibinfo{pages}{2028--2041}.
\newblock


\bibitem[\protect\citeauthoryear{Dwork, Lynch, and Stockmeyer}{Dwork
  et~al\mbox{.}}{1988}]%
        {partialsync}
\bibfield{author}{\bibinfo{person}{Cynthia Dwork}, \bibinfo{person}{Nancy
  Lynch}, {and} \bibinfo{person}{Larry Stockmeyer}.}
  \bibinfo{year}{1988}\natexlab{}.
\newblock \showarticletitle{Consensus in the presence of partial synchrony}.
\newblock \bibinfo{journal}{{\em JACM\/}} \bibinfo{volume}{35},
  \bibinfo{number}{2} (\bibinfo{year}{1988}), \bibinfo{pages}{288--323}.
\newblock


\bibitem[\protect\citeauthoryear{Fischer, Lynch, and Paterson}{Fischer
  et~al\mbox{.}}{1985}]%
        {FLP85}
\bibfield{author}{\bibinfo{person}{Michael~J Fischer}, \bibinfo{person}{Nancy~A
  Lynch}, {and} \bibinfo{person}{Michael~S Paterson}.}
  \bibinfo{year}{1985}\natexlab{}.
\newblock \showarticletitle{Impossibility of Distributed Consensus with One
  Faulty Process}.
\newblock \bibinfo{journal}{{\em Journal of the Assccktion for Computing
  Machinery\/}} \bibinfo{volume}{32}, \bibinfo{number}{2}
  (\bibinfo{year}{1985}), \bibinfo{pages}{374--382}.
\newblock


\bibitem[\protect\citeauthoryear{Fitzi and Garay}{Fitzi and Garay}{2003}]%
        {fitzi2003efficient}
\bibfield{author}{\bibinfo{person}{Matthias Fitzi} {and}
  \bibinfo{person}{Juan~A Garay}.} \bibinfo{year}{2003}\natexlab{}.
\newblock \showarticletitle{Efficient player-optimal protocols for strong and
  differential consensus}. In \bibinfo{booktitle}{{\em Proceedings of the
  twenty-second annual symposium on Principles of distributed computing}}.
  \bibinfo{pages}{211--220}.
\newblock


\bibitem[\protect\citeauthoryear{Gao, Lu, Lu, Tang, Xu, and Zhang}{Gao
  et~al\mbox{.}}{}]%
        {gao2022dumbong}
\bibfield{author}{\bibinfo{person}{Yingzi Gao}, \bibinfo{person}{Yuan Lu},
  \bibinfo{person}{Zhenliang Lu}, \bibinfo{person}{Qiang Tang},
  \bibinfo{person}{Jing Xu}, {and} \bibinfo{person}{Zhenfeng Zhang}.}
\newblock \showarticletitle{Dumbo-NG: Fast Asynchronous BFT Consensus with
  Throughput-Oblivious Latency}. In \bibinfo{booktitle}{{\em Proc. CCS 2022}}.
\newblock


\bibitem[\protect\citeauthoryear{Gao, Lu, Lu, Tang, Xu, and Zhang}{Gao
  et~al\mbox{.}}{2022}]%
        {gao2021efficient}
\bibfield{author}{\bibinfo{person}{Yingzi Gao}, \bibinfo{person}{Yuan Lu},
  \bibinfo{person}{Zhenliang Lu}, \bibinfo{person}{Qiang Tang},
  \bibinfo{person}{Jing Xu}, {and} \bibinfo{person}{Zhenfeng Zhang}.}
  \bibinfo{year}{2022}\natexlab{}.
\newblock \showarticletitle{Efficient Asynchronous Byzantine Agreement without
  Private Setups}. In \bibinfo{booktitle}{{\em Proc. ICDCS 2022}}.
\newblock


\bibitem[\protect\citeauthoryear{Gelashvili, Kokoris-Kogias, Sonnino,
  Spiegelman, and Xiang}{Gelashvili et~al\mbox{.}}{2021a}]%
        {gelashvili2021jolteon}
\bibfield{author}{\bibinfo{person}{Rati Gelashvili}, \bibinfo{person}{Lefteris
  Kokoris-Kogias}, \bibinfo{person}{Alberto Sonnino},
  \bibinfo{person}{Alexander Spiegelman}, {and} \bibinfo{person}{Zhuolun
  Xiang}.} \bibinfo{year}{2021}\natexlab{a}.
\newblock \showarticletitle{Jolteon and Ditto: Network-Adaptive Efficient
  Consensus with Asynchronous Fallback}.
\newblock \bibinfo{journal}{{\em arXiv preprint arXiv:2106.10362\/}}
  (\bibinfo{year}{2021}).
\newblock


\bibitem[\protect\citeauthoryear{Gelashvili, Kokoris-Kogias, Spiegelman, and
  Xiang}{Gelashvili et~al\mbox{.}}{2021b}]%
        {gelashvili2021prepared}
\bibfield{author}{\bibinfo{person}{Rati Gelashvili}, \bibinfo{person}{Lefteris
  Kokoris-Kogias}, \bibinfo{person}{Alexander Spiegelman}, {and}
  \bibinfo{person}{Zhuolun Xiang}.} \bibinfo{year}{2021}\natexlab{b}.
\newblock \showarticletitle{Be Prepared When Network Goes Bad: An Asynchronous
  View-Change Protocol}.
\newblock \bibinfo{journal}{{\em arXiv preprint arXiv:2103.03181\/}}
  (\bibinfo{year}{2021}).
\newblock


\bibitem[\protect\citeauthoryear{Gennaro, Jarecki, Krawczyk, and Rabin}{Gennaro
  et~al\mbox{.}}{}]%
        {gennaro1999secure}
\bibfield{author}{\bibinfo{person}{Rosario Gennaro},
  \bibinfo{person}{Stanis{\l}aw Jarecki}, \bibinfo{person}{Hugo Krawczyk},
  {and} \bibinfo{person}{Tal Rabin}.}
\newblock \showarticletitle{Secure distributed key generation for discrete-log
  based cryptosystems}. In \bibinfo{booktitle}{{\em Proc. EUROCRYPT 1999}}.
  \bibinfo{pages}{295--310}.
\newblock


\bibitem[\protect\citeauthoryear{Guerraoui, Kne{\v{z}}evi{\'c}, Qu{\'e}ma, and
  Vukoli{\'c}}{Guerraoui et~al\mbox{.}}{2010}]%
        {700}
\bibfield{author}{\bibinfo{person}{Rachid Guerraoui}, \bibinfo{person}{Nikola
  Kne{\v{z}}evi{\'c}}, \bibinfo{person}{Vivien Qu{\'e}ma}, {and}
  \bibinfo{person}{Marko Vukoli{\'c}}.} \bibinfo{year}{2010}\natexlab{}.
\newblock \showarticletitle{The next 700 BFT protocols}. In
  \bibinfo{booktitle}{{\em Proceedings of the 5th European conference on
  Computer systems}}. \bibinfo{pages}{363--376}.
\newblock


\bibitem[\protect\citeauthoryear{Gueta, Abraham, Grossman, Malkhi, Pinkas,
  Reiter, Seredinschi, Tamir, and Tomescu}{Gueta et~al\mbox{.}}{}]%
        {sbft}
\bibfield{author}{\bibinfo{person}{Guy~Golan Gueta}, \bibinfo{person}{Ittai
  Abraham}, \bibinfo{person}{Shelly Grossman}, \bibinfo{person}{Dahlia Malkhi},
  \bibinfo{person}{Benny Pinkas}, \bibinfo{person}{Michael Reiter},
  \bibinfo{person}{Dragos-Adrian Seredinschi}, \bibinfo{person}{Orr Tamir},
  {and} \bibinfo{person}{Alin Tomescu}.}
\newblock \showarticletitle{SBFT: a Scalable and Decentralized Trust
  Infrastructure}. In \bibinfo{booktitle}{{\em Proc. DSN 2019}}.
  \bibinfo{pages}{568--580}.
\newblock


\bibitem[\protect\citeauthoryear{Guo, Lu, Lu, Tang, Xu, and Zhang}{Guo
  et~al\mbox{.}}{2022}]%
        {guo2022speeding}
\bibfield{author}{\bibinfo{person}{Bingyong Guo}, \bibinfo{person}{Yuan Lu},
  \bibinfo{person}{Zhenliang Lu}, \bibinfo{person}{Qiang Tang},
  \bibinfo{person}{Jing Xu}, {and} \bibinfo{person}{Zhenfeng Zhang}.}
  \bibinfo{year}{2022}\natexlab{}.
\newblock \showarticletitle{Speeding Dumbo: Pushing Asynchronous BFT Closer to
  Practice}. In \bibinfo{booktitle}{{\em The 29th Network and Distributed
  System Security Symposium (NDSS)}}.
\newblock


\bibitem[\protect\citeauthoryear{Guo, Lu, Tang, Xu, and Zhang}{Guo
  et~al\mbox{.}}{}]%
        {guo2020dumbo}
\bibfield{author}{\bibinfo{person}{Bingyong Guo}, \bibinfo{person}{Zhenliang
  Lu}, \bibinfo{person}{Qiang Tang}, \bibinfo{person}{Jing Xu}, {and}
  \bibinfo{person}{Zhenfeng Zhang}.}
\newblock \showarticletitle{Dumbo: Faster asynchronous bft protocols}. In
  \bibinfo{booktitle}{{\em Proc. CCS 2020}}. \bibinfo{pages}{803--818}.
\newblock


\bibitem[\protect\citeauthoryear{Kate and Goldberg}{Kate and Goldberg}{}]%
        {kate2009distributed}
\bibfield{author}{\bibinfo{person}{Aniket Kate} {and} \bibinfo{person}{Ian
  Goldberg}.}
\newblock \showarticletitle{Distributed key generation for the internet}. In
  \bibinfo{booktitle}{{\em Proc. ICDCS 2009}}. \bibinfo{pages}{119--128}.
\newblock


\bibitem[\protect\citeauthoryear{Keidar, Kokoris-Kogias, Naor, and
  Spiegelman}{Keidar et~al\mbox{.}}{2021}]%
        {dag}
\bibfield{author}{\bibinfo{person}{Idit Keidar}, \bibinfo{person}{Eleftherios
  Kokoris-Kogias}, \bibinfo{person}{Oded Naor}, {and}
  \bibinfo{person}{Alexander Spiegelman}.} \bibinfo{year}{2021}\natexlab{}.
\newblock \showarticletitle{All you need is dag}.
\newblock \bibinfo{journal}{{\em arXiv preprint arXiv:2102.08325\/}}
  (\bibinfo{year}{2021}).
\newblock


\bibitem[\protect\citeauthoryear{Kokoris~Kogias, Malkhi, and
  Spiegelman}{Kokoris~Kogias et~al\mbox{.}}{}]%
        {kokoris2020asynchronous}
\bibfield{author}{\bibinfo{person}{Eleftherios Kokoris~Kogias},
  \bibinfo{person}{Dahlia Malkhi}, {and} \bibinfo{person}{Alexander
  Spiegelman}.}
\newblock \showarticletitle{Asynchronous Distributed Key Generation for
  Computationally-Secure Randomness, Consensus, and Threshold Signatures.}. In
  \bibinfo{booktitle}{{\em Proc. CCS 2020}}. \bibinfo{pages}{1751--1767}.
\newblock


\bibitem[\protect\citeauthoryear{Kursawe and Shoup}{Kursawe and Shoup}{2002}]%
        {kursawe05}
\bibfield{author}{\bibinfo{person}{Klaus Kursawe} {and} \bibinfo{person}{Victor
  Shoup}.} \bibinfo{year}{first announced in 2002}\natexlab{}.
\newblock \showarticletitle{Optimistic asynchronous atomic broadcast}. In
  \bibinfo{booktitle}{{\em Proc. ICALP 2005}}. \bibinfo{pages}{204--215}.
\newblock


\bibitem[\protect\citeauthoryear{Loss and Moran}{Loss and Moran}{2018}]%
        {loss2018combining}
\bibfield{author}{\bibinfo{person}{Julian Loss} {and} \bibinfo{person}{Tal
  Moran}.} \bibinfo{year}{2018}\natexlab{}.
\newblock \showarticletitle{Combining Asynchronous and Synchronous Byzantine
  Agreement: The Best of Both Worlds.}
\newblock \bibinfo{journal}{{\em IACR Cryptol. ePrint Arch.\/}}
  \bibinfo{volume}{2018} (\bibinfo{year}{2018}), \bibinfo{pages}{235}.
\newblock


\bibitem[\protect\citeauthoryear{Lu, Lu, Tang, and Wang}{Lu et~al\mbox{.}}{}]%
        {lu2020dumbo}
\bibfield{author}{\bibinfo{person}{Yuan Lu}, \bibinfo{person}{Zhenliang Lu},
  \bibinfo{person}{Qiang Tang}, {and} \bibinfo{person}{Guiling Wang}.}
\newblock \showarticletitle{Dumbo-mvba: Optimal multi-valued validated
  asynchronous byzantine agreement, revisited}. In \bibinfo{booktitle}{{\em
  Proc. PODC 2020}}. \bibinfo{pages}{129--138}.
\newblock


\bibitem[\protect\citeauthoryear{Miller, Xia, Croman, Shi, and Song}{Miller
  et~al\mbox{.}}{}]%
        {honeybadger}
\bibfield{author}{\bibinfo{person}{Andrew Miller}, \bibinfo{person}{Yu Xia},
  \bibinfo{person}{Kyle Croman}, \bibinfo{person}{Elaine Shi}, {and}
  \bibinfo{person}{Dawn Song}.}
\newblock \showarticletitle{The honey badger of BFT protocols}. In
  \bibinfo{booktitle}{{\em Proc. CCS 2016}}. \bibinfo{pages}{31--42}.
\newblock


\bibitem[\protect\citeauthoryear{Momose, Cruz, and Kaji}{Momose
  et~al\mbox{.}}{2020}]%
        {momose2020hybrid}
\bibfield{author}{\bibinfo{person}{Atsuki Momose}, \bibinfo{person}{Jason~Paul
  Cruz}, {and} \bibinfo{person}{Yuichi Kaji}.} \bibinfo{year}{2020}\natexlab{}.
\newblock \showarticletitle{Hybrid-BFT: Optimistically Responsive Synchronous
  Consensus with Optimal Latency or Resilience.}
\newblock \bibinfo{journal}{{\em IACR Cryptol. ePrint Arch.\/}}
  \bibinfo{volume}{2020} (\bibinfo{year}{2020}), \bibinfo{pages}{406}.
\newblock


\bibitem[\protect\citeauthoryear{Momose and Ren}{Momose and Ren}{}]%
        {momosemulti}
\bibfield{author}{\bibinfo{person}{Atsuki Momose} {and} \bibinfo{person}{Ling
  Ren}.}
\newblock \showarticletitle{Multi-Threshold Byzantine Fault Tolerance}. In
  \bibinfo{booktitle}{{\em Proc. CCS 2021}}.
\newblock


\bibitem[\protect\citeauthoryear{Mostefaoui, Moumen, and Raynal}{Mostefaoui
  et~al\mbox{.}}{}]%
        {MMR15}
\bibfield{author}{\bibinfo{person}{Achour Mostefaoui}, \bibinfo{person}{Hamouma
  Moumen}, {and} \bibinfo{person}{Michel Raynal}.}
\newblock \showarticletitle{Signature-free asynchronous byzantine consensus
  with $t< n/3$ and $\mathcal{O}(n^2)$ messages}. In \bibinfo{booktitle}{{\em
  Proc. PODC 2014}}. \bibinfo{pages}{2--9}.
\newblock


\bibitem[\protect\citeauthoryear{Nakamoto}{Nakamoto}{2008}]%
        {bitcoin}
\bibfield{author}{\bibinfo{person}{Satoshi Nakamoto}.}
  \bibinfo{year}{2008}\natexlab{}.
\newblock \showarticletitle{Bitcoin: A peer-to-peer electronic cash system}.
\newblock  (\bibinfo{year}{2008}).
\newblock


\bibitem[\protect\citeauthoryear{Pass and Shi}{Pass and Shi}{2017}]%
        {pass2017sleepy}
\bibfield{author}{\bibinfo{person}{Rafael Pass} {and} \bibinfo{person}{Elaine
  Shi}.} \bibinfo{year}{2017}\natexlab{}.
\newblock \showarticletitle{The sleepy model of consensus}. In
  \bibinfo{booktitle}{{\em Advances in Cryptology -- ASIACRYPT 2017}}.
  \bibinfo{pages}{380--409}.
\newblock


\bibitem[\protect\citeauthoryear{Patra}{Patra}{}]%
        {patra2011error}
\bibfield{author}{\bibinfo{person}{Arpita Patra}.}
\newblock \showarticletitle{Error-free multi-valued broadcast and Byzantine
  agreement with optimal communication complexity}. In \bibinfo{booktitle}{{\em
  Proc. OPODIS 2011}}. \bibinfo{pages}{34--49}.
\newblock


\bibitem[\protect\citeauthoryear{Patra, Choudhary, and Pandu~Rangan}{Patra
  et~al\mbox{.}}{}]%
        {patra2009simple}
\bibfield{author}{\bibinfo{person}{Arpita Patra}, \bibinfo{person}{Ashish
  Choudhary}, {and} \bibinfo{person}{Chandrasekharan Pandu~Rangan}.}
\newblock \showarticletitle{Simple and efficient asynchronous byzantine
  agreement with optimal resilience}. In \bibinfo{booktitle}{{\em Proc. PODC
  2009}}. \bibinfo{pages}{92--101}.
\newblock


\bibitem[\protect\citeauthoryear{Pedersen}{Pedersen}{}]%
        {pedersen1991threshold}
\bibfield{author}{\bibinfo{person}{Torben~Pryds Pedersen}.}
\newblock \showarticletitle{A Threshold Cryptosystem without a Trusted Party}.
  In \bibinfo{booktitle}{{\em Proc. EUROCRYPT 1991}}.
  \bibinfo{pages}{522--526}.
\newblock


\bibitem[\protect\citeauthoryear{{R. Pass, and E. Shi}}{{R. Pass, and E.
  Shi}}{2018}]%
        {thunderella}
\bibfield{author}{\bibinfo{person}{{R. Pass, and E. Shi}}.}
  \bibinfo{year}{2018}\natexlab{}.
\newblock \showarticletitle{Thunderella: Blockchains with optimistic instant
  confirmation}. In \bibinfo{booktitle}{{\em Proc. EUROCRYPT 2018}}.
  \bibinfo{pages}{3--33}.
\newblock


\bibitem[\protect\citeauthoryear{Rabin}{Rabin}{1983}]%
        {rabin1983randomized}
\bibfield{author}{\bibinfo{person}{Michael~O Rabin}.}
  \bibinfo{year}{1983}\natexlab{}.
\newblock \showarticletitle{Randomized byzantine generals}. In
  \bibinfo{booktitle}{{\em 24th Annual Symposium on Foundations of Computer
  Science}}. IEEE, \bibinfo{pages}{403--409}.
\newblock


\bibitem[\protect\citeauthoryear{Ramasamy and Cachin}{Ramasamy and Cachin}{}]%
        {cachin05}
\bibfield{author}{\bibinfo{person}{HariGovind~V Ramasamy} {and}
  \bibinfo{person}{Christian Cachin}.}
\newblock \showarticletitle{Parsimonious asynchronous byzantine-fault-tolerant
  atomic broadcast}. In \bibinfo{booktitle}{{\em Proc. OPODIS 2005}}.
  \bibinfo{pages}{88--102}.
\newblock


\bibitem[\protect\citeauthoryear{Saad, Anwar, Ravi, and Mohaisen}{Saad
  et~al\mbox{.}}{2021}]%
        {saad2021revisiting}
\bibfield{author}{\bibinfo{person}{Muhammad Saad}, \bibinfo{person}{Afsah
  Anwar}, \bibinfo{person}{Srivatsan Ravi}, {and} \bibinfo{person}{David
  Mohaisen}.} \bibinfo{year}{2021}\natexlab{}.
\newblock \showarticletitle{Revisiting Nakamoto Consensus in Asynchronous
  Networks: A Comprehensive Analysis of Bitcoin Safety and ChainQuality}. In
  \bibinfo{booktitle}{{\em Proceedings of the 2021 ACM SIGSAC Conference on
  Computer and Communications Security}}. \bibinfo{pages}{988--1005}.
\newblock


\bibitem[\protect\citeauthoryear{Shrestha, Abraham, Ren, and Nayak}{Shrestha
  et~al\mbox{.}}{}]%
        {shrestha2020optimality}
\bibfield{author}{\bibinfo{person}{Nibesh Shrestha}, \bibinfo{person}{Ittai
  Abraham}, \bibinfo{person}{Ling Ren}, {and} \bibinfo{person}{Kartik Nayak}.}
\newblock \showarticletitle{On the Optimality of Optimistic Responsiveness}. In
  \bibinfo{booktitle}{{\em Proc. CCS 2020}}. \bibinfo{pages}{839--857}.
\newblock


\bibitem[\protect\citeauthoryear{Spiegelman}{Spiegelman}{2021}]%
        {spiegelman2020search}
\bibfield{author}{\bibinfo{person}{Alexander Spiegelman}.}
  \bibinfo{year}{2021}\natexlab{}.
\newblock \showarticletitle{In Search for an Optimal Authenticated Byzantine
  Agreement}. In \bibinfo{booktitle}{{\em Proc. DISC 2021}}.
\newblock


\bibitem[\protect\citeauthoryear{Veronese, Correia, Bessani, and Lung}{Veronese
  et~al\mbox{.}}{}]%
        {veronese2009spin}
\bibfield{author}{\bibinfo{person}{Giuliana~Santos Veronese},
  \bibinfo{person}{Miguel Correia}, \bibinfo{person}{Alysson~Neves Bessani},
  {and} \bibinfo{person}{Lau~Cheuk Lung}.}
\newblock \showarticletitle{Spin one's wheels? Byzantine fault tolerance with a
  spinning primary}. In \bibinfo{booktitle}{{\em Proc. SRDS 2009}}.
  \bibinfo{pages}{135--144}.
\newblock


\bibitem[\protect\citeauthoryear{Yang, Park, Alizadeh, Kannan, and Tse}{Yang
  et~al\mbox{.}}{2022}]%
        {yang2021dispersedledger}
\bibfield{author}{\bibinfo{person}{Lei Yang}, \bibinfo{person}{Seo~Jin Park},
  \bibinfo{person}{Mohammad Alizadeh}, \bibinfo{person}{Sreeram Kannan}, {and}
  \bibinfo{person}{David Tse}.} \bibinfo{year}{2022}\natexlab{}.
\newblock \showarticletitle{{DispersedLedger}: {High-Throughput} Byzantine
  Consensus on Variable Bandwidth Networks}. In \bibinfo{booktitle}{{\em 19th
  USENIX Symposium on Networked Systems Design and Implementation (NSDI 22)}}.
\newblock


\bibitem[\protect\citeauthoryear{Yin, Malkhi, Reiter, Gueta, and Abraham}{Yin
  et~al\mbox{.}}{2019}]%
        {yin2018hotstuff-full}
\bibfield{author}{\bibinfo{person}{Maofan Yin}, \bibinfo{person}{Dahlia
  Malkhi}, \bibinfo{person}{Michael~K Reiter}, \bibinfo{person}{Guy~Golan
  Gueta}, {and} \bibinfo{person}{Ittai Abraham}.}
  \bibinfo{year}{2019}\natexlab{}.
\newblock \showarticletitle{Hotstuff: Bft consensus with linearity and
  responsiveness}. In \bibinfo{booktitle}{{\em Proc. PODC 2019}}.
  \bibinfo{pages}{347--356}.
\newblock


\bibitem[\protect\citeauthoryear{Zhang and Duan}{Zhang and Duan}{}]%
        {zhang2022pace}
\bibfield{author}{\bibinfo{person}{Haibing Zhang} {and} \bibinfo{person}{Sisi
  Duan}.}
\newblock \showarticletitle{PACE: Fully Parallelizable BFT from Reproposable
  Byzantine Agreement}. In \bibinfo{booktitle}{{\em Proc. CCS 2022}}.
\newblock


\end{thebibliography}






  

  %
  


\section{Deferred   Proofs for $\tobc$ Constructions}\label{append:bolt}

\begin{lemma}	
	The algorithm in Figure \ref{fig:multicast} satisfies the {\em total-order}, {\em notarizability}, {\em abandonability} and {\em optimistic liveness} properties   of $\tobc$  except with negligible probability.
	\label{abc1}
\end{lemma}

{\em Proof}: Here we prove the four properties one by one:

{\em For total-order}: First, we prove at same position, for any two honest parties $\node_i$ and $\node_j$ return $\block_i$ and $\block_j$, respectively, then $\block_i = \block_j$. It is clear that if the honest party $\node_i$ outputs $\block_i$, then at least $f+1$ honest parties did vote for $\block_i$ because $\TSIG.\Vrfy_{2f+1}$ passes verification. So did  $f+1$ honest parties vote for $\block_j$. 
That means at least one honest party votes for both blocks, so $\block_i = \block_j$.

{\em For notarizability}: Suppose a party $\node_i$ outputs   $\blocks[j]:=\langle \id, j, $ $\payload_j, \prf_j \rangle$, it means at least $f+1$ honest parties   vote  for $\block[j]$, according to  the pseudocode,  at least those same $f+1$ honest parties already output $\blocks[j-1]$ and  received the $\payload_j$, hence,  those honest parties can further use the valid $\prf_j$ to extract $\blocks[j]$ from the receivied protocol messages.

{\em For abandonability}: it is immediate to see from the pseudocode of the abandon interface.

{\em For optimistic liveness}: suppose that the optimistic condition is that the leader is honest, then any honest party would output $\block[1]$ in three asynchronous rounds after entering the protocol
and   would output     $\out[j+1]$ within two asynchronous rounds after outputting $\out[j]$ (for all $j\ge1$).
$\hfill\square$

\begin{lemma}	
	The algorithm in Figure \ref{fig:prbc} satisfies the {\em total-order}, {\em notarizability}, {\em abandonability} and {\em optimistic liveness} properties  of $\tobc$  except with negligible probability.
	\label{abc2}
\end{lemma}

{\em Proof:} It is clear that the {\em total-order}, {\em notarizability} and {\em abandonability} follow immediately from the properties of $\rbc$ and the pseudocode of $\mathsf{Bolt}$-$\mathsf{sRBC}$, since the {agreement} of $\rbc$ guarantees that the output $\payload$ by any parties is the same and a valid proof along with the sequentially executing nature of all $\rbc$ instances would ensure  total-order and notarizability. {\em For optimistic liveness}, the optimistic condition remains to be that the leader is honest, and  $\kappa$ is  4 due to the $\rbc$ construction in \cite{honeybadger} and an extra vote step.
$\hfill\square$


\section{Deferred   Proofs for    $\tcvba$}\label{append:tcvba}

\begin{lemma}	
	The algorithm in Figure \ref{fig:tcvba} satisfies the {\em termination}, {\em validity} and {\em agreement} properties  of $\tcvba$  except with negligible probability.
	\label{tcvba1}
\end{lemma}

{\em Proof:} First, from \cite{MMR15} we know: for any $v\in S_r$, then $v$ was the input of at least one honest party, then in next round $r+1$, every honest party's input  $\est_{r+1}$ will always from at least one honest party's input of round $r$ by the code. Again, according to the pseudocode, the output is the element of $S$, hence, validity is satisfied.

Second, suppose one honest party $\node_{i}$ is the first party to output and $\node_{i}$ outputs $v$ in round $r$. Then for any other honest parties,  either output the same $v$, or have $S_r=\{v,v+1\}$, hence, all honest parties will have same input $\est_{r+1}=v$ ($v\%2=c_r\%2$) in next round $r+1$, then, $S_{r+1}=\{v\}$, and $\est_{r+2}=v$ for round $r+2$. Once in some round $r'$, the  $v\%2 = c_{r'}\%2$, all honest parties output the same $v$. So the agreement is met.

Third, the termination analysis is similar to that in \cite{ababug,MMR15}.
$\hfill\square$

\begin{lemma}	
	The algorithm in Figure \ref{fig:blackaba} satisfies the {\em termination}, {\em validity} and {\em agreement} properties  of $\tcvba$  except with negligible probability.
	\label{tcvba2}
\end{lemma}

{\em Proof}: 
{\em For termination}: Since all honest parties input a value in $\{v, v+1\}$ where $v \in \N$, without loss of generality, suppose value $v$ was input by at least $f+1$ honest parties, then after amplifying, every honest parties can receive $2f+1$ $\val(\id,v)$ messages from distinct parties carrying the same $v$. Hence, all honest parties can activate $\ABA$ with one input $v'\%2$. Then, the $\ABA$ guarantees that all honest parties return $b$. Since the validity of $\ABA$ guarantees the output of $\ABA$ is at least one honest party's input, according to the code, if one honest party input $b$ into $\ABA$, then the party has received at least $2f+1$ $\val(\id,v')$ messages from distinct parties carrying the same $v'$ and $v'\%2 = b$, hence, all honest parties can receive $f+1$ $\val(\id,v')$ messages from distinct honest parties containing the same $v'$ such that $v'\%2=b$.

{\em For validity}: Since the validity of $\ABA$ guarantees the output $b$ of $\ABA$ is at least one honest party's input, then according to the code, the party has received at least $2f+1$ $\val(\id,v')$ messages from distinct parties carrying the same $v'$, where $v'\%2 = b$, hence, at least one honest party with taking $v'$ as input and multicast $\val(\id,v')$.

{\em For agreement}: Since the agreement of $\ABA$ guarantees all honest parties have the same output $b$. Hence, all honest parties will output value $v'$, where $v'\%2=b$. Without loss of generality, suppose honest party $\node_{i}$ output $v$ and honest party $\node_j$ output $v+2k$ ($k\neq 0$), then following the validity proof, both $v$ and $v+2k$ are honest party's input, then it is a contradiction with the $\tcvba$ input assumption.
$\hfill\square$


\section{Deferred   Proofs for    $\sys$}\label{append:bdt}


\noindent
{\bf Safety proof}. We first prove the  total-order and agreement, assuming the underlying $\tobc$, $\tcvba$ and $\acs$ are secure.

\begin{claim}
	\label{aa}
	If an honest party activates $\tcvba[e]$,   then at least $n-2f$ honest parties have already invoked $\abandon(e)$, and from now on: suppose these same parties invoke $\abandon(e)$ before they output $\block$:=$ \langle e, R, \cdot,\cdot \rangle $, then	any party (including the faulty ones) cannot receive (or forge) a valid $\block$:=$ \langle e, R+1, \cdot,\cdot  \rangle $, and all honest parties  would activate $ \tcvba[e]$. 
	
\end{claim}

{\em Proof:} When  an honest party $\node_i$ activates $ \tcvba[e]$, it must have received $n-f$ valid $\view$ messages from distinct parties, so there would be at least $n-2f$ honest parties multicast $\view$ messages. By the pseudocode, it also means that  at least $n-2f$ honest parties have invoked $\abandon(e)$. Note that $n-2f\ge f+1$ and these same parties invoke $\abandon(e)$ before they output $\block$:=$ \langle e, R, \cdot,\cdot  \rangle $, so no party would deliver any valid $\block$:=$ \langle e, R+1, \cdot,\cdot  \rangle $ in this epoch's $\tobc$ phase due to the {\em abandonability} property of $\tobc$. It is also implies that any parties cannot from the $\tobc$ receive valid  $\block$:=$ \langle e, R+1, \cdot,\cdot  \rangle $ , then all honest parties will eventually be interrupted by the ``timeout'' mechanism after $\tau$ asynchronous rounds and then multicast $\view$ messages. This   ensures that all honest parties finally receive $n-f$ valid $\view$ messages from distinct parties, causing all honest parties to activate $ \tcvba[e]$. 
$\hfill\square$

\begin{claim}\label{dd}
	Suppose that  some party receives a valid $\block$:=$\langle e, R, \cdot , \cdot\rangle$ from $\tobc[e]$ when an honest party invokes $ \tcvba[e]$ s.t. this $\block$ is the one with largest slot number among all parties' valid $\notarized$ blocks (which means the union of the honest parties' actual $\notarized$ blocks and the malicious parties' arbitrary  valid  $\tobc[e]$ block), 	then all honest parties'     $\maxview_e$ must be either $R$ or $R-1$.
\end{claim}	

{\em Proof:}
Following   Claim \ref{aa}, once an honest party invokes $ \tcvba[e]$, 
the $\tobc[e]$ $\block$ with the largest slot number $R_{max}$ would not change anymore.
Let us call this already fixed $\tobc[e]$ $\block$ with highest slot number as $\block_{max}$.
Since there is someone that receives $\block_{max}$:=$\langle e, R, \cdot , \cdot\rangle$,
at least $f+1$ honest parties (e.g., denoted by $Q$) have already received the block $\langle e, R-1, \cdot , \cdot\rangle$,
which is because of the {\em notarizability} property of $\tobc$. 
So these honest parties would broadcast a valid $\view(e, R-1,\cdot)$ message or a valid $\view(e, R,\cdot)$ message. 
According to the pseudocode in Figure \ref{fig:mule}, $\maxview_e$ is the maximum number in the set of $\mathsf{Paces}_e$, where       $\mathsf{Paces}_e$ contains the slot numbers encapsulated by $n-f$ valid $\view$ messages. Therefore, $\mathsf{Paces}_e$ must contain one $\view$ message's slot number from at least $n-2f \ge f+1$ honest parties (e.g., denoted by $\bar{Q}$). 
All honest parties' local $\mathsf{Paces}_e$ set must   contain $R-1$ and/or $R$, because $\bar{Q}$ and $Q$ contain at least one common honest party. Moreover, there is no valid $\view$ message containing any slot larger than $R$ since the proof for that is unforgeable, which means $R$ is the largest possible value in all honest parties' $\mathsf{Paces}_e$. So any honest party's $\maxview_e$ must be     $R$ or $R-1$.
$\hfill\square$

\begin{claim}\label{norevoke}
	No honest party would get a $\mathsf{syncPace}_e$ smaller than the slot number of it latest finalized block $\out[-1]$ (i.e., no block  finalized in some honest party's $\out$ can   be revoked).
\end{claim}
{\em Proof:} 
Suppose an honest party invokes $ \tcvba[e]$ and a valid $\tobc[e]$  block $\langle e, R, \cdot , \cdot\rangle$ is the one with largest slot number among  the union of the honest parties' actual $\notarized$ blocks and the malicious parties' arbitrary  valid  $\tobc[e]$ block. Because of Claim \ref{dd}, all honest parties will activate $\tcvba[e]$ with taking either $R$ or $R-1$ as input. According to the {\em strong validity} of  $\tcvba$, the output $\mathsf{syncPace}_e$ of $\tcvba[e]$  must be either $R$ or $R-1$. 
Then we consider the next two cases:

\begin{enumerate}
	\item Only malicious parties have this  $\langle e, R, \cdot , \cdot\rangle$ block;
	
	\item Some honest party $\node_{i}$ also has the $\langle e, R, \cdot , \cdot\rangle$ block.
\end{enumerate}	

{\em For Case 1)} Due to the {\em notarizability} property of $\tobc$ and this case's baseline, there exist $f+1$ honest parties (denoted by a set   $Q$) have  the block $\langle e, R-1, \cdot , \cdot\rangle$ as their local $\notarized$. 
Note that remaining honest parties  (denoted by a set   $\bar{Q}$) would have local $\notarized$ block not higher than $R-1$.  According to the algorithm in Figure \ref{fig:mule}, we can state that: 
(i) if the output is $R$, then all honest parties will sync their $\out$ up to the block  $\langle e, R, \cdot , \cdot\rangle$ (which include all honest parties' local $\notarized$); 
(ii) similarly, if the output is $R-1$, all honest parties will sync up till $\langle e, R-1, \cdot , \cdot\rangle$ (which also include all honest parties' local $\notarized$). So in this case, all honest parties (i.e., $\bar{Q} \cup Q$) will not discard their $\notarized$ $\block$, let alone discard some $\tobc$ that are already finalized to output into $\out$.

{\em For Case 2)} Let $Q$ denote the set of honest parties that have  the block $\langle e, R, \cdot , \cdot\rangle$ as their local $\notarized$. 
Note the remaining honest parties $\bar{Q}$ would have the $\notarized$ block not higher than $R$. In this case,  following the algorithm in Figure \ref{fig:mule}, we can see that: 
(i) if the output is $R$, then all honest parties will sync their $\out$ up to the block  $\langle e, R, \cdot , \cdot\rangle$ (which include all honest parties' local $\notarized$); 
(ii) similarly, if the output is $R-1$, all honest parties will sync up to $\langle e, R-1, \cdot , \cdot\rangle$ (which include $\bar{Q}$ parties' local $\notarized$ and ${Q}$ parties' finalized output $\out$). So in this case, all honest parties (i.e., $\bar{Q} \cup Q$) will not discard any block in their finalized  $\out$.
$\hfill\square$

\begin{claim}
	If $\tcvba[e]$ returns $\mathsf{syncPace}_e$, then at least $f+1$ honest parties  can append $\block$s with slot numbers from $1$ to $\mathsf{syncPace}_e$ that all received from $\tobc[e]$  into the $\out$ without invoking $\gethelp$  function.
	\label{provable}
\end{claim}	

{\em Proof:}
Suppose $\tcvba[e]$ returns $\mathsf{syncPace}_e$, then from the {\em strong validity}  of $\tcvba$, at least one honest party inputs the number $\mathsf{syncPace}_e$. The same honest party must receive a valid message 	$\view(e, \mathsf{syncPace}_e, \prf)$, 
which means there must exists a valid   $\tobc$ block  $\langle e, \mathsf{syncPace}_e, \cdot , \prf \rangle$. By the code, the honest party will multicst $\view(e, \mathsf{syncPace}_e, \prf)$ if $\tcvba[e]$ returns $\mathsf{syncPace}_e$, then all honest parties can get the $\prf$. Following the {\em notarizability} and {\em total-order} properties of $\tobc$, at least $f+1$ honest parties  can append  $\block$s  from $\langle e, 1, \cdot , \cdot\rangle$ to $\langle e, \mathsf{syncPace}_e, \cdot , \cdot\rangle$ into the $\out$ without invoking $\gethelp$  function.
$\hfill\square$

\begin{claim}
	\label{callhelp}
	If an honest party invokes $\gethelp$  function to retrieve a block $\out[i]$, it eventually can get it;
	if another honest party   retrieves a block $\out[i]'$   at the same log position $i$ from the $\gethelp$ function, then $\out[i]= \out[i]'$.
\end{claim}	

{\em Proof:} Due to Claim \ref{provable} and {\em total-order} properties of $\tobc$, any block $\out[i]$ that an honest party is retrieving through $\gethelp$ function shall have been in the output $\log$ of at least $f+1$ honest parties, so it eventually can get $f+1$ correct $\hlp$ messages with the same Merkle tree root $h$ from distinct parties, then it can interpolate the $f+1$ leaves to reconstruct $\out[i]$ which is same to other honest parties' local $\out[i]$. 
We can argue the agreement by contradiction, in case the interpolation of honest party  fails or it recovers a block $\out'[i]$ different from the the honest party's local $\out[i]$, that means the Merkle tree with root $h$ commits some leaves that are not coded fragments of  $\out[i]$; nevertheless, there is at least one honest party encode $\out[i]$ and commits the block's erasure code to have a Merkle tree root $h$; so the adversary indeed breaks the collision resistance of Merkle tree, implying the break of the collision-resistance of hash function, which is computationally infeasible. So all honest parties that attempt to retrieve a missing block $\out[i]$ must   fetch the block consistent to other honest parties'.
$\hfill\square$

\begin{lemma}
	\label{agreement}
	If all honest parties enter the  $\tobc$ phase with the same $\out$, then they will always finish the  $\trans$ phase with still having the same   $\out'$. 
\end{lemma}

{\em Proof:} 
If all honest parties enter the epoch with the same $\out$, it is easy to see that they all will eventually   interrupt  to abandon the $\tobc$ phase.	
Due to Claim \ref{aa},  all honest parties  would activate $\tcvba[e]$. Following the  {\em agreement and termination} of $\tcvba[e]$, all parties would finish $\tcvba[e]$ to get a common $\mathsf{syncPace}_e$,  and then by the pseudocode, all honest parties will sync up to the same $\out$, and the last block of $\out$ with slot number $\mathsf{syncPace}_e$ (due to {\em total-order} properties of $\tobc$, Claim \ref{provable} and Claim \ref{callhelp}), hence, all parties will finish the  $\trans$ phase with   the same output $\out$.
$\hfill\square$

\begin{lemma}
	\label{order}
	For any two honest parties before finishing the  Transformer phase, then 
	there exists one party, such that its $\out$ is a prefix of (or equal to) the other's.
\end{lemma}

{\em Proof:} The blocks outputted before the completion of the  Transformer phase were originally generated from the $\tobc$ phase, then the Lemma holds immediately by following the {\em total-order} property of $\tobc$ and Claim \ref{norevoke}.
$\hfill\square$

\begin{lemma}
	\label{dumbo}
	If any honest party enters the  $\Pessimistic$ phase, then all honest parties will enters the phase and
	always leave the  phase with having the same   $\out$.
\end{lemma}

{\em Proof:} If any honest party enter the $\Pessimistic$ phase, all honest parties would enter this phase, which is due to Claim \ref{aa}, the {\em agreement and termination} property of $\tcvba$ and $\mathsf{syncPace}_e=0$.
Let us assume that all honest parties enter the $\Pessimistic$ phase with the same $\out$, it would be trivial too see the statement for the {\em agreement and termination} properties of $\acs$ and the {\em correct and robustness} properties of threshold public key encryption.
Then considering Lemma \ref{agreement} and the simple fact that all honest parties activate with the common empty $\out$,
we can inductively reason that all honest parties must enter any $\Pessimistic$ phase with the same $\out$. So the Lemma holds.
$\hfill\square$

\begin{lemma}
	\label{orderr}
	For any two honest parties   in the same epoch,  
	there exists one party, such that its $\out$ is a prefix of (or equal to) the other's. 
\end{lemma}

{\em Proof:} 
If two honest parties do not enter the $\Pessimistic$ phase during the epoch, 
both of them only participate in $\tobc$ or $\trans$, so this Lemma holds immediately by following Lemma \ref{order}.
For two honest parties that one enters the $\Pessimistic$ phase and one does not, this Lemma holds because the latter one's $\out$ is either a prefix of the former one's or equal to the former one's due to Lemma  \ref{agreement} and \ref{order}.
For the remaining case that both   honest parties   enter the $\Pessimistic$ phase, they must   initially have exactly same $\out$ (due to Lemma \ref{agreement}).
Moreover, in the phase, all honest parties would execute the $\acs$ instances in a sequential manner (e.g., there is only one $\acs$ instance in our exemplary pseudocode), 
so every honest party would output in one ACS instance only if it has already outputted in all earlier ACS instances.
Besides,  any two honest party would output the same transaction batch in every $\acs$ instance for the agreement property of $\acs$.
So this Lemma also holds for any two honest parties that are staying in the same epoch.
$\hfill\square$

\begin{theorem}
	\label{key1}
	The $\sys$ protocol satisfies the {\em agreement} and {\em total order} properties.
\end{theorem}	

{\em Proof:}
The {total order} be induced by Lemma \ref{orderr} along with the fact the protocol is executed epoch by epoch. 
The agreement  follows immediately from   Lemma \ref{agreement}  and \ref{dumbo} along with the fact that all honest parties initialize with the same empty $\out$ to enter the first epoch's $\tobc$ phase.
$\hfill\square$

\medskip
\noindent
{\bf Liveness proof}. Then we prove the liveness property of $\BDT$.

\begin{lemma}
	\label{boltlive} If all honest parties enter the $\tobc$ phase, once the liveness failed, then they will leave the phase in at most polynomial number of asynchronous rounds and also all enter the $\trans$ phase.
\end{lemma}

{\em Proof:} The liveness failed in the $\tobc$ phase, it could be either (1). no progress within $\tau$ time or (2). some oldest transactions is not output within $T$ time.
For (1), at worst case, all honest parties' timeout  will  interrupt, causing them to abandon the $\tobc$ phase in at most $\bigO(\tau)$ asynchronous rounds.
For (2), it will take at most  $\bigO(T)$ asynchronous rounds to leave the $\tobc$ phase if there is a suspiciously censored $\tx$ in all honest parties' buffers due to some timeout parameter $T$. Hence, once the liveness failed, all honest parties will leave the phase in at most $\bigO(\tau+T)$  asynchronous rounds. 
After that, the broadcast of $\view$ message will take one more asynchronous round. After that, all honest parties would receive enough $\view$ messages to enter the $\trans$ phase, which costs at most $\bigO(\tau+T+1)$ asynchronous rounds. 
$\hfill\square$

\begin{lemma}
	\label{translive}  If all honest parties enter the $\trans$ phase, they all leave the phase in expected constant asynchronous rounds and then either enter the $\Pessimistic$ phase or enter the next epoch's $\tobc$ phase.
\end{lemma}

{\em Proof:} 
If all honest parties enter the $\trans$ phase, it
is trivial to see the Lemma since the underlying $\tcvba$ terminates in on-average constant asynchronous rounds. If the output of $\tcvba$ equal 0, then enter the $\Pessimistic$ phase, otherwise,   enter the next epoch's $\tobc$ phase.
$\hfill\square$

\begin{lemma}
	\label{dumbolive} If all honest parties enter the $\Pessimistic$ phase, all honest parties will leave this $\Pessimistic$ phase in on-average constant asynchronous rounds with outputting some blocks containing on-average $\bigO(B)$-sized transactions. 
\end{lemma}

{\em Proof: } 
 Similar to \cite{honeybadger}'s analysis, $\Pessimistic$ phase at least outputs $\bigO(B)$-sized transactions (without worrying that the adversary can learn any bit about the transactions to be outputted before they are actually finalized as output) for each execution. Here we remark that the original analysis in \cite{honeybadger} only requires IND-CPA security of threshold  public key encryption might be not enough, since we need to simulate that the adversary can query decryption oracle by inserting her ciphertext into the ACS output. Moreover, the underlying $\dumbo$ ACS construction \cite{guo2020dumbo}   ensures all parties to leave the phase in on-average constant asynchronous rounds.
$\hfill\square$

\begin{theorem}
	\label{key2}	
	The $\sys$ protocol satisfies the {\em liveness} property.
\end{theorem}

{\em Proof:}
Due to Lemma \ref{boltlive} and \ref{translive}, the adversary would not be able to stuck the honest parties during the $\tobc$ and $\trans$ phases.
Even if in the worst cases, the two phases do not deliver any useful output and the adversary intends to prevent the parties from running the
$\Pessimistic$ phase (thus not eliminating any transactions from the honest parties' input buffer),
we still have a  timeout mechanism against censorship, which can ensure to execute the $\Pessimistic$ phase for every $\bigO(T)$ asynchronous rounds if there is a suspiciously censored $\tx$ in all honest parties' buffers due to some timeout parameter $T$. 
Recall Lemma \ref{dumbolive}, for each $\tx$ in all honest parties' buffers, it would take $\bigO(XT/B)$ asynchronous rounds at worst (i.e., we always rely on the timeout $T$ to invoke the $\Pessimistic$ phase) to make   $\tx$ be one of the top $B$ transactions in all parties' buffers, where $X$ is the bound of buffer size (e.g., an unfixed polynomial in $\lambda$).
After that, any luckily finalized optimistic phase block would output $\tx$ (in few more $\delta$), or still relying on the timeout to invoke the $\Pessimistic$ phase, causing the worst case latency $\bigO( (X/B + \lambda)\delta T)$, which is a function in the actual network delay $\delta$ factored by some  (unfixed) polynomial of security parameters. 
$\hfill\square$ 


\section{Optimistic conditions of fastlane}\label{app:optimistic}

Here we discuss that the fastlane can successfully execute without invoking pace-synchronization, if the following two optimistic conditions hold:
(i) the network stays in synchrony, such that the guessed timeout parameter is larger than the ``heartbeat'' period of the underlying fastlane;
(ii)  the optimistic liveness condition of fastlane is satisfied, e.g. the leader is honest.

First, it is clear to see: if all honest parties have already enter   the  same epoch's fastlane at the same time,  then the fastlane must successfully progress
in the presence of   above  optimistic conditions.
Actually, the above argument   still holds,  even if the honest parties enter the fastlane with minor difference in time, because one can slightly tune  up the guessed timing parameter.

Then, let us briefly argue that when the network is synchronous, 
$\BDT$ can ensure   all honest parties to  enter the same fastlane within a bounded period.
This actually reduces to the next question: when some honest party first outputs and halts in the asynchronous pessimistic path, would all honest parties output soon (if the network is synchronous)?
Fortunately, the answer is yes if we check the detailed construction of asynchronous protocols (such as $\dumbo$). This indicates that the asynchronous pessimistic path itself can work as a clock  synchronizer to ensure that all honest parties   restart the fastlane nearly at the same time (when network synchrony holds).


\section{More Experiments in WAN setting}
\label{app:data}

For sake of completeness, we also measure (i)    latency and throughput on varying batch sizes and (ii) the      latency-throughput trade-off (with $n/3$ faults) in the WAN experiment setting, and plot the results as follows. 

\smallskip
\noindent
{\bf Varying batch sizes}.
For understanding to what an extent the batch size   matters,
we   report how   throughput and latency depend on varying   batch sizes   in Figure \ref{fig:batch_latency} and 	\ref{fig:batch_tps}, respectively.

\begin{figure}[h] 
	\vspace{-0.3cm}
	\captionsetup{font={normalsize}}
	\begin{center}
		\includegraphics[width=7.5cm]{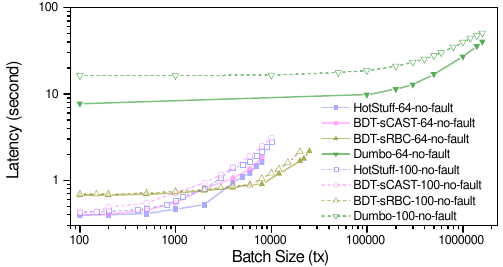}\\
		\vspace{-0.3cm}
		\caption{Latency v.s. batch size for experiments  over wide-area network when $n=64$ and $n=100$, respectively.}	
		\label{fig:batch_latency}
	\end{center}
	\vspace{-0.3cm}
\end{figure}

Figure \ref{fig:batch_latency} illustrates how latency increases with larger batch size in  $\mathsf{BDT}$-$\mathsf{sCAST}$, $\mathsf{BDT}$-$\mathsf{sRBC}$, $\mathsf{HotStuff}$ and $\dumbo$ when $n=64$ and $n=100$, respectively. It clearly states that: 
Dumbo always takes a latency much larger than $\mathsf{BDT}$-$\mathsf{sCAST}$, $\mathsf{BDT}$-$\mathsf{sRBC}$ and $\mathsf{HotStuff}$;
for $\mathsf{BDT}$-$\mathsf{sRBC}$, its latency increases much slower than $\mathsf{BDT}$-$\mathsf{sCAST}$ and $\mathsf{HotStuff}$, in particular when $B=10000$, the latency of $\mathsf{BDT}$-$\mathsf{sRBC}$ is around one second only, while these of $\mathsf{BDT}$-$\mathsf{sCAST}$ and $\mathsf{HotStuff}$ have been more than 2 seconds. The slow increasing of $\mathsf{BDT}$-$\mathsf{sRBC}$'s latency is mainly because its better balanced network load pattern.

\begin{figure}[h] 
	\vspace{-0.3cm}
	\captionsetup{font={normalsize}}
	\begin{center}
		\includegraphics[width=7.5cm]{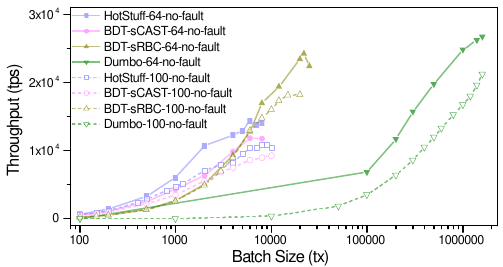}\\
		\vspace{-0.3cm}
		\caption{Throughput v.s. batch size for experiments over wide-area network when $n=64$ and $n=100$, respectively.}	
		\label{fig:batch_tps}
	\end{center}
	\vspace{-0.3cm}
\end{figure}

Figure \ref{fig:batch_tps} illustrates how throughput increases with larger batch size in  $\mathsf{BDT}$-$\mathsf{sCAST}$, $\mathsf{BDT}$-$\mathsf{sRBC}$, $\mathsf{HotStuff}$ and $\dumbo$ when $n=64$ and $n=100$, respectively.
$\dumbo$ really needs very large batch size to have   acceptable throughput; $\mathsf{BDT}$-$\mathsf{sCAST}$ and $\mathsf{HotStuff}$   have a similar trend that the throughput would stop to increase soon after the   batch sizes become larger (e.g., $10000$ transactions per block); in contrast,
the throughput of $\mathsf{BDT}$-$\mathsf{sRBC}$ is increasing faster than those of $\mathsf{BDT}$-$\mathsf{sCAST}$ and $\mathsf{HotStuff}$, 
because larger batch sizes in $\mathsf{BDT}$-$\mathsf{sRBC}$  would not place much worse bandwidth load on the leader,  and thus can raise more significant increment in the throughput. 

\smallskip
\noindent
{\bf Latency-throughput trade-off (1/3 crashes)}. 
We also report the latency-throughput trade-off  in the presence of $1/3$ crashes. The crashes   not only lag the execution of all protocols, but also mimic that a portion of $\mathsf{Bolt}$ instances are under denial-of-services. 
{We might  fix the  batch size of the pessimistic path  in $\mathsf{BDT}$-$\mathsf{sCAST}$ and $\mathsf{BDT}$-$\mathsf{sRBC}$ as $10^6$ transactions in these tests,
	because this batch size parameter brings reasonable throughput-latency trade-off in $\dumbo$.}
Shown in Figure \ref{fig:fault}, both $\mathsf{BDT}$-$\mathsf{sCAST}$ and $\mathsf{BDT}$-$\mathsf{sRBC}$ have some design spaces that show a  latency   better than $\dumbo$'s  and presents a throughput always better than $\mathsf{HotStuff}$'s, despite that on average $1/3$ instances of $\mathsf{Bolt}$ are unluckily stuck to wait for 2.5 sec to timeout without returning any optimistic output.
That means  our practical $\mathsf{BDT}$ framework does create new design space to harvest the best of both paths, resulting in that we can achieve reasonable throughput and latency simultaneously in fluctuating deployment environments.
\footnote{We would like to note that here we did a very pessimistic evaluation for BDT while optimistic evaluation for HotStuff in the sense that we manually trigger $\Trans$  by manually muting a leader for 2.5s once in 50 blocks, while for HotStuff we did a stable leader version (with honest leader). In reality, the performance curves for BDT might be a bit more to the left, while HotStuff will surely be more to the right/bottom.}

\begin{figure}[h] 
	\vspace{-0.3cm}
	\captionsetup{font={normalsize}}
	\begin{center}
		\includegraphics[width=7.5cm]{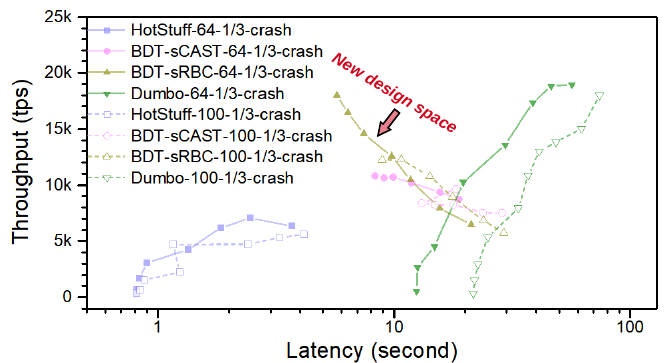}\\
		\vspace{-0.3cm}
		\caption{Throughput v.s. latency for experiments  over wide-area network  when $n=64$ and $n=100$, respectively (in case of 1/3 crash fault). We fix the fallback batch size of $\mathsf{BDT}$ instances to $10^6$ transactions in all tests.}	
		\label{fig:fault}
	\end{center}
	\vspace{-0.3cm}
\end{figure}


\section{Complexity and Numerical Analyses}\label{sec:complex}

This section discusses the critical complexity metrics of the
$\mathsf{BDT}$ framework and those of its major  modules. We
then assign  each   module a running time cost according to
our real-world experimental data, thus enabling more precise
numerical analysis to estimate the expected latency of $\mathsf{BDT}$
in various ``simulated'' unstable deployment environments.

\smallskip
\noindent
{\bf Complexity analysis}.
Here we analyze $\BDT$ regarding its complexities. 
Recall that we assume the batch size $B$ sufficiently large, e.g., $\Omega (\lambda n^2\log n)$, throughout the paper.

\smallskip
\noindent\underline{\smash{\em Complexities of the fastlane (also of the optismtic cases).}} 
For the optimistic fastlane, we have two instantiations, namely $\mathsf{Bolt}$-$\mathsf{sCAST}$ and $\mathsf{Bolt}$-$\mathsf{sRBC}$. As shown in Table \ref{tab:tobc},
$\mathsf{Bolt}$-$\mathsf{sCAST}$ is with linear $\mathcal{O}(n)$ per-block message complexity, and the leader's per-block bandwidth usage $\mathcal{O}(nB)$ is also linear in $n$;
$\mathsf{Bolt}$-$\mathsf{sRBC}$ is with quadratic per-block message complexity as $\mathcal{O}(n^2)$, while the per-block bandwidth usage of every party is not larger than the batch size $\mathcal{O}(B)$. We can also consider their latency in term of ``rounds'' to generate a block, i.e., the time elapsed between when a block's transaction is first multicasted and when the honest parties output this block with valid proof. The latency of generating two successive blocks can also be considered to reflect the confirmation latency of $\BDT$'s fastlane. \rev{Though both ``fastlane'' instantiations  will cost   $\mathcal{O}(1)$ rounds to generate  fastlane blocks, $\mathsf{Bolt}$-$\mathsf{sCAST}$ has slightly less concrete rounds: $\mathsf{Bolt}$-$\mathsf{sCAST}$ can use at most 3 rounds to generate one (pending) block and can use 5 rounds to output two successive blocks (thus the former block can be finalized in $\BDT$ framework); $\mathsf{Bolt}$-$\mathsf{sRBC}$ would cost 4 rounds to generate one (pending) block and use 8 rounds to output two   successive blocks.}

Note that in the {\em optimistic case} when (i) the fastlane leaders are always honest and (ii) the network condition is benign such that the fastlanes never timeout, the $\Pessimistic$ phase is not executed, 
so the fastlane cost shown in Table \ref{tab:tobc} would also reflect the amortized complexities of the overall $\mathsf{BDT}$ protocol (in case that the epoch size $Esize$ is   large enough, e.g., $n$).

\begin{table}[h]
		\vspace{-0.2cm}
	\captionsetup{font={normalsize}}
	\caption{Per-block performance    of different $\Bolt$ instantiations (which is also per-block cost of  {$\BDT$}  in the good cases)}
	\vspace{-0.3cm}
	\label{tab:tobc}
	\centering
	\begin{footnotesize}
		\begin{tabular}{c|c|c|c|c|c|c}
			 \hline\rule{0pt}{10pt}
			\multirow{2}{*}{} & \multirow{2}{*}{Msg.} & \multirow{2}{*}{Comm.} & \multirowcell{2}{Per-block\\ latency}& \multirowcell{2}{Two blocks\\ latency} & \multicolumn{2}{c}{Bandwidth Cost} \\ \cline{6-7} \rule{0pt}{10pt}
			&                       &                &    &    & Leader               & Others              \\ \hline\rule{0pt}{10pt}
			$\mathsf{Bolt}$-$\mathsf{sCAST}$               & $\mathcal{O}(n)$                   & $\mathcal{O}(nB)$  &   \rev{$3$ rounds} & \rev{$5$ rounds}            &  $\mathcal{O}(nB)$                 &  $\mathcal{O}(B)$                \\ \hline\rule{0pt}{10pt}
			$\mathsf{Bolt}$-$\mathsf{sRBC}$                   &$\mathcal{O}(n^2)$                   & $\mathcal{O}(nB)$      &   \rev{$4$ rounds}   & \rev{$8$ rounds}      &$\mathcal{O}(B)$             & $\mathcal{O}(B)$              \\ \hline 
		\end{tabular}
	\end{footnotesize}
		\vspace{-0.2cm}
\end{table}

\smallskip
\noindent\underline{\smash{\em Complexities of  the  worst cases (disregarding the adversary)}}. 
In the optimistic fastlane, there is a worst-case overhead of using  $\bigO(\tau)$ asynchronous rounds to leave the tentatively optimistic execution without outputting any valid blocks.
After the stop of fastlane, all parties enter the $\trans$ phase, and would participate in $\tcvba$, in which the expected message complexity is  $\mathcal{O}(n^2)$,  the expected communication complexity is $\mathcal{O}(\lambda n^2)$, and the expected bandwidth cost of each parties is $\mathcal{O}(\lambda n)$. Besides, if the output value of $\tcvba$ is large than zero, then the $\gethelp$ subroutine could probably be invoked, this process will incur $\mathcal{O}(n^2)$ overall message complexity and $\mathcal{O}(nB)$ per-block communication complexity and causes each party to spend $\mathcal{O}(B)$ bandwidth to fetch each block  on average.  
In the worst case, $\Dumbo$ is executed after $\trans$,   which on average costs  overall $\mathcal{O}(n^3)$   messages,\footnote{Note that if instantiating the pessimistic path by more recent asynchronous BFT consensus protocols (e.g., Speeding Dumbo) instead of Dumbo-BFT, the $\mathcal{O}(n^3)$  per-block messages can be reduced to $\mathcal{O}(n^2)$.} overall $\mathcal{O}(nB)$    communicated bits,  and $\mathcal{O}(B)$ bandwidth per party for each block if batch size $B$ is sufficiently large. The latency of generating a block in the pessimistic path is of $\bigO(1)$ rounds on average. 

\begin{table}[h]
	\vspace{-0.2cm}
	\captionsetup{font={normalsize}}
	\caption{Per-block performance of $\mathsf{BDT}$ in the worst cases}
	\vspace{-0.3cm}
	\label{tab:bolt}
	\centering
	\begin{footnotesize}
		\begin{tabular}{c|c|c|c|c|c}
			 \hline\rule{0pt}{10pt}
			\multirow{2}{*}{} & \multirow{2}{*}{Msg.} & \multirow{2}{*}{Comm.} & \multirow{2}{*}{Block latency (rounds)} & \multicolumn{2}{c}{Bandwidth Cost} \\ \cline{5-6} \rule{0pt}{10pt}
			&                       &                 &       & Leader               & Others              \\ \hline\rule{0pt}{10pt}
			$\mathsf{BDT}$-$\mathsf{sCAST}$           & $\mathcal{O}(n^3)$                   & $\mathcal{O}(nB)$        &    \rev{3+$1$+$T_{\tcvba}$+$T_{\mathsf{Dumbo}}$}      &  $\mathcal{O}(nB)$                 &  $\mathcal{O}(B)$                \\ \hline\rule{0pt}{10pt}
			$\mathsf{BDT}$-$\mathsf{sRBC}$                 &$\mathcal{O}(n^3)$                   & $\mathcal{O}(nB)$      &     \rev{4+$1$+$T_{\tcvba}$+$T_{\mathsf{Dumbo}}$}     & $\mathcal{O}(B)$               & $\mathcal{O}(B)$                \\ \hline 
		\end{tabular}
	\end{footnotesize}
	\\{\ }
	{\footnotesize
		\begin{itemize}[leftmargin=0.8cm]
			\item[$^*$]			 Note that the worst-case block latency reflects the case of turning off the level-1 fallback. 
		\end{itemize}

	}
	\vspace{-0.3cm}
\end{table}

To summarize these, we can have the worst-case   performance illustrated in Table \ref{tab:bolt}.
Note that 
the     latency of generating a block shall consider the following possible worst case: 
the fastlane  times out to run pace-sync, but pace-sync   finalizes no fastlane block, and all parties have to start the pessimistic path to generate a block.
\rev{Thus, to count the worst-case latency,
we need to include:
(i) the timeout parameter $\tau$; (ii) the latency of fallback (including 1 round for multicast $\view$  message and the expected latency $T_{\tcvba}$ of ${\tcvba}$), and (iii) the expected pessimistic path latency $T_{\mathsf{Dumbo}}$. 
Here the timeout parameter $\tau$ in our system is not necessarily close to the network delay upper bound $\Delta$, 
and it can represent  some adversary-controlling ``clock ticks''  to approximate   the number of asynchronous rounds spent to generate each fastlane block, i.e., $\bigO(\tau) = \bigO(1)$. 
For example, in $\mathsf{BDT}$-$\mathsf{sCAST}$, $\tau$ can approximate   3  rounds because in $\mathsf{Bolt}$-$\mathsf{sCAST}$, the first fastlane block (i.e., the first ``heartbeat'') needs 3 rounds to deliver  and the interval of two successive fastlane blocks (i.e., the interval of two ``heartbeats'')  is 2 rounds; similarly, $\tau$ can approximate 4 rounds in $\mathsf{BDT}$-$\mathsf{sRBC}$.}



\smallskip
\noindent\underline{\smash{\em Complexities in comparision to other BFT consensuses.}} 
Here we also summarize the communication complexities of $\mathsf{BDT}$ and some known BFT protocols in the optimistic case and the worst case, respectively.
To quantify   the latency of those protocols in unstable network environment, 
Table \ref{tab:performance_comparison} also  lists each protocol's  average latency (in ``unit'' of fastlane's good-case latency). 

This metric considers that the fastlane   has a probability $\alpha\in [0,1]$ to output blocks in time, 
and also has a chance of $\beta = 1-\alpha$ that  falls back and then executes the pessimistic asynchronous protocol. 
Let  $C$ be the latency of using earlier asynchronous protocols \cite{cachin01,honeybadger} directly as the pessimistic path
and $c$ be that of the state-of-the-art  $\dumbo$ BFT \cite{guo2020dumbo} and that of using MVBA  for  synchronization during fallback.
Both $C$ and $c$ are  represented in the unit of fastlane latency. According to the experimental data \cite{honeybadger,guo2020dumbo}, $C$ is normally at hundreds and the latter $c$ is typically dozens (\cite{honeybadger} runs $n$ instances of $\ABA$, thus rounds depend on number of parties, while \cite{guo2020dumbo} reduces it to constant).   
Our fallback  is almost as fast as the fastlane, so its magnitude around one.
As such, we can do a simple calculation as shown in Table  \ref{tab:performance_comparison} to    {\em roughly estimate}   the latency of all those protocols deployed in the realistic fluctuating network.

\begin{table}[h]
	\vspace{-0.2cm}
	\captionsetup{font={normalsize}}
	\begin{center}
		\centering
		
		\begin{footnotesize}
			\newcommand\xrowht[2][0]{\addstackgap[.5\dimexpr#2\relax]{\vphantom{#1}}}
			\caption{Complexities of BFT protocols in various settings
				(where $B$ is sufficiently large s.t. all $\lambda$ terms are omitted, and $\alpha + \beta =1$)}  
			\vspace{-0.3cm}
			\label{tab:performance_comparison}  
			\resizebox{0.49\textwidth}{!}{%
				\begin{tabular}{c|c|c|c}  
					\hline 	    \rule{0pt}{10pt}
					\multirow{2}{*}{Protocol}&\multicolumn{2}{c|}{Per-block Com. Compl.}&\multicolumn{1}{c}{Normalized average latency}\cr
					\cline{2-3}  \rule{0pt}{10pt}
					& Optim. &Worst& considering fastlane latency as ``unit''\cr  
					
					\hline  \xrowht{8pt}
					PBFT \cite{pbft}  &$\mathcal{O}(nB)$&$\infty$&  $1/\alpha$ \cr
					\hline \rule{0pt}{8pt}
					HotStuff \cite{yin2018hotstuff-full} &$\mathcal{O}(nB)$&$\infty$&$1/\alpha$ \cr
					\hline \rule{0pt}{8pt}
					HBBFT \cite{honeybadger} &$\mathcal{O}(nB)$&$\mathcal{O}(nB)$&  $C$ \cr
					\hline \rule{0pt}{8pt}
					Dumbo \cite{guo2020dumbo} &$\mathcal{O}(nB)$&$\mathcal{O}(nB)$& $c$ \cr
					\hline \xrowht{9pt}
					KS02 \cite{kursawe05}&$\mathcal{O}(n^2B)$&$\mathcal{O}(n^3B)$& $(\alpha+\frac{\beta}{C+kc})^{-1}$  $^\star$\cr
					\hline \xrowht{9pt}
					RC05 \cite{cachin05} &$\mathcal{O}(nB)$&$\mathcal{O}(n^3B)$& $(\alpha+\frac{\beta}{C+c})^{-1}$\cr
					
					\hline \xrowht{9pt}
					BDT  (ours) &$\mathcal{O}(nB)$&$\mathcal{O}(nB)$& $(\alpha+\frac{\beta}{c+1})^{-1}$\cr				
					\hline   
				\end{tabular}  
			}
		\end{footnotesize}
	\end{center}
	{
		\footnotesize	
		\begin{itemize}[leftmargin=0.4cm]
			\item[$^\star$]		 There is an integer parameter $k$ in \cite{kursawe05} to specify the degree of parallelism for the fastlane, thus probably incurring extra cost of fallback.
		\end{itemize}	

	}
	\vspace{-0.3cm}
\end{table}

\smallskip
\noindent
{\bf Numerical analysis on latency in unstable network}.
To understand the applicability level of $\BDT$ framework, we further conduct more precise numerical estimations  to visualize the average latency of $\mathsf{BDT}$ and prior art (e.g. RC05) in   the  unstable Internet deployment environment, in particular for some typical scenarios between the best   and the worst cases.

\begin{figure}[h] 
	\vspace{-0.1cm}
	\captionsetup{font={normalsize}}
	\begin{center}
		\includegraphics[width=7.8cm]{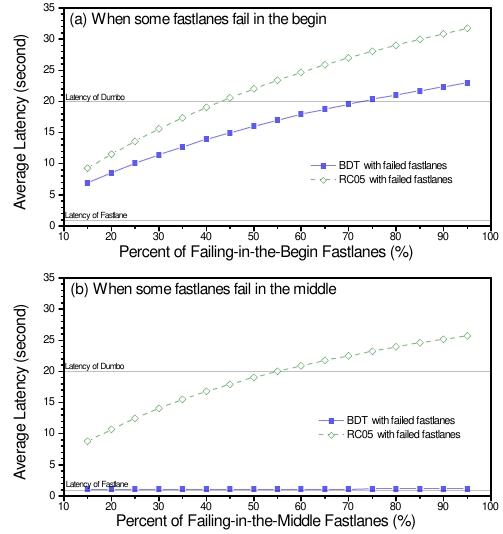}\\
		\vspace{-0.2cm}
		\caption{Numerical analysis to reflect the average latency of BDT and RC05 \cite{cachin05} in fluctuating deployment environment. The analysis methodology   is similar to the formulas in Table  \ref{tab:performance_comparison} except that here consider more protocol parameters such as batch size, epoch size, timeout, etc.
		}	
		\label{fig:numerical}
		\vspace{-0.3cm}
	\end{center}
\end{figure}

The real-world experiment data shown in Section \ref{evaluation} is considered to specify the cost of each protocol module in the estimations. In particular,
we use our experimental results over the globe when $n=100$ to specify the parameters used in the numerical estimations
regarding both RC05 and $\mathsf{BDT}$: we set the latency of fastlane as 1 second (to reflect the actual latency of $\mathsf{Bolt}$) and set the latency of pessimistic path  as 20 seconds (according to the measured latency of $\dumbo$); the fastlane block and the pessimistic   block are set to contain  $10^4$ and  $10^6$ transactions, respectively;  
the MVBA fallback in RC05 is set to use 10 seconds and our $\trans$ is set to cost 1 second (cf. Figure \ref{fig:overhead}); for fair comparison, we let RC05 to use the state-of-the-art $\dumbo$ protocol as its pessimistic path;
other protocols parameters (e.g., epoch size and timeout) are also taken into the consideration and are set as same as those  in the experiments. Note that a  ``second'' in the simulations is  a measurement of virtual time (normalized by the fastlane latency)  rather than a second in the real world.

We   consider two simulated scenarios. 
One is illustrated in Figure \ref{fig:numerical} (a), in which there are some portion of fastlane instances that completely fail and output nothing but just timeout and fallback after idling for 2.5 seconds, while the else fastlane instances successfully output  all optimistic blocks. In the case,  $\mathsf{BDT}$ can save up to almost 10 seconds on average latency relative to RC05. This is a result of the much more efficient fallback mechanism;
more importantly, the efficient fallback brings much robuster performance against unstable network environment, for example, RC05 starts to perform worse than $\dumbo$ once more than 45\% fastlane instances completely fail in the beginning of their executions, while $\mathsf{BDT}$ can be faster than $\dumbo$ until more than 75\% fastlane instances completely fail.
The other   case is shown in Figure \ref{fig:numerical} (b), where some fastlane instances stop to progress in the middle of their executions (e.g., stop to progress after 25 optimistic blocks are finalized) and then wait for 2.5 seconds to timeout and fallback. In the case, $\mathsf{BDT}$ performs {\em almost as fast as its underlying fastlane} (i.e., the average delay is really close to 1 second) despite the overheads of   timeout and $\trans$ in this fluctuating network condition; in contrast, RC05 can be an order of magnitude slower than $\mathsf{BDT}$, and it would be even slower than $\dumbo$ if more than 55\% fastlane instances fail to progress in the middle of their optimistic executions.

\end{document}